\documentclass[]{scrartcl}

\usepackage[utf8]{inputenc}
\usepackage[english]{babel}
\usepackage{amsmath}
\usepackage{amsthm}
\usepackage{thmtools}
\usepackage{amssymb}
\usepackage{float}
\usepackage{tikz}
\usetikzlibrary{shapes.geometric}
\usepackage[normalem]{ulem}
\usepackage{enumitem}
\usepackage{xparse}
\usepackage{hyperref}
\usepackage[perpage]{footmisc}
\usepackage{microtype}
\usepackage{stmaryrd}

\hypersetup{
	colorlinks,
	citecolor=black,
	linkcolor=black
}

\newcommand{\uem}[1]{\uline{#1}} 
\newcommand{\conc}{\cdot}
\newcommand{\comp}{\circ}
\newcommand{\class}[1]{\emph{#1}}
\renewcommand{\qedsymbol}{$\square$}

\NewDocumentCommand\Surface{ss}
{\text{\IfBooleanT{#1}{-}%
\IfBooleanTF{#2}{surface}{Surface}}%
}
\NewDocumentCommand\Target{ss}
{\text{\IfBooleanT{#1}{-}%
\IfBooleanTF{#2}{target}{Target}}%
}

\newcommand{\delbu}[1][\delta]{\widehat{#1}}
\newcommand{\deltd}[1][\delta]{\widetilde{#1}}
\newcommand{\BT}{\emph{BT}}
\newcommand{\X}{X}
\newcommand{\fyield}{\text{yield}}
\newcommand{\fheight}{\text{height}}
\newcommand{\fsubtree}{\text{sub}}
\newcommand{\fpath}{\text{path}}
\newcommand{\td}{\text{td}}
\newcommand{\bu}{\text{bu}}
\newcommand{\rc}{\text{rc}}
\newcommand{\tc}{\text{tc}}
\newcommand{\iterated}{\text{it}}
\newcommand{\dom}{\text{dom}}
\newcommand{\mbar}[1]{\overline{#1}}

\newcommand{\RECOG}{\class{RECOG}}
\newcommand{\CFL}{\class{CFL}}
\newcommand{\AFL}{\class{AFL}}
\newcommand{\B}{\class{B}}
\newcommand{\LB}{\class{LB}}
\newcommand{\DB}{\class{DB}}
\newcommand{\LDB}{\class{LDB}}
\newcommand{\NDB}{\class{NDB}}
\newcommand{\PNB}{\class{PNB}}
\newcommand{\PNLB}{\class{PNLB}}
\newcommand{\NLDB}{\class{NLDB}}
\newcommand{\NLB}{\class{NLB}}
\newcommand{\PDTB}{\class{PD$_t$B}}
\newcommand{\PLDTB}{\class{PLD$_t$B}}
\newcommand{\LDTB}{\class{LD$_t$B}}
\newcommand{\PLT}{\class{PLT}}
\newcommand{\PLB}{\class{PLB}}
\newcommand{\PNT}{\class{PNT}}
\newcommand{\PNLT}{\class{PNLT}}
\newcommand{\PDTT}{\class{PD$_t$T}}
\newcommand{\PLDTT}{\class{PLD$_t$T}}
\newcommand{\REL}{\class{REL}}
\newcommand{\QREL}{\class{QREL}}
\newcommand{\DBQREL}{\class{DBQREL}}
\newcommand{\DTQREL}{\class{DTQREL}}
\newcommand{\HOM}{\class{HOM}}
\newcommand{\LHOM}{\class{LHOM}}
\newcommand{\FTA}{\class{FTA}}
\newcommand{\T}{\class{T}}
\newcommand{\ZT}{\class{ZT}}
\newcommand{\LT}{\class{LT}}
\newcommand{\DT}{\class{DT}}
\newcommand{\LDT}{\class{LDT}}
\newcommand{\NLT}{\class{NLT}}
\newcommand{\TR}{\class{T$\,'$}}
\newcommand{\DTR}{\class{DT$\,'$}}
\newcommand{\LTR}{\class{LT$\,'$}}
\newcommand{\LDTR}{\class{LDT$\,'$}}
\newcommand{\ZTR}{\class{ZT$\,'$}}
\newcommand{\Z}{\class{Z}}
\newcommand{\LIN}{\class{L}}
\newcommand{\D}{\class{D}}
\newcommand{\LD}{\class{LD}}
\newcommand{\IO}{\class{IO}}
\newcommand{\Fclass}{\class{F}}
\newcommand{\Xclass}{\class{F}}

\declaretheorem[style=definition,name=Definition,numberwithin=section,qed=\qedsymbol]{definition}
\declaretheorem[style=definition,name=Example,qed=\qedsymbol,sibling=definition]{example}
\declaretheorem[style=definition,name=Example,sibling=definition]{examplewithoutqed}
\declaretheorem[style=definition,name=Theorem,qed=,sibling=definition]{theorem}
\declaretheorem[style=definition,name=Corollary,sibling=definition]{corollary}
\declaretheorem[style=definition,name=Corollary,qed=\qedsymbol,sibling=definition]{corollarywithqed}
\declaretheorem[style=definition,name=Lemma,qed=\qedsymbol,sibling=definition]{lemma}
\declaretheorem[style=definition,name=Lemma,sibling=definition]{lemmawithoutqed}
\declaretheorem[style=definition,name=Examples,qed=\qedsymbol,sibling=definition]{examples}
\declaretheorem[style=definition,name=Principle,qed=\qedsymbol,sibling=definition]{principle}
\declaretheorem[style=definition,name=Exercise,qed=\qedsymbol,sibling=definition]{exercise}
\declaretheorem[style=definition,name=Remark,qed=\qedsymbol,sibling=definition]{remark}
\declaretheorem[style=definition,name=Remarks,qed=\qedsymbol,sibling=definition]{remarks}
\declaretheorem[style=definition,name=Notation,qed=\qedsymbol,sibling=definition]{notation}
\declaretheorem[style=definition,name=Lemma,numbered=no]{lemma*}

\declaretheorem[style=definition,name=Idea,qed=\qedsymbol,sibling=definition]{idea}

\tikzset{lighttree/.style={level distance=0.6cm, sibling distance=0.7cm,
 minimum size=20pt, inner sep=1pt, outer sep=-10pt, shorten >=6pt, shorten <=6pt, execute at begin node=$, execute at end node=$},
tree/.style={baseline=-3pt, lighttree}}

\makeatletter
\newcommand{\xRightarrow}[2][]{\ext@arrow 0359\Rightarrowfill@{#1}{#2}}

\newcommand*{\rom}[1]{\expandafter\@slowromancap\romannumeral #1@}

\newcommand{\tosymbol}[1]{
\ifcase#1
$0$\or
$\dagger$\or
$\dagger\dagger$\or
$\dagger\dagger\dagger$\fi
}
\makeatother

\NewDocumentCommand\listt{sO{a}O{t_i}O{k}}
{#2[\replace{#3}{i}{1}
		\IfBooleanT{#1}{\replace{#3}{i}{2}}
		\cdots \replace{#3}{i}{#4} ] }  

\NewDocumentCommand\li{sO{t_i}O{k}O{,}O{\dots}}
{\replacei{#2}{1}#4
		\IfBooleanT{#1}{\replacei{#2}{2}#4}
		#5#4 \replacei{#2}{#3}}

\newcommand{\replacei}[2]{\replace{#1}{_i}{_{#2}}}

\ExplSyntaxOn
\NewDocumentCommand{\replace}{mmm}
 {
  \marian_replace:nnn {#1} {#2} {#3}
 }

\tl_new:N \l_marian_input_text_tl

\cs_new_protected:Npn \marian_replace:nnn #1 #2 #3
 {
  \tl_set:Nn \l_marian_input_text_tl { #1 }
  \tl_replace_all:Nnn \l_marian_input_text_tl { #2 } { #3 }
  \tl_use:N \l_marian_input_text_tl
 }
\ExplSyntaxOff

\title{TREE AUTOMATA AND\\ TREE GRAMMARS}
\author{}
\date{}

\begin{document}
	\pagenumbering{gobble}
	\maketitle 
	\vfill
	\vfill
~\\
	~\\
	by\\
	Joost Engelfriet	
	\vfill
	\vfill
	~\\
	DAIMI FN-10\\
	April 1975
	\vfill
	~\\
	Institute of Mathematics, University of Aarhus\\
	DEPARTMENT OF COMPUTER SCIENCE\\
	Ny Munkegade, 8000 Aarhus C, Denmark\\
	
\newpage
\section*{Preface}

I wrote these lecture notes during my stay in Aarhus in the academic year 1974/75.
As a young researcher I had a wonderful time at DAIMI,
and I have always been happy to have had that early experience. 

I wish to thank Heiko Vogler for his noble plan to move these notes into the digital world,
and I am grateful to Florian Starke and Markus Napierkowski (and Heiko) for the 
excellent transformation of my hand-written manuscript into \LaTeX. 
 
Apart from the reparation of errors and some cosmetical changes, 
the text of the lecture notes has not been changed. 
Of course, many things have happened in tree language theory since 1975. 
In particular, most of the problems mentioned in these notes have been solved. 
The developments until 1984 are described in the book ``Tree Automata'' 
by Ferenc G\'ecseg and Magnus Steinby, and for recent developments I recommend the Appendix
of the reissue of that book at arXiv.org/abs/1509.06233. 

\bigskip
\hfill Joost Engelfriet, October 2015

\bigskip
\hfill LIACS, Leiden University,  

\hfill The Netherlands

\newpage
\section*{Tree automata and tree grammars}
To appreciate the theory of tree automata and tree grammars one should already be motivated by the goals and results of formal language theory. In particular one should be interested in ``derivation trees". A derivation tree models the grammatical structure of a sentence in a (context-free) language. By considering only the bottom of the tree the sentence may be recovered from the tree.

The first idea in tree language theory is to generalize the notion of a finite automaton working on strings to that of a finite automaton operating on trees. It turns out that a large part of the theory of regular languages can rather easily be generalized to a theory of regular tree languages. Moreover, since a regular tree language is (almost) the same as the set of derivation trees of some context-free language, 
one obtains results about context-free languages by ``taking the bottom" of results about regular tree languages.

The second idea in tree language theory is to generalize the notion of a generalized sequential machine (that is, finite automaton with output) to that of a finite state tree transducer. Tree transducers are more complicated than string transducers since they are equipped with the basic capabilities of copying, deleting and reordering (of subtrees). The part of (tree) language theory that is concerned with translation of languages is mainly motivated by compiler writing (and, to a lesser extent, by natural linguistics). When considering bottoms of trees, finite state transducers are essentially the same as syntax-directed translation schemes. Results in this part of tree language theory treat the composition and decomposition of tree transformations, and the properties of those tree languages that can be obtained by finite state transformation of regular tree languages (or, taking bottoms, those languages that can be obtained by syntax-directed translation of context-free languages).

Thirdly there are, of course, many other ideas in tree language theory. In the literature one can find, for instance, context-free tree grammars, recognition of subsets of arbitrary algebras, tree walking automata, hierarchies of tree languages (obtained by iterating old ideas), decomposition of tree automata, Lindenmayer tree grammars, etc.

These lectures will be divided in the following five parts: (1) and (2) contain preliminaries, (3), (4) and (5) are the main parts.
\begin{enumerate}[label=(\arabic*)]
	\item Introduction. (p.\,\pageref{sec.1})
	\item Some basic definitions. (p.\,\pageref{sec.2})
	\item Recognizable (= regular) tree languages. (p.\,\pageref{sec.3})
	\item Finite state tree transformations. (p.\,\pageref{sec.4})
	\item Whatever there is more to consider.
\end{enumerate}
Part (5) is not contained in these notes; instead, 
some Notes on the literature are given on p.\,\pageref{sec.5}.

\newpage
\tableofcontents
\newpage
	\binoppenalty=10000
	\relpenalty=10000
	\pagenumbering{arabic}
	\setcounter{page}{1}
	\begin{sloppypar}
	
\section{Introduction}\label{sec.1}
Our basic data type is the kind of tree used to express the grammatical structure of strings in a context-free language.
\begin{example}\label{exa.1.1}
	Consider the context-free grammar $G=(N,\Sigma,R,S)$ with nonterminals $N=\left\lbrace S,A,D \right\rbrace$, terminals $\Sigma=\left\lbrace a,b,d \right\rbrace$, initial nonterminal $S$ and the set of rules $R$, consisting of the rules $S\to AD$, $A\to aAb$, $A\to bAa$, $A\to AA$, $A\to\lambda$, $D\to Ddd$ and $D\to d$ (we use $\lambda$ to denote the empty string).

	The string $baabddd\in\Sigma^*$ can be generated by $G$ and has the following derivation tree (see \cite[\rom{2}.6]{Sal}, \cite[0.5 and 2.4.1]{AU}):
	\begin{figure}[h]
	\centering
	\begin{tikzpicture}[tree, level distance=0.9cm,
	lvl1/.style={sibling distance=4.5cm},
	lvl2/.style={sibling distance=2.3cm},
	lvl3/.style={sibling distance=0.7cm}]
		\node {S} [lvl1]
			child {node {A} [lvl2]
				child {node {A} [lvl3]
					child {node {b}}
					child {node {A}
						child {node {e}}}
					child {node {a}}}
				child {node {A} [lvl3]
					child {node {a}}
					child {node {A}
						child {node {e}}}
					child {node {b}}}
					}
			child {node {D} [lvl3]
				child {node {D}
					child {node {d}}}
				child {node {d}}
				child {node {d}}};
	\end{tikzpicture}
	\end{figure}

	\begin{itemize}
		\item Note that we use $e$ as a symbol standing for the empty string $\lambda$.
		\item The string $baabddd$ is called the ``yield" or ``result" of the derivation tree.\qedhere
	\end{itemize}
\end{example}
Thus, in graph terminology, our trees are finite (finite number of nodes and branches), directed (the branches are ``growing downwards"), rooted (there is a node, the root, with no branches entering it), ordered (the branches leaving a node are ordered from left to right) and labeled (the nodes are labeled with symbols from some alphabet).
						
The following intuitive terminology will be used:
\begin{itemize}
	\item the \uem{rank} (or out-degree) of a node is the number of branches leaving it (note that  the in-degree of a node is always 1, except for the root which has in-degree 0)
	\item a \uem{leaf} is a node with rank 0
	\item the \uem{top} of a tree is its root
	\item the \uem{bottom} (or frontier) of a tree is the set (or sequence) of its leaves
	\item the \uem{yield} (or result, or frontier) of a tree is the string obtained by writing the labels of its leaves (except the label $e$) from left to right
	\item a \uem{path} through a tree is a sequence of nodes connected by branches (``leading downwards"); the \uem{length} of the path is the number of its nodes minus one (that is, the number of its branches)
	\item the \uem{height} (or depth) of a tree is the length of the longest path from the top to the bottom
	\item if there is a path of $\text{length}\geq 1$ (of $\text{length}=1$) from a node $a$ to a node $b$ then $b$ is a \uem{descendant} (\uem{direct descendant}) of $a$ and $a$ is an \uem{ancestor} (\uem{direct ancestor}) of $b$ 
	\item a \uem{subtree} of a tree is a tree determined by a node together with all its descendants; a \uem{direct subtree} is a subtree determined by a direct descendant of the root of the tree; note that each tree is uniquely determined by the label of its root and the (possibly empty) sequence of its direct subtrees
	\item the phrases ``bottom-up", ``bottom-to-top" and ``frontier-to-root" are used to indicate this $\uparrow$ direction, while the phrases ``top-down", ``top-to-bottom" and ``root-to-frontier" are used to indicate that $\downarrow$ direction.
\end{itemize}
						
In derivation trees of context-free grammars each symbol may only label nodes of certain ranks. For instance, in the above example, $a$, $b$, $d$ and $e$ may only label leaves (nodes of rank 0), $A$ labels nodes with ranks 1, 2 and 3, $S$ labels nodes with rank 2, and $D$ nodes of rank 1 and 3 (these numbers being the lengths of the right hand sides of rules).
						
Therefore, given some alphabet, we require the specification of a finite number of ranks for each symbol in the alphabet, and we restrict attention to those trees in which nodes of rank $k$ are labeled by symbols of rank $k$.

\section{Some basic definitions}\label{sec.2}
The mathematical definition of a tree may be given in several, equivalent, ways. We will define a \uwave{tree as a special kind of string} (others call this string a representation of the tree, see \cite[0.5.7]{AU}). Before doing so, let us define ranked alphabets.
\begin{definition}\label{def.2.1}
	An alphabet $\Sigma$ is said to be \uem{ranked} if for each nonnegative integer $k$ a subset $\Sigma_k$ of $\Sigma$ is specified, such that $\Sigma_k$ is nonempty for a finite number of $k$'s only, and such that $\Sigma=\bigcup\limits_{k\geq0}\Sigma_k$. 
	\footnote{To be more precise one should define a ranked alphabet as a pair $(\Sigma,f)$, where $\Sigma$ is an alphabet and $f$ is a mapping from $\mathbb{N}$ into $\mathcal{P}(\Sigma)$ such that $\exists n~\forall k\geq n:f(k)=\emptyset$, and then denote $f(k)$ by $\Sigma_k$ and $(\Sigma,f)$ by $\Sigma$.
Note that $\mathbb{N}=\{0,1,2,\ldots\}$ is the set of natural numbers and that
$\mathcal{P}(\Sigma)$ is the set of subsets of $\Sigma$. }
		
	If $a\in\Sigma_k$, then we say that $a$ has rank $k$ (note that $a$ may have more than one rank).
		
	Usually we define a specific ranked alphabet $\Sigma$ by specifying those $\Sigma_k$ that are nonempty.
\end{definition}
\begin{example}\label{exa.2.2}
	The alphabet $\Sigma=\left\lbrace a,b,+,-\right\rbrace $ is made into a ranked alphabet by specifying $\Sigma_0=\left\lbrace a,b\right\rbrace $, $\Sigma_1=\left\lbrace-\right\rbrace $ and $\Sigma_2=\left\lbrace +,-\right\rbrace $. (Think of negation and subtraction).
\end{example}
\begin{remark}\label{rem.2.3}
	Throughout our discussions we shall use the symbol $e$ as a special symbol, intuitively representing $\lambda$. Whenever $e$ belongs to a ranked alphabet, it is of rank 0.
\end{remark}
Operations on ranked alphabets should be defined as for instance in the following definition.
\begin{definition}\label{def.2.4}
	Let $\Sigma$ and $\Delta$ be ranked alphabets. The \uem{union} of $\Sigma$ and $\Delta$, denoted by $\Sigma\cup\Delta$, is defined by $(\Sigma\cup\Delta)_k=\Sigma_k\cup\Delta_k$, for all $k\geq0$. We say that $\Sigma$ and $\Delta$ are \uem{equal}, denoted by $\Sigma=\Delta$, if, for all $k\geq0$, $\Sigma_k=\Delta_k$.
\end{definition}
We now define the notion of tree. Let ``$\,[\,$" and ``$\,]\,$" be two symbols which are never elements of a ranked alphabet. 
\begin{definition}\label{def.2.5}
	Given a ranked alphabet $\Sigma$, the set of \uem{trees over $\Sigma$}, denoted by $T_\Sigma$, is the language over the alphabet $\Sigma\cup\left\lbrace [\,, ] \right\rbrace $ defined inductively as follows.
	\begin{enumerate}[label=(\roman*)]
		\item If $a\in\Sigma_0$, then $a\in T_\Sigma$.
		\item For $k\geq1$, if $a\in\Sigma_k$ and $\li*\in T_\Sigma$, then $\listt*\in T_\Sigma$. \qedhere
	\end{enumerate}
\end{definition}
Intuitively, $a$ is a tree with one node labeled ``$a$", and $\listt*$ is the tree
\begin{figure}[h]
\centering
\begin{tikzpicture}[execute at begin node=$, execute at end node=$, baseline=-3pt, leaf/.style={isosceles triangle,draw,shape border rotate=90,,opacity=1}]
	\tikzstyle{level 1}=[level distance=0.5cm]
	\tikzstyle{level 2}=[level distance=0.96cm]
	\node[] {a}
	    child {node[scale=0.01] {} child[opacity=0] {node[leaf] {t_1}}}
	    child {node[scale=0.01] {} child[opacity=0] {node[leaf] {t_2}}}
	    child[level distance=1.4cm] {node[isosceles triangle,shape border rotate=90] {.\hspace{0.3cm}.\hspace{0.3cm}.} edge from parent[opacity=0] node[below,opacity=1] {\dots}}
	    child {node[scale=0.01] {} child[level distance=0.98cm,opacity=0] {node[leaf] {t_k}}};
\end{tikzpicture}.
\end{figure}

\begin{example}\label{exa.2.6}
	Consider the ranked alphabet of Example \ref{exa.2.2}. Then $+[-[a-[b]]a]$ is a tree over this alphabet, intuitively ``representing" the tree
	\begin{figure}[h]
	\centering
	\begin{tikzpicture}[tree]
	\node {+}
		child {node {-}
			child {node {a}}
			child {node {-}
				child {node {b}}}}
		child {node {a}};
	\end{tikzpicture}
	\end{figure}\\
	which on its turn ``represents" the expression $(a-(-b))+a$ (note that the ``{}official" tree is the prefix notation of this expression).
\end{example}
\begin{example}\label{exa.2.7}
	Consider the ranked alphabet $\Delta$, where $\Delta_0=\left\lbrace a,b,d,e\right\rbrace$, $\Delta_1=\Delta_3=$ $\left\lbrace A,D\right\rbrace$ and $\Delta_2=\left\lbrace A,S\right\rbrace$.
A picture of the tree 
$$S[A[A[bA[e]a]A[aA[e]b]]D[D[d]dd]]$$ 
in $T_\Delta$ is given in Example \ref{exa.1.1}.
\end{example}
\begin{exercise}\label{exe.2.8}
	Take some ranked alphabet $\Sigma$ and show that $T_\Sigma$ is a context-free language over $\Sigma\cup \left\lbrace [\,,] \right\rbrace $.
\end{exercise}
Our main aim will be to study  several ways of constructively representing sets of trees and relations between trees. The basic terminology is the following.
\begin{definition}\label{def.2.9}
	Let $\Sigma$ be a ranked alphabet. A \uem{tree language} over $\Sigma$ is any subset of~$T_\Sigma$.
\end{definition}
\begin{definition}\label{def.2.10}
	Let $\Sigma$ and $\Delta$ be ranked alphabets. A \uem{tree transformation} from $T_\Sigma$ into $T_\Delta$ is any subset of $T_\Sigma\times T_\Delta$.
\end{definition}
\begin{exercise}\label{exe.2.11}
	Show that the context-free grammar $G=(N,\Sigma,R,S)$ with $N=\left\lbrace S\right\rbrace $, $\Sigma=\left\lbrace a,b,[\,,]\right\rbrace$ and $R=\left\lbrace S\to  b[aS] ,\,S\to  a \right\rbrace $ generates a tree language over $\Delta$, where $\Delta_0=\left\lbrace a\right\rbrace $ and $\Delta_2=\left\lbrace b\right\rbrace $.
\end{exercise}
The above definition of ``tree" (Definition \ref{def.2.5}) gives rise to the following principles of proof by induction and definition by induction for trees. (Note that each tree is, uniquely, either in $\Sigma_0$ or of the form $\listt$).
\begin{principle}\label{pri.2.12}
	Principle of proof by induction (or recursion) on trees. Let $P$ be a property of trees (over $\Sigma$).
	\begingroup
	\setlength{\leftmargini}{14pt}
	\begin{description} 
        \item[If]
		\begin{enumerate}[label=(\roman*)]
			\item all elements of $\Sigma_0$ have property $P$, and
			\item for each $k\geq1$ and each $a\in\Sigma_k$, if $\li$ have property $P$, then $\listt$ has property $P$,
		\end{enumerate}
		\item[then] all trees in $T_\Sigma$ have property $P$.\qedhere
	\end{description}
	\endgroup 
\end{principle}
\begin{principle}\label{pri.2.13}
	Principle of definition by induction (or recursion) on trees. Suppose we want to associate a value $h(t)$ with each tree $t$ in $T_\Sigma$. Then it suffices to define $h(a)$ for all $a\in \Sigma_0$, and to show how to compute the value $h(\listt)$ from the values $\li[h(t_i)]$.
	More formally expressed,
	\begin{description}
		\item[given] a set $O$ of objects, and
		\begin{enumerate}[label=(\roman*)]
			\item for each $a\in\Sigma_0$, an object $o_a\in O$, and
			\item for each $k\geq1$ and each $a\in\Sigma_k$, a mapping $f^k_a:O^k\to O$,
		\end{enumerate}
		\item[there is] exactly one mapping $h:T_\Sigma\to O$ such that
		\begin{enumerate}[label=(\roman*)]
			\item $h(a)=o_a$ for all $a\in\Sigma_0$, and
			\item $h(\listt)=f^k_a(\li[h(t_i)])$ for all $k\geq1$, $a\in\Sigma_k$ and $\li\in T_\Sigma$.\qedhere
		\end{enumerate}
	\end{description}
\end{principle}
\begin{example}\label{exa.2.14}
	Let $\Sigma_0=\{e\}$ and $\Sigma_1=\{\,/\,\}$. The trees in $T_\Sigma$ are in an obvious one-to-one correspondence with the natural numbers. The above principles are the usual induction principles for these numbers.
\end{example}
To illustrate the use of the induction principles we give the following useful definitions.
\begin{definition}\label{def.2.15}
	The mapping \uem{\fyield} from $T_\Sigma$ into $\Sigma^*_0$ is defined inductively as follows.
	\begin{enumerate}[label=(\roman*)]
		\item For $a\in\Sigma_0$, $\fyield(a)=
		      \begin{cases}
		      	a       & \text{if }a\neq e \\
		      	\lambda & \text{if } a=e.
		      \end{cases}$
	    \item\label{def.2.15ii} For $a\in\Sigma_k$ and $\li\in T_\Sigma$,\\ $\fyield(\listt)=\fyield(t_1)\cdot\fyield(t_2)\cdots\fyield(t_k)$.\footnote{That is, the concatenation of $\li[\fyield(t_i)]$.}
	\end{enumerate}
	Moreover, for a tree language $L\subseteq T_\Sigma$, we define 
	\[\fyield(L)=\{\fyield(t)\mid t\in L\}.\]
	We shall sometimes abbreviate ``\fyield" by ``y".
\end{definition}
\begin{definition}\label{def.2.16}
	The mapping \uem{\fheight} from $T_\Sigma$ into $\mathbb{N}$ is defined recursively as follows.
	\begin{enumerate}[label=(\roman*)]
		\item For $a\in\Sigma_0$, $\fheight(a)=0$.
		      \item\label{def.2.16ii} For $a\in\Sigma_k$ and $\li\in T_\Sigma$,\\ $\fheight(\listt)=\max\limits_{1\leq i\leq k}(\fheight(t_i))+1$.\qedhere
	\end{enumerate}
\end{definition}
\begin{examplewithoutqed}\label{exa.2.17}
	As an example of a proof by induction on trees we show that, if $e\notin\Sigma_0$ and $\Sigma_1=\emptyset$, then, for all $t\in T_\Sigma$, $\fheight(t)<| \fyield(t)| $. 
\end{examplewithoutqed}
\begin{proof}
	For $a\in\Sigma_0$, $\fheight(a)=0$ and $| \fyield(a)| =| a| =1$ (since $a\neq e$). Now let $a\in\Sigma_k$ ($k\geq2$) and assume (induction hypothesis) that $\fheight(t_i)<| \fyield(t_i)| $ for $1\leq i\leq k$.\\
	\begin{tabular}{llr}
		Then & $| \fyield(\listt)| =\sum\limits_{i=1}^{k}| \fyield(t_i)| $ & (Def. \ref{def.2.15}\ref{def.2.15ii})    \\
		     & $\geq (\mathop{\sum\limits_{i=1}^{k}}\fheight(t_i))+k$      & (ind. hypothesis)                            \\
		     & $\geq( \max\limits_{1\leq i\leq k} \fheight(t_i))+2$        & ($k\geq2$ and $\fheight(t_i)\geq0$)                \\
		     & $> \fheight(\listt)$                                                   & (Def. \ref{def.2.16}\ref{def.2.16ii}). 
	\end{tabular}

\end{proof}

\begin{exercise}\label{exe.2.18}
Let $\Sigma$ be a ranked alphabet such that $\Sigma_0\cap\Sigma_k=\emptyset$ for all $k\geq1$.  
	Define a (string) homomorphism $h$ from $(\Sigma\cup\{[\,,]\})^*$ into $\Sigma_0^*$ such that, for all $t\in T_\Sigma$, $h(t)=\fyield(t)$.
\end{exercise}
\begin{exercise}\label{exe.2.19}
	Give a recursive definition of the notion of ``subtree", for instance as a mapping $\fsubtree : T_\Sigma\to\mathcal{P}(T_\Sigma)$ such that $\fsubtree(t)$ is the set of all subtrees of $t$. Give also an alternative definition of ``subtree" in a more string-like fashion.
\end{exercise}
\begin{exercise}\label{exe.2.20}
	Let $\fpath(t)$ denote the set of all paths from the top of $t$ to its bottom. Think of a formal definition for ``path".
\end{exercise}
The generalization of formal language theory to a formal tree language theory will come about by viewing a \uwave{string as a special kind of tree} and taking the obvious generalizations. To be able to view strings as trees we ``turn them 90 degrees" to a vertical position, as follows.
\begin{definition}\label{def.2.21}
	A ranked alphabet $\Sigma$ is \uem{monadic} if
	\begin{enumerate}[label=(\roman*)]
		\item $\Sigma_0=\{e\}$, and
		\item for $k\geq2$, $\Sigma_k=\emptyset$.
	\end{enumerate}
	The elements of $T_\Sigma$ are called monadic trees.
\end{definition}
Thus a monadic ranked alphabet $\Sigma$ is fully determined by the alphabet $\Sigma_1$. Monadic trees obviously can be made to correspond to the strings in $\Sigma_1^*$. There are two ways to do this, depending on whether we read top-down or bottom-up:

$f_{\td}:T_\Sigma\to \Sigma_1^*$ is defined by
\begin{enumerate}[leftmargin=2cm,label=(\roman*)]
	\item $f_{\td}(e)=\lambda$
	\item $f_{\td}(a[t])=a\cdot f_{\td}(t)$ for $a\in\Sigma_1$ and $t\in T_\Sigma$
\end{enumerate}
and $f_{\bu}:T_\Sigma\to \Sigma_1^*$ is defined by
\begin{enumerate}[leftmargin=2cm,label=(\roman*)]
	\item $f_{\bu}(e)=\lambda$
	\item $f_{\bu}(a[t])=f_{\bu}(t)\cdot a$ for $a\in\Sigma_1$ and $t\in T_\Sigma$.
\end{enumerate}
(Obviously both $f_{\td}$ and $f_{\bu}$ are bijections).\\
Accordingly, when generalizing a string-concept to trees, we often have the choice between a top-down and a bottom-up generalization.
\begin{example}\label{exa.2.22}
	The string alphabet $\Delta=\{a,b,c\}$ corresponds to the monadic alphabet $\Sigma$ with $\Sigma_0=\{e\}$ and $\Sigma_1=\Delta$. The tree

	\begin{center}
	\begin{tikzpicture}[tree]
	\node {a}
		child {node {b}
			child {node {c}
				child {node {b}
					child {node {e}}}}};
	\end{tikzpicture}
\end{center}

\noindent in $T_\Sigma$ corresponds either to the string \emph{abcb} in $\Delta^*$ (top-down), or to the string \emph{bcba} in $\Delta^*$ (bottom-up). Note that, due to our ``prefix definition" of trees (Definition \ref{def.2.5}), the above tree looks ``top-down like" in its official form $a[b[c[b[e]]]]$. Obviously this is not essential.
\end{example}
Let us consider some basic operations on trees. A basic operation on strings is right-concatenation with one symbol (that is, for each symbol $a$ in the alphabet there is an operation $\rc_a$ such that, for each string $w$, $\rc_a(w)=wa$). Every string can uniquely be built up from the empty string by these basic operations (consider the way you write and read!). Generalizing bottom-up, the corresponding basic operations on trees, here called ``top concatenation", are the following.
\begin{definition}\label{def.2.23}
	For each $a\in\Sigma_k$ ($k\geq1$) we define the ($k$-ary) operation of \uem{top~concatenation} with $a$, denoted by $\tc_a^k$, to be the mapping from $T_\Sigma^k$ into $T_\Sigma$ such that, for all $\li\in T_\Sigma$,
	\[\tc_a^k(\li)=\listt.\]
	Moreover, for tree languages $\li[L_i]$, we define\\
	\[\tc_a^k(\li[L_i])=\{\listt \mid t_i\in L_i\text{ for all }1\leq i\leq k\}.\]
\end{definition}
Note that every tree can uniquely be built up from the elements of $\Sigma_0$ by repeated top concatenation.

The next basic operation on strings is concatenation. When viewed monadically, concatenation corresponds to substituting one vertical string into the $e$ of the other vertical string. In the general case, we may take one tree and substitute a tree into each leaf of the original tree, such that different trees may be substituted into leaves with different labels. Thus we obtain the following basic operation on trees.
\begin{definition}\label{def.2.24}
	Let $n\geq1$, $\li[a_i][n]\in\Sigma_0$ all different, and $\li[s_i][n]\in T_\Sigma$. For $t\in T_\Sigma$, the \uem{tree concatenation} of $t$ with $\li[s_i][n]$ at $\li[a_i][n]$, denoted by $t\langle\li[a_i\gets s_i][n] \rangle$, is defined recursively as follows.
	\begin{enumerate}[label=(\roman*)]
		\item for $a\in\Sigma_0$,\\
		      $a\langle\li[a_i\gets s_i][n] \rangle=
		      \begin{cases}
		      	s_i & \text{if }a=a_i  \\
		      	a   & \text{otherwise} 
		      \end{cases}$
		\item for $a\in\Sigma_k$ and $\li\in T_\Sigma$,\\[1mm]				
		      $\listt\langle\dots\rangle=\listt[a][t_i\langle\dots\rangle]$,\\[1mm]
		      where $\langle\dots\rangle$ abbreviates $\langle\li[a_i\gets s_i][n]\rangle$.
	\end{enumerate}
	If, in particular, $n=1$, then, for each $a\in\Sigma_0$ and $t,s\in T_\Sigma$, the tree $t\langle a\gets s\rangle$ is also denoted by $t\conc_a s$.
\end{definition}
\begin{example}\label{exa.2.25}
	Let $\Delta_0=\{x,y,c\}$, $\Delta_2=\{b\}$ and $\Delta_3=\{a\}$.
	If $t= a[b[xy]xc] $, then $t\langle x\gets  b[cx] ,y\gets c \rangle= a[b[b[cx]c]b[cx]c] $.
\end{example}
\begin{exercise}\label{exe.2.26}
	Check that in the monadic case tree concatenation corresponds to string concatenation.
\end{exercise}
For tree languages tree concatenation is defined analogously.
\begin{definition}\label{def.2.27}
	Let $n\geq1$, $\li[a_i][n]\in\Sigma_0$ all different, and $\li[L_i][n]\subseteq T_\Sigma$. For $L\subseteq T_\Sigma$ we define the \uem{tree concatenation} of $L$ with $\li[L_i][n]$ at $\li[a_i][n]$, denoted by $L\langle\li[a_i\gets L_i][n]\rangle$, as follows.\footnote{As usual, given a string $w$, we use $w$ also to denote the language $\{w\}$.}
	\begin{enumerate}[label=(\roman*)]
		\item for $a\in\Sigma_0$,\\
		      $a\langle\li[a_i\gets L_i][n]\rangle=
		      \begin{cases}
		      	L_i & \text{if }a=a_i  \\
		      	a   & \text{otherwise} 
		      \end{cases}$
		\item for $a\in\Sigma_k$ and $\li\in T_\Sigma$,\\[1mm]
		      $\listt\langle\dots\rangle=\listt[a][t_i\langle\dots\rangle]$\;\;\footnote{For tree languages $\li[M_i]$ we also write $\listt[a][M_i]$ to denote $\tc_a^k(\li[M_i])$. This notation is fully justified since $\listt[a][M_i]$ is the (string) concatenation of the languages $a$, $[\,$, $\li[M_i]$ and $]\,$!}
		\item for $L\subseteq T_\Sigma$,\\
		      $L\langle\li[a_i\gets L_i][n]\rangle=\bigcup\limits_{t\in L} t\langle\li[a_i\gets L_i][n]\rangle$.
	\end{enumerate}
	If, in particular, $n=1$, then, for each $a\in\Sigma_0$ and each $L_1,L_2\subseteq T_\Sigma$, we denote $L_1\langle a\gets L_2\rangle$ also by $L_1\conc_a L_2$.
\end{definition}
\begin{remarks}\label{rem.2.28}~
	\begin{enumerate}[label=(\arabic*)]
		\item Obviously, if $L,L_1,\dots,L_n$ are singletons, then Definition \ref{def.2.27} is the same as Definition \ref{def.2.24}.
	    \item\label{rem:2.28ii} Note that tree concatenation, as defined above, is ``nondeterministic" in the sense that, for instance, to obtain $t\langle\li[a_i\gets L_i][n]\rangle$ different elements of $L_1$ may be substituted at different occurrences of $a_1$ in $t$.
	    ``Deterministic" tree concatenation of $t$ with $\li[L_i][n]$ at $\li[a_i][n]$ could be defined as
	    $\{t\langle\li[a_i\gets s_i][n]\rangle\mid s_i\in L_i\text{ for all }1\leq i\leq n\}$. In this case different occurrences of $a_1$ in $t$ should be replaced by the same element of $L_1$. It is clear that, in the case that $\li[L_i][n]$ are singletons, this distinction cannot be made.\qedhere
	\end{enumerate}
\end{remarks}
Intuitively, since trees are strings, tree concatenation is nothing else but ordinary string substitution, familiar from formal language theory (see, for instance, \cite[\rom{1}.3]{Sal}).

For completeness we give the definition of substitution of string languages.
\begin{definition}\label{def.2.29}
	Let $\Delta$ be an alphabet. Let $n\geq1$, $\li[a_i][n]\in\Delta$ all different and let $\li[L_i][n]$ be languages over $\Delta$. For any $L\subseteq\Delta^*$, the \uem{substitution} of $\li[L_i][n]$ for $\li[a_i][n]$ in $L$, denoted by $L\langle\li[a_i\gets L_i][n]\rangle$, is the language over $\Delta$ defined as follows:
	\begin{enumerate}[label=(\roman*)]
		\item $\lambda\langle\li[a_i\gets L_i][n]\rangle=\lambda$
		\item for $a\in\Delta$,\\
		      $a\langle\li[a_i\gets L_i][n]\rangle=
		      \begin{cases}
		      	L_i & \text{if }a=a_i  \\
		      	a   & \text{otherwise} 
		      \end{cases}$
		\item for $w\in\Delta^*$ and $a\in\Delta$,\\
		      $wa\langle\dots\rangle=w\langle\dots\rangle\cdot a\langle\dots\rangle$
		\item for $L\subseteq \Delta^*$,\\
		      $L\langle\li[a_i\gets L_i][n]\rangle=\bigcup\limits_{w\in L} w\langle\li[a_i\gets L_i][n]\rangle$.
	\end{enumerate}
	If $n=1$, $L_1\langle a\gets L_2\rangle$ will also be denoted as $L_1\conc_a L_2$.
	If $\li[L_i][n]$ are singletons, then the substitution is called a \uem{homomorphism}.
\end{definition}
\begin{exercise}\label{exe.2.30}
	Let $n\geq1$, $\li[a_i][n]\in\Sigma_0$ all different, $a_i\neq e$ for all $1\leq i\leq n$, and $L,L_1,\dots L_n\subseteq T_\Sigma$. Prove that
	\[\fyield(L\langle\li[a_i\gets L_i][n]\rangle)=\fyield(L)\langle\li[a_i\gets \fyield(L_i)][n]\rangle.\]
	(Thus: ``yield of tree concatenation is string substitution of yields").
\end{exercise}
\begin{exercise}\label{exe.2.31}
	Prove that Definitions \ref{def.2.27} and \ref{def.2.29} give exactly the same result
for $L\langle\li[a_i\gets L_i][n]\rangle$ where $\li[a_i][n]\in\Sigma_0$ and $L,L_1,\dots L_n$ are tree languages over $\Sigma$ (and thus, string languages over $\Sigma\cup\{[\,,]\}$).
\end{exercise}
\begin{exercise}\label{exe.2.32}
	Define the notion of associativity for tree concatenation, and show that tree concatenation is associative. Show that, in general, ``deterministic tree concatenation" is not associative (cf. Remark \ref{rem.2.28}\ref{rem:2.28ii}).
\end{exercise}
We shall need the following special case of tree concatenation.
\begin{definition}\label{def.2.33}
	Let $\Sigma$ be a ranked alphabet and let $S$ be a set of symbols or a tree language. Then the set of \uem{trees indexed by $S$}, denoted by $T_\Sigma(S)$, is defined inductively as follows.
	\begin{enumerate}[label=(\roman*)]
		\item $S\cup\Sigma_0\subseteq T_\Sigma(S)$
		\item If $k\geq1$, $a\in\Sigma_k$ and $\li\in T_\Sigma(S)$,
		      then $\listt\in T_\Sigma(S)$.
	\end{enumerate}
	Note that $T_\Sigma(\emptyset)=T_\Sigma$.
\end{definition}
Thus, if $S$ is a set of symbols, then $T_\Sigma(S)=T_{\Sigma\cup S}$, where the elements of $S$ are assumed to have rank 0. If $S$ is a tree language over a ranked alphabet $\Delta$, then $T_\Sigma(S)$ is a tree language over the ranked alphabet $\Sigma\cup\Delta$.
\begin{exercise}\label{exe.2.34}
	Show that, for any $a\in\Sigma_0$, $T_\Sigma(S)=T_\Sigma\conc_a(S\cup\{a\})$.
\end{exercise}
We close this section with two general remarks.
\begin{remark}\label{rem.2.35}
	Definition \ref{def.2.5} of a tree is of course rather arbitrary. Other, equally useful, ways of defining trees as a special kind of strings are obtained by replacing $\listt$ in Definition \ref{def.2.5} by $[t_1\cdots t_k]a$ or $at_1\cdots t_k]$ or $[at_1\cdots t_k]$ or $at_1\cdots t_k$ (only in the case that each symbol has exactly one rank) or $\underset{a}{[}t_1\cdots t_k]$ (where $\underset{a}{[}$ is a new symbol for each $a$) or $a[t_1\boldsymbol{,} t_2\boldsymbol{,}\dots\boldsymbol{,}t_k]$ (where ``$\,\boldsymbol{,}\,$" is a new symbol).
\end{remark}
\begin{remark}\label{rem.2.36}
	Remark on the general philosophy in tree language theory. The general philosophy looks like this:
	\begin{figure}[H]
	\centering
		\begin{tikzpicture}[thick,scale=2]
		\usetikzlibrary{decorations.pathmorphing}
		\draw (0,0) -- (1,0);
		\draw (2,0) -- (2,1);
		\draw (3,0) -- (4,0);
		\draw (3,0) -- (3.5,1);
		\draw (4,0) -- (3.5,1);
		\draw (3.5,0) edge[dashed] (3.5,1);
		\draw (5,0) -- (6,0);
		\draw (5,0) edge[dashed] (5.5,1);
		\draw (6,0) edge[dashed] (5.5,1);
		\draw (0.5,0) edge[out=90,in=180,shorten >=0.5cm,shorten <=0.5cm,>>->] node[above] {(1)} (2,0.5);
		\draw (2.3,0.5) edge[>>-] (2.4,0.5);
		\draw (2.4,0.5) edge[->,decorate,
		decoration={snake,amplitude=.4mm,segment length=2.5mm,post length=0.7mm}] node[above,pos=0.45] {(2)} (3.00,0.5);
		\draw (4,0.5) edge[out=0,in=90,shorten >=0.3cm,>>->] node[above,pos=0.4] {(3)} (5.5,0);
		\end{tikzpicture}
	\end{figure}
	\begin{enumerate}[label=(\arabic*)]
		\item Take vertical string language theory (cf. Definition \ref{def.2.21}),
		\item generalize it to tree language theory, and
		\item map this into horizontal string language theory via the yield operation (Definition~\ref{def.2.15}).
	\end{enumerate}
	The fourth part of the philosophy is
	\begin{enumerate}[label=(\arabic*)]
		\setcounter{enumi}{3}
		\item Tree language theory is a specific part of string language theory, illustrated as follows:
	\end{enumerate}

	\begin{align*}&
		\raisebox{-0.5\height}{
			\begin{tikzpicture}
				\node[scale=0.14] at (-1.8,0.8) {\includegraphics{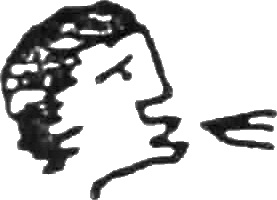}};
				\draw[thick] (-0.9,0.7) -- +(1.8,0);
				\node {$a[b[cd]d]$};
			\end{tikzpicture}}&&\hspace{0.8cm}
		\raisebox{-0.4\height}{
			\begin{tikzpicture}
				\node[scale=0.14] at (0,1.1) {\includegraphics{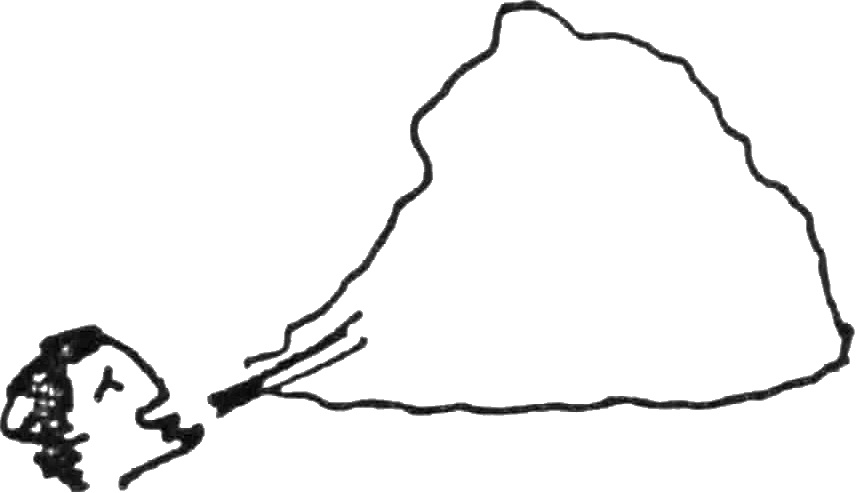}};
				\node(1) {$a$};
				\node(2) [below of=1,node distance=0.2cm,xshift=0.05cm]{$[b~~d]$};
				\node(3) [below of=2,node distance=0.4cm,xshift=-0.06cm]{$[cd]$};
			\end{tikzpicture}}&&\hspace{0.8cm}
		\raisebox{-0.5\height}{
			\begin{tikzpicture}[lighttree]
				\coordinate(t) at (-1.1,0.5);
				\node[scale=0.14] at (-2.1,0.5) {\includegraphics{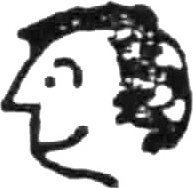}};
				\draw[shorten >=0pt, shorten <=0pt, thick] (t) -- +(1.6,0);
				\draw[shorten >=0pt, shorten <=0pt, thick] (t) -- +(0.8,1.5);
				\draw[shorten >=0pt, shorten <=0pt, thick] (t)+(1.6,0) -- +(0.8,1.5);
				\node at (-0.2,0) {a}
					child {node {b}
						child {node {c}}
						child {node {d}}}
					child {node {d}};
			\end{tikzpicture}}
	\end{align*}
	Example:
	\begin{enumerate}[label=(\arabic*).]
		\item (vertical) string concatenation
		\item tree concatenation
		\item (horizontal) string substitution \hfill{(see Exercise \ref{exe.2.30})}
		\item (2) is a special case of (3) \hfill{(see Exercise \ref{exe.2.31})}
	\end{enumerate}
\end{remark}

\section{Recognizable tree languages}\label{sec.3}
\subsection{Finite tree automata and regular tree grammars}
Let us first consider the usual finite automaton on strings. A deterministic finite automaton is a structure $M=(Q,\Sigma,\delta,q_0,F)$, where $Q$ is the set of states, $\Sigma$ is the input alphabet, $q_0$ is the initial state, $F$ is the set of final states and $\delta$ is a family $\{\delta_a\}_{a\in\Sigma}$, where $\delta_a:Q\to Q$ is the transition function for the input $a$. There are several ways to describe the functioning of $M$ and the language it recognizes. One of them (see for instance \cite[\rom{1}.4]{Sal}), is to describe explicitly the sequence of steps taken by the automaton while processing some input string. This point of view will be considered in Part (\ref{sec.4}). Another way is to give a recursive definition of the effect of an input string on the state of $M$. Since a recursive definition is in particular suitable for generalization to trees, let us consider one in detail. We define a function $\delbu:\Sigma^*\to Q$ such that, for $w\in\Sigma^*$, $\delbu(w)$ is intuitively the state $M$ reaches after processing $w$, starting from the initial state $q_0$:
\begin{enumerate}[label=(\roman*)]
	\item $\delbu(\lambda)=q_0$
	\item for $w\in\Sigma^*$ and $a\in\Sigma$,
	      $\delbu(wa)=\delta_a(\delbu(w))$.
\end{enumerate}
The language recognized by $M$ is $L(M)=\{w\in\Sigma^*\mid \delbu(w)\in F\}$. When considering this definition of $\delbu$ for ``bottom-up" monadic trees (see Definition \ref{def.2.21}), one easily arrives at the following generalization to the tree case:\\
There should be a start state for each element of $\Sigma_0$. The finite tree automaton starts at all leaves (``{}at the same time", ``in parallel") and processes the tree in a bottom-up fashion. The automaton arrives at each node of rank $k$ with a sequence of $k$ states (one state for each direct subtree of the node), and the transition function $\delta_a$ of the label $a$ of that node is a mapping $\delta_a:Q^k\to Q$, which, from that sequence of $k$ states, determines the state at that node. A tree is recognized iff the tree automaton is in a final state at the root of the tree. Formally:
\begin{definition}\label{def.3.1}
	A \uem{deterministic bottom-up finite tree automaton} is a structure \\
 $M=(Q,\Sigma,\delta,s,F)$,\\
	\begin{tabular}{llp{0.78\linewidth}}
		where & $Q$      & is a finite set (of \uem{states}),                                                                                                       \\
		      & $\Sigma$ & is a ranked alphabet (of \uem{input symbols}),                                                                                           \\
		      & $\delta$ & is a family $\{\delta_a^k\}_{k\geq1,a\in\Sigma_k}$ of mappings $\delta_a^k:Q^k\to Q$ (the \uem{transition function} for $a\in\Sigma_k$), \\
		      & $s$      & is a family $\{s_a\}_{a\in\Sigma_0}$ of states $s_a\in Q$ (the \uem{initial state} for $a\in\Sigma_0$), and                                                \\
		      & $F$      & is a subset of $Q$ (the set of \uem{final states}).                                                                                      
	\end{tabular}\\
	The mapping $\delbu:T_\Sigma\to Q$ is defined recursively as follows:
	\begin{enumerate}[label=(\roman*)]
		\item for $a\in\Sigma_0$, $\delbu(a)=s_a$,
		\item for $k\geq1$, $a\in\Sigma_k$ and $\li\in T_\Sigma$,\\[1mm]
		      $\delbu(\listt)=\delta_a^k(\li[\delbu(t_i)])$.
	\end{enumerate}
	The \uem{tree language recognized by $M$} is defined to be $L(M)=\{t\in T_\Sigma\mid \delbu(t)\in F\}$.
\end{definition}
Intuitively, $\delbu(t)$ is the state reached by $M$ after bottom-up processing of $t$.

For convenience, when $k$ is understood, we shall write $\delta_a$ rather than $\delta_a^k$. Note therefore that each symbol $a\in\Sigma$ may have several transition functions $\delta_a$ (one for each of its ranks).

We shall abbreviate ``finite tree automaton" by ``fta", and ``deterministic" by ``det.".

\begin{definition}\label{def.3.2}
	A tree language $L$ is called \uem{recognizable} (or regular) if $L=L(M)$ for some det.\ bottom-up fta $M$.
	
	The class of recognizable tree languages will be denoted by \RECOG.
\end{definition}

\begin{example}\label{exa.3.3}
	Let us consider the det.\ bottom-up fta $M=(Q,\Sigma,\delta,s,F)$, where $Q=\{0,1,2,3\}$, $\Sigma_0=\{0,1,2,\dots,9\}$, $\Sigma_2=\{+,*\}$, $s_a\equiv a\pmod{4}$, $F=\{1\}$, and $\delta_+$ and $\delta_*$ (both mappings $Q^2\to Q$) are addition modulo~4 and multiplication modulo~4 respectively. Then $M$ recognizes the set of all ``expressions" whose value modulo~4 is 1. Consider for instance the expression $+[+[07]*[2*[73]]]$, the prefix form of $(0+7)+(2*(7*3))$. 
In the following picture, 

	\centering
	\begin{tikzpicture}[tree, level distance=1.2cm, node distance=0.5cm,
	lvl1/.style={sibling distance=6cm},
	lvl2/.style={sibling distance=2.3cm}]
	\node(e) {+} [lvl1]
		child {node(0) {+} edge from parent[shorten >=8pt] [lvl2]
			child {node(00) {\vphantom{()}0}}
			child {node(01) {\vphantom{()}7}}}
		child {node(1) {*} [lvl2]
			child {node(10) {\vphantom{()}2}}
			child {node(11) {*}
				child {node(110) {\vphantom{()}7}}
				child {node(111) {\vphantom{()}3}}}};
	\node[right=0.3cm,right of=e] {(1)};
	\node[right=0.3cm,right of=0] {(3)};
	\node[right of=00] {(0)};
	\node[right of=01] {(3)};
	\node[right of=1] {(2)};
	\node[right of=10] {(2)};
	\node[right of=11] {(1)};
	\node[right of=110] {(3)};
	\node[right of=111] {(3)};
	\end{tikzpicture} 

the state of $M$ at each node of the tree is indicated between parentheses.
\end{example}
\begin{example}\label{exa.3.4}
	Let $\Sigma_0=\{a\}$ and $\Sigma_2=\{b\}$. Consider the language of all trees in $T_\Sigma$ which have a ``right comb-like" structure like for instance the tree $b[ab[ab[ab[aa]]]]$. This tree language is recognized by the det.\ bottom-up fta $M=(Q,\Sigma,\delta,s,F)$, where $Q=\{A,C,W\}$, $s_a=A$, $F=\{C\}$ and $\delta_b$ is defined by $\delta_b(A,A)=\delta_b(A,C)=C$ and $\delta_b(q_1,q_2)=W$ for all other pairs of states $(q_1,q_2)$.
\end{example}

\begin{exercise}\label{exe.3.5}
	Let $\Sigma_0=\{a,b\}$, $\Sigma_1=\{p\}$ and $\Sigma_2=\{p,q\}$. Construct det.\ bottom-up finite tree automata recognizing the following tree languages:
	\begin{enumerate}[label=(\roman*)]
		\item\label{exe.3.5i} the language of all trees $t$, such that if a node of $t$ is labeled $q$, then its descendants are labeled $q$ or $a$;
		\item\label{exe.3.5ii} the set of all trees $t$ such that $\fyield(t)\in a^+b^+$;
		\item\label{exe.3.5iii} the set of all trees $t$ such that the total number of $p$'s occurring in $t$ is odd.\qedhere
	\end{enumerate}
\end{exercise}

A (theoretically) convenient extension of the deterministic finite automaton is to make it nondeterministic. A nondeterministic finite automaton (on strings) is a structure $M=(Q,\Sigma,\delta,S,F)$, where $Q$, $\Sigma$ and $F$ are the same as in the deterministic case, $S$ is a set of initial states, and, for each $a\in\Sigma$, $\delta_a$ is a mapping $Q\to\mathcal{P}(Q)$ (intuitively, $\delta_a(q)$ is the set of states which $M$ can possibly, nondeterministically, enter when reading $a$ in state $q$). Again a mapping $\delbu$, now from $\Sigma^*$ into $\mathcal{P}(Q)$, can be defined, such that for every $w\in\Sigma^*$, $\delbu(w)$ is the set of states $M$ can possibly reach after processing $w$, having started from one of the initial states in $S$:
\begin{enumerate}[label=(\roman*)]
	\item $\delbu(\lambda)=S$,
	\item for $w\in\Sigma^*$ and $a\in\Sigma$,
	      $\delbu(wa)=\bigcup\{\delta_a(q)\mid q\in\delbu(w)\}$.
\end{enumerate}
The language recognized by $M$ is 
$L(M)=\{w\in\Sigma^*\mid \delbu(w)\cap F\neq\emptyset\}$.\\
Generalizing to trees we obtain the following definition.

\begin{definition}\label{def.3.6}
    A \uem{nondeterministic bottom-up finite tree automaton} is a 5-tuple \\
$M=(Q,\Sigma,\delta,S,F)$, where $Q$, $\Sigma$ and $F$ are as in the deterministic case, $S$ is a family $\{S_a\}_{a\in\Sigma_0}$ such that $S_a\subseteq Q$ for each $a\in\Sigma_0$, and $\delta$ is a family $\{\delta_a^k\}_{k\geq1,a\in\Sigma_k}$ of mappings $\delta_a^k:Q^k\to\mathcal{P}(Q)$.\\[1mm]
	The mapping $\delbu:T_\Sigma\to\mathcal{P}(Q)$ is defined recursively by
	\begin{enumerate}[label=(\roman*)]
		\item for $a\in\Sigma_0$, $\delbu(a)=S_a$,
		\item for $k\geq1$, $a\in\Sigma_k$ and $\li\in T_\Sigma$,\\[1mm]
		      $\delbu(\listt)=\bigcup\{\delta_a(\li[q_i])\mid q_i\in\delbu(t_i)\text{ for }1\leq i\leq k\}$.
	\end{enumerate}
	The \uem{tree language recognized by $M$} is $L(M)=\{t\in T_\Sigma\mid \delbu(t)\cap F\neq\emptyset\}$.
\end{definition}

Note that, for $\mbar{q}\in Q^k$, $\delta_a^k(\mbar{q})$ may be empty.

\begin{example}\label{exa.3.7}
	Let $\Sigma_0=\{p\}$ and $\Sigma_2=\{a,b\}$. Consider the following tree language over~$\Sigma$:
\[
	\begin{array}{lll}
		L & = & \{ u_1a[a[s_1s_2]a[t_1t_2]]u_2 \in T_\Sigma\mid \cdots\} \;\cup   \\
		  & &   \{u_1b\,[b[s_1s_2]\,b[t_1t_2]]u_2 \in T_\Sigma\mid \cdots\} 
	\end{array}
\]
where ``$\cdots$'' stands for $u_1,u_2\in(\Sigma\cup\{[\,,]\})^*,s_1,s_2,t_1,t_2\in T_\Sigma$.
In other words, $L$~is the set of all trees containing a configuration
	\raisebox{-0.5\height}{\begin{tikzpicture}[lighttree]
	\tikzstyle{level 2}=[sibling distance=0.5cm]
	\node {a}
		child {node {a}
			child{}
			child{}}
		child {node {a}
			child{}
			child{}};
	\end{tikzpicture}}
	or a configuration
	\raisebox{-0.5\height}{\begin{tikzpicture}[lighttree]
	\tikzstyle{level 2}=[sibling distance=0.5cm]
	\node {b}
		child {node {b}
			child{}
			child{}}
		child {node {b}
			child{}
			child{}};
	\end{tikzpicture}}
	(or both).
	$L$ is recognized by the nondet.\ bottom-up fta $M=(Q,\Sigma,\delta,S,F)$, where $Q=\{q_s,q_a,q_b,r\}$, $S_p=\{q_s\}$, $F=\{r\}$ and $\delta_a(q_s,q_s)=\{q_s,q_a\}$, $\delta_b(q_s,q_s)=\{q_s,q_b\}$, $\delta_a(q_a,q_a)=\delta_b(q_b,q_b)=\{r\}$,
	for all $q\in Q:\delta_a(q,r)=\delta_a(r,q)=\delta_b(q,r)=\delta_b(r,q)=\{r\}$,
	and $\delta_x(q_1,q_2)=\emptyset$ for all other possibilities.
\end{example}
It is rather obvious in the last example that we can find a deterministic bottom-up fta recognizing the same language (find it!). We now show that this is possible in general (as in the case of strings).

\begin{theorem}\label{the.3.8}
	For each nondeterministic bottom-up fta we can find a deterministic one recognizing the same language.
\end{theorem}
\begin{proof}
	The proof uses the ``subset-construction", well known from the string-case.
	Let $M=(Q,\Sigma,\delta,S,F)$ be a nondeterministic bottom-up fta. Construct the deterministic bottom-up fta $M_1=(\mathcal{P}(Q),\Sigma,\delta_1,s_1,F_1)$ such that $(s_1)_a=S_a$ for all $a\in\Sigma_0$,
	$F_1=\{Q_1\in\mathcal{P}(Q)\mid Q_1\cap F\neq\emptyset\}$, and, for $a\in\Sigma_k$ and $\li[Q_i]\subseteq Q$,\[(\delta_1)_a(\li[Q_i])=\bigcup\{\delta_a(\li[q_i])\mid q_i\in Q_i\text{ for all }1\leq i\leq k\}.\]
	It is straightforward to show, using Definitions \ref{def.3.1} and \ref{def.3.6}, that
$\delbu_1(t)=\delbu(t)$ for all $t\in T_\Sigma$
(proof by induction on $t$). From this it follows that $L(M_1)=\{t\mid \delbu_1(t)\in F_1\}=$ $\{t\mid \delbu(t)\cap F\neq\emptyset\}=L(M)$.
\end{proof}

\begin{exercise}\label{exe.3.9}
	Check the proof of Theorem \ref{the.3.8}. Construct the det.\ bottom-up fta corresponding to the fta $M$ of Example \ref{exa.3.7} according to that proof, and compare this det.\ fta with the one you found before.
\end{exercise}
Let us now consider the top-down generalization of the finite automaton. Let $M=(Q,\Sigma,\delta,q_0,F)$ be a det. finite automaton. Another way to define $L(M)$ is by giving a recursive definition of a mapping $\widetilde{\delta}:\Sigma^*\to\mathcal{P}(Q)$ such that intuitively, for each $w\in\Sigma^*$, $\deltd(w)$ is the set of states $q$ such that the machine $M$, when started in state $q$, enters a final state after processing $w$. The definition of $\deltd$ is as follows:
\begin{enumerate}[label=(\roman*)]
	\item $\deltd(\lambda)=F$
	\item for $w\in\Sigma^*$ and $a\in\Sigma$,\\
	      $\deltd(aw)=\{q\mid \delta_a(q)\in\deltd(w)\}$
\end{enumerate}
(the last line may be read as: to check whether, starting in $q$, $M$ recognizes $aw$, compute $q_1=\delta_a(q)$ and check whether $M$ recognizes $w$ starting in $q_1$). The language recognized by $M$ is $L(M)=\{w\in\Sigma^*\mid q_0\in\deltd(w)\}$.
This definition, applied to ``top-down" monadic trees, leads to the following generalization to arbitrary trees. The finite tree automaton starts at the root of the tree in the initial state, and processes the tree in a top-down fashion. The automaton arrives at each node in one state, and the transition function $\delta_a$ of the label $a$ of that node is a mapping $\delta_a:Q\to Q^k$ (where $k$ is the rank of the node), which, from that state, determines the state in which to continue for each direct descendant of the node (the automaton ``{}splits up" into $k$ independent copies, one for each direct subtree of the node). Finally the automaton arrives at all leaves of the tree. There should be a set of final states for each element of $\Sigma_0$. The tree is recognized if the fta arrives at each leaf in a state which is final for the label of that leaf. Formally:
\begin{definition}\label{def.3.10}
	A \uem{deterministic top-down finite tree automaton} is a 5-tuple \\
$M=(Q,\Sigma,\delta,q_0,F)$,\\
	\begin{tabular}{llp{0.78\linewidth}}
		where & $Q$      & is a finite set (of \uem{states}),                                                                                                       \\
		      & $\Sigma$ & is a ranked alphabet (of \uem{input symbols}),                                                                                           \\
		      & $\delta$ & is a family $\{\delta_a^k\}_{k\geq1,a\in\Sigma_k}$ of mappings $\delta_a^k:Q\to Q^k$ (the \uem{transition function} for $a\in\Sigma_k$), \\
		      & $q_0$    & is in $Q$ (the \uem{initial state}), and                                                                                                 \\
		      & $F$      & is a family $\{F_a\}_{a\in\Sigma_0}$ of sets $F_a\subseteq Q$ (the set of \uem{final states} for $a\in\Sigma_0$).                        
	\end{tabular}\\
	The mapping $\deltd:T_\Sigma\to\mathcal{P}(Q)$ is defined recursively by
	\begin{enumerate}[label=(\roman*)]
		\item for $a\in\Sigma_0$, $\deltd(a)=F_a$
		\item for $k\geq1$, $a\in\Sigma_k$ and $\li\in T_\Sigma$,\\[1mm]
		      $\deltd(\listt)=\{q\mid \delta_a(q)\in\li[\deltd(t_i)][k][\times]\}$.
	\end{enumerate}
	The \uem{tree language recognized by $M$} is defined to be
	$L(M)=\{t\in T_\Sigma\mid q_0\in\deltd(t)\}$.
\end{definition}
Intuitively, $\deltd(t)$ is the set of states $q$ such that $M$, when starting at the root of $t$ in state $q$, arrives at the leaves of $t$ in final states.
\begin{example}\label{exa.3.11}
	Consider the tree language of Exercise \ref{exe.3.5}\ref{exe.3.5i}. A det. top-down fta recognizing this language is $M=(Q,\Sigma,\delta,q_0,F)$ where $Q=\{A,R,W\}$, $q_0=A$, $F_a=\{A,R\}$, $F_b=\{A\}$ and
	\begin{align*}
		&\delta_p^1(A)=A,     && \delta_p^1(R)=\delta_p^1(W)=W,         \\
		&\delta_p^2(A)=(A,A), && \delta_p^2(R)=\delta_p^2(W)=(W,W),     \\
		&\delta_q(A)=(R,R),   && \delta_q(R)=(R,R),\;\;\delta_q(W)=(W,W). \qedhere
	\end{align*}
	
\end{example}
\begin{exercise}\label{exe.3.12}
	Let $\Sigma$ be a ranked alphabet, and $p\in\Sigma_2$. Let $L$ be the tree language defined recursively by
	\begin{enumerate}[label=(\roman*)]
		\item for all $t_1,t_2\in T_\Sigma$, $ p[t_1t_2] \in L$
		\item for all $a\in\Sigma_k$, if $\li\in L$, then $\listt\in L$ ($k\geq1$).
	\end{enumerate}
	Construct a deterministic top-down fta recognizing $L$. Give a nonrecursive description of~$L$.
\end{exercise}
\begin{exercise}\label{exe.3.13}
	Construct a det.\ top-down fta $M$ such that $\fyield(L(M))=a^+b^+$.
\end{exercise}
We now show that the det.\ top-down fta recognizes less languages than its bottom-up counterpart.
\begin{theorem}\label{the.3.14}
	There are recognizable tree languages which cannot be recognized by a deterministic top-down fta.
\end{theorem}
\begin{proof}
	Let $\Sigma_0=\{a,b\}$ and $\Sigma_2=\{S\}$. Consider the (finite!) tree language $L=\{ S[ab] , S[ba] \}$.
	Suppose that the det. top-down fta $M=(Q,\Sigma,\delta,q_0,F)$ recognizes $L$. Let $\delta_S(q_0)=(q_1,q_2)$. Since $ S[ab] \in L(M)$, $q_1\in F_a$ and $q_2\in F_b$. But, since $ S[ba] \in L(M)$, $q_1\in F_b$ and $q_2\in F_a$. Hence both $S[aa]$ and $S[bb]$ are in $L(M)$. Contradiction.
\end{proof}
\begin{exercise}\label{exe.3.15}
	Show that the tree languages of Exercise \ref{exe.3.5}(ii,iii) are not recognizable by a det. top-down fta.
\end{exercise}
It will be clear that the nondeterministic top-down fta is able to recognize all recognizable languages. We give the definition without comment.
\begin{definition}\label{def.3.16}
	A \uem{nondeterministic top-down finite tree automaton} is a structure 
$M=(Q,\Sigma,\delta,S,F)$, where $Q$, $\Sigma$ and $F$ are as in the deterministic case, $S$ is a subset of $Q$ and $\delta$ is a family $\{\delta_a^k\}_{k\geq1,a\in\Sigma_k}$ of mappings $\delta_a^k:Q\to\mathcal{P}(Q^k)$.\\
	The mapping $\deltd:T_\Sigma\to\mathcal{P}(Q)$ is defined recursively as follows
	\begin{enumerate}[label=(\roman*)]
		\item for $a\in\Sigma_0$, $\deltd(a)=F_a$,
		\item for $k\geq1$, $a\in\Sigma_k$ and $\li\in T_\Sigma$,\\[1mm]
		      $\deltd(\listt)=\{q\mid \exists(\li[q_i])\in\delta_a(q):q_i\in\deltd(t_i)\text{ for all }1\leq i\leq k\}$.
	\end{enumerate}
	The \uem{tree language recognized by $M$} is $L(M)=\{t\in T_\Sigma\mid \deltd(t)\cap S\neq\emptyset\}$.
\end{definition}
We now show that, nondeterministically, there is no difference between bottom-up or top-down recognition.
\begin{theorem} \label{the.3.17}
	A tree language is recognizable by a nondet.\ bottom-up fta iff it is recognizable by a nondet.\ top-down fta.
\end{theorem}
\begin{proof}
	Let us say that a nondet.\ bottom-up fta $M=(Q,\Sigma,\delta,S,F)$ and a nondet.\ top-down fta $N=(P,\Delta,\mu,R,G)$ are ``{}associated" if the following requirements are satisfied:
	\begin{enumerate}[label=(\roman*)]
		\item $Q=P$, $\Sigma=\Delta$, $F=R$ and, for all $a\in\Sigma_0$, $S_a=G_a$;
		\item for all $k\geq1$, $a\in\Sigma_k$ and $\li[q_i],q\in Q$,\\
		      $q\in\delta_a(\li[q_i])$ iff $(\li[q_i])\in\mu_a(q)$.
	\end{enumerate}
	In that case, one can easily prove by induction that $\delbu=\deltd[\mu]$, and so $L(M)=L(N)$. Since obviously for each nondet.\ bottom-up fta there is an associated nondet.\ top-down fta, and vice versa, the theorem holds.
\end{proof}
Thus the classes of tree languages recognized by the nondet.\ bottom-up, det.\ bottom-up and nondet.\ top-down fta are all equal (and are called \RECOG), whereas the class of tree languages recognized by the det.\ top-down fta is a proper subclass of \RECOG.

The next victim of generalization is the regular grammar (right-linear, type-3 grammar). In this case it seems appropriate to take the top-down point of view only. Consider an ordinary regular grammar $G=(N,\Sigma,R,S)$. All rules have either the form $A\to wB$ or the form $A\to w$, where $A,B\in N$ and $w\in\Sigma^*$. Monadically, the string $wB$ may be considered as the result of treeconcatenating the tree $we$ with $B$ at $e$, where $B$ is of rank~0. Thus we can take the generalization of strings of the form $wB$ or $w$ to be trees in $T_\Delta(N)$, where $\Delta$ is a ranked alphabet (for the definition of $T_\Delta(N)$, see Definition \ref{def.2.33}). Thus, let us consider a ``tree grammar" with rules of the form $A\to t$, where $A\in N$ and $t\in T_\Delta(N)$. Obviously, the application of a rule $A\to t$ to a tree $s\in T_\Delta(N)$ should intuitively consist of replacing one occurrence of $A$ in $s$ by the tree $t$. Starting with the initial nonterminal, nonterminals at the frontier of the tree are then repeatedly replaced by right hand sides of rules, until the tree does not contain nonterminals any more. Now, since trees are defined as strings, it turns out that this process is precisely the way a context-free grammar works. Thus we arrive at the following formal definition.
\begin{definition}\label{def.3.18}
	A \uem{regular tree grammar} is a tuple $G=(N,\Sigma,R,S)$ where $N$ is a finite set (of \uem{nonterminals}), $\Sigma$ is a ranked alphabet (of \uem{terminals}), such that $\Sigma\cap N=\emptyset$, $S\in N$ is the \uem{initial nonterminal}, and $R$ is a finite set of \uem{rules} of the form $A\to t$ with $A\in N$ and $t\in T_\Sigma(N)$. The \uem{tree language generated by $G$}, denoted by $L(G)$, is defined to be $L(H)$, where $H$ is the context-free grammar $(N,\Sigma\cup\{[\,,]\},R,S)$.
	We shall use $\xRightarrow[G]{}$ and $\xRightarrow[G]{*}$ (or $\Rightarrow$ and $\xRightarrow{*}$ when $G$ is understood) to denote the restrictions of $\xRightarrow[H]{}$ and $\xRightarrow[H]{*}$ to $T_\Sigma(N)$.
\end{definition}
\begin{example}\label{exa.3.19}
	Let $\Sigma_0=\{a,b,c,d,e\}$, $\Sigma_2=\{p\}$ and $\Sigma_3=\{p,q\}$. Consider the regular tree grammar $G=(N,\Sigma,R,S)$, where $N=\{S,T\}$ and $R$ consists of the rules $S\to p[aTa] $, $T\to q[cp[dT]b] $ and $T\to  e $. Then $G$ generates the tree $p[aq[cp[de]b]a]$ as follows:
	\[S\Rightarrow p[aTa] \Rightarrow p[aq[cp[dT]b]a] \Rightarrow p[aq[cp[de]b]a] \]
	or, pictorially,
	\begin{figure}[H]
	\centering
	\begin{tikzpicture}[tree]
		\node {S};
	\end{tikzpicture}
	$\Rightarrow$
	\begin{tikzpicture}[tree]
		\node {p} 
			child {node {a}}
			child {node {T}}
			child {node {a}};
	\end{tikzpicture}
	$\Rightarrow$
	\begin{tikzpicture}[tree]
		\node {p}
			child {node {a}}
			child {node {q} 
				child {node {c}}
				child {node {p}
					child {node {d}}
					child {node {T}}}
				child {node {b}}}
			child {node {a}};
	\end{tikzpicture}
	$\Rightarrow$
	\begin{tikzpicture}[tree]
		\node {p} 
			child {node {a}}
			child {node {q}
				child {node {c}}
				child {node {p}
					child {node {d}}
					child {node {e}}}
				child {node {b}}}
			child {node {a}};
	\end{tikzpicture}.
	\end{figure}
	The tree language generated by $G$ is $\{p[a(q[cp[d)^ne(]b])^na]\mid n\geq 0\}$.
\end{example}
\begin{exercise}\label{exe.3.20}
	Write regular tree grammars generating the tree languages of Exercise~\ref{exe.3.5}.
\end{exercise}
As in the case of strings, each regular tree grammar is equivalent to one that has the property that at each step in the derivation exactly one terminal symbol is produced.
\begin{definition}\label{def.3.21}
	A regular tree grammar $G=(N,\Sigma,R,S)$ is \uem{in normal form}, if each of its rules is either of the form $A\to\listt[a][B_i]$ or of the form $A\to b$, where $k\geq1$, $a\in\Sigma_k$, $A,\li[B_i]\in N$ and $b\in\Sigma_0$.
\end{definition}
\begin{theorem}\label{the.3.22}
	Each regular tree grammar has an equivalent regular tree grammar in normal form.
\end{theorem}
\begin{proof}
	Consider an arbitrary regular tree grammar $G=(N,\Sigma,R,S)$. Let $G_1=(N,\Sigma,R_1,S)$ be the regular tree grammar such that $(A\to t)\in R_1$ if and only if $t\notin N$ and there is a $B$ in $N$ such that $A\xRightarrow[G]{*}B$ and $(B\to t)\in R_1$. Then $L(G_1)=L(G)$, and $R_1$ does not contain rules of the form $A\to B$ with $A,B\in N$. (This is the well-known procedure of removing rules $A\to B$ from a context-free grammar). Suppose that $G_1$ is not yet in normal form. Than there is a rule of the form $A\to a[t_1\cdots t_i\cdots t_k]$ such that $t_i\notin N$. Construct a new regular tree grammar $G_2$ by adding a new nonterminal $B$ to $N$ and replacing the rule $A\to a[t_1\cdots t_i\cdots t_k]$ by the two rules $A\to a[t_1\cdots B\cdots t_k]$ and $B\to t_i$ in $R_1$. It should be clear that $L(G_2)=L(G_1)$, and that, by repeating the latter process a finite number of times, one ends up with an equivalent grammar in normal form.
\end{proof}
\begin{exercise}\label{exe.3.23}
	Put the regular tree grammar of Example \ref{exa.3.19} into normal form.
\end{exercise}
\begin{exercise}\label{exe.3.24}
	What does Theorem \ref{the.3.22} actually say in the case of strings (the monadic case)?
\end{exercise}
In the next theorem we show that the regular tree grammars generate exactly the class of recognizable tree languages.
\begin{theorem}\label{the.3.25}
	A tree language can be generated by a regular tree grammar iff it is an element of \RECOG.
\end{theorem}
\begin{proof}
	Exercise.
\end{proof}
Note therefore that each recognizable tree language is a special kind of context-free language.
\begin{exercise}\label{exe.3.26}
	Show that all finite tree languages are in \RECOG.
\end{exercise}
\begin{exercise}\label{exe.3.27}
	Show that each recognizable tree language can be generated by a ``backwards deterministic" regular tree grammar. A regular tree grammar is called ``backwards deterministic" if (1) it may have more than one initial nonterminal, (2) it is in normal form, and (3) rules with the same right hand side are equal.
\end{exercise}
It is now easy to show the connection between recognizable tree languages and context-free languages.
Let \CFL{} denote the class of context-free languages. 
\begin{theorem}\label{the.3.28}
	$\fyield(\RECOG)=\CFL$ (in words, the yield of each recognizable tree language is context-free, and each context-free language is the yield of some recognizable tree language).
\end{theorem}
\begin{proof}
	Let $G=(N,\Sigma,R,S)$ be a regular tree grammar. Consider the context-free grammar $\mbar{G}=(N,\Sigma_0,\mbar{R},S)$, where $\mbar{R}=\{A\to\fyield(t)\mid A\to t\in R\}$. Then $L(\mbar{G})=\fyield(L(G))$.
	
	Now let $G=(N,\Sigma,R,S)$ be a context-free grammar. Let $*$ be a new symbol, and let $\Delta=\Sigma\cup\{e,*\}$ be the ranked alphabet such that $\Delta_0=\Sigma\cup\{e\}$, and, for $k\geq1$, $\Delta_k=\{*\}$ if and only if there is a rule in $R$ with a right hand side of length $k$. Consider the regular tree grammar $\mbar{G}=(N,\Delta,\mbar{R},S)$ such that
	\begin{enumerate}[label=(\roman*)]
		\item if $A\to w$ is in $R$, $w\neq\lambda$, then $A\to  *[w] $ is in $\mbar{R}$,
		\item if $A\to\lambda$ is in $R$, then $A\to e$ is in $\mbar{R}$.
	\end{enumerate}
	Then $\fyield(L(\mbar{G}))=L(G)$.
\end{proof}
In the next section we shall give the connection between regular tree languages and derivation trees of context-free languages.
\begin{exercise}\label{exe.3.29}
	A context-free grammar is ``invertible" if rules with the same right hand side are equal. Show that each context-free language can be generated by an invertible context-free grammar.
\end{exercise}
For regular string languages a useful stronger version of Theorem \ref{the.3.28} can be proved.
\begin{theorem}\label{the.3.30}
	Let $\Sigma$ be a ranked alphabet. If $R$ is a regular string language over $\Sigma_0$, then the tree language $\{t\in T_\Sigma\mid \fyield(t)\in R\}$ is recognizable.
\end{theorem}
\begin{proof}
	Let $M=(Q,\Sigma,\delta,q_0,F)$ be a deterministic finite automaton recognizing $R$. We construct a nondeterministic bottom-up fta $N=(Q\times Q,\Sigma,\mu,S,G)$, which, for each tree $t$, checks whether a successful computation of $M$ on $\fyield(t)$ is possible. The states of $N$ are pairs of states of $M$. Intuitively we want that $(q_1,q_2)\in\delbu[\mu](t)$ if and only if $M$ arrives in state $q_2$ after processing $\fyield(t)$, starting from state $q_1$. Thus we define
	\begin{enumerate}[label=(\roman*)]
		\item for all $a\in\Sigma_0$, $S_a=\{(q_1,q_2)\mid \delta_a(q_1)=q_2\}$,
		\item for all $k\geq1$, $a\in\Sigma_k$ and states $\li*[q_i][2k]\in Q$,\\
		      $\mu_a((q_1,q_2),(q_3,q_4),\dots,(q_{2k-1},q_{2k}))=
		      \begin{cases}
		      	\{(q_1,q_{2k})\} & \text{if } q_{2i}=q_{2i+1}\text{ for}\\
		      	&\hspace{10pt}\text{all } 1\leq i\leq k-1 \\
		      	\emptyset        & \text{otherwise .}                                           
		      \end{cases}$
	\end{enumerate}
	Then $L(N)=\{t\in T_\Sigma\mid \fyield(t)\in R\}$.
\end{proof}
\begin{exercise}\label{exe.3.31}
	Show that, if $\Sigma_2\neq\emptyset$, then Theorem \ref{the.3.30} holds conversely: if $L$ is a string language such that $\{t\in T_\Sigma\mid \fyield(t)\in L\}$ is recognizable, then $L$ is regular. What can you say in case $\Sigma_2=\emptyset$?
\end{exercise}

\subsection{Closure properties of recognizable tree languages}
We first consider set-theoretic operations.
\begin{theorem}\label{the.3.32}
	\RECOG{} is closed under union, intersection and complementation.
\end{theorem}
\begin{proof}
	To show closure under complementation, consider a deterministic bottom-up fta $M=(Q,\Sigma,\delta,s,F)$. Let $N$ be the det.\ bottom-up fta $(Q,\Sigma,\delta,s,Q-F)$. Then, obviously, $L(N)=T_\Sigma-L(M)$.
	
	To show closure under union, consider two regular tree grammars $G_i=(N_i,\Sigma_i,R_i,S_i)$, $i=1,2$ (with $N_1\cap N_2=\emptyset$).
	Then $G=(N_1\cup N_2\cup\{S\},\Sigma_1\cup\Sigma_2,R_1\cup R_2\cup\{S\to S_1,$ $S\to S_2\},S)$
	is a regular tree grammar such that $L(G)=L(G_1)\cup L(G_2)$.
\end{proof}
As a corollary we obtain the following closure property of context-free languages.
\begin{corollary}\label{cor.3.33}
	\CFL{} is closed under intersection with regular languages.
\end{corollary}
\begin{proof}
	Let $L$ and $R$ be a context-free and regular language respectively. According to Theorem \ref{the.3.28}, there is a recognizable tree language $U$ such that $\fyield(U)=L$. Consequently, by Theorems \ref{the.3.30} and \ref{the.3.32}, the tree language $V=U\cap\{t\mid \fyield(t)\in R\}$ is recognizable. Obviously $L\cap R=\fyield(V)$ and so, again by Theorem \ref{the.3.28}, $L\cap R$ is context-free.
\end{proof}
We now turn to the closure of \RECOG{} under concatenation operations 
(see Definitions~\ref{def.2.23} and~\ref{def.2.27}).
\begin{theorem}\label{the.3.34}
	For every $k\geq1$ and $a\in\Sigma_k$, \RECOG{} is closed under $\tc_a^k$.
\end{theorem}
\begin{proof}
	Exercise.
\end{proof}
\begin{theorem}\label{the.3.35}
	\RECOG{} is closed under tree concatenation.
\end{theorem}
\begin{proof}
	The proof is obtained by generalizing that for regular string languages. Let $n\geq1$, $\li[a_i][n]\in\Sigma_0$ all different and $L_0,\li[L_i][n]$ recognizable tree languages (we may assume that all languages are over the same ranked alphabet $\Sigma$). Let $G_i=(N_i,\Sigma,R_i,S_i)$ be a regular tree grammar in normal form for $L_i$ ($i=0,1,\dots,n$). A regular tree grammar generating $L_0\langle\li[a_i\gets L_i][n]\rangle$ is $G=(\bigcup\limits_{i=0}^nN_i,\Sigma,R,S_0)$, where $R=\mbar{R}_0\cup\bigcup\limits_{i=1}^n R_i$, and $\mbar{R}_0$ is $R_0$ with each rule of the form $A\to a_i$ replaced by the rule $A\to S_i$ ($1\leq i\leq n$).
\end{proof}
\begin{corollary}\label{cor.3.36}
	\CFL{} is closed under substitution.
\end{corollary}
\begin{proof}
	Use Theorem \ref{the.3.28} and Exercise \ref{exe.2.30}.
\end{proof}
Note also that Theorem \ref{the.3.35} is essentially a special case of Corollary \ref{cor.3.36}.

Next we generalize the notion of (concatenation) closure of string languages to trees, and show that \RECOG{} is closed under this closure operation. We shall, for convenience, restrict ourselves to the case that tree concatenation happens at one element of $\Sigma_0$.
\begin{definition}\label{def.3.37}
	Let $a\in\Sigma_0$ and let $L$ be a tree language over $\Sigma$. Then the \uem{tree concatenation closure} of $L$ at $a$, denoted by $L^{*a}$, is defined to be $\bigcup\limits_{n=0}^\infty X_n$, where $X_0=\{a\}$ and, for $n\geq0$, $X_{n+1}=X_n\conc_a(L\cup\{a\})$.\footnote{Recall the notation $L_1\conc_a L_2$ from Definition \ref{def.2.27}.}
\end{definition}
\begin{example}\label{exa.3.38}
	Let $G=(N,\Sigma,R,S)$ be the regular tree grammar with $N=\{S\}$, $\Sigma_0=\{a\}$, $\Sigma_2=\{b\}$ and $R=\{S\to b[aS],\,S\to a \}$. Then $L(G)=\{ b[aS] \}^{*S}\conc_Sa$.
\end{example}
The ``corresponding" operation on strings has several names in the literature. Let us call it ``{}substitution closure".
\begin{definition}\label{def.3.39}
	Let $\Delta$ be an alphabet and $a\in\Delta$. For a language $L$ over $\Delta$, the \uem{substitution closure} of $L$ at $a$, denoted by $L^{*a}$, is defined to be $\bigcup\limits_{n=0}^\infty X_n$, where $X_0=\{a\}$ and, for $n\geq0$, $X_{n+1}=X_n\conc_a(L\cup\{a\})$.
\end{definition}
\begin{exercise}\label{exe.3.40}
	Let $a\in\Sigma_0$, $a\neq e$, and let $L\subseteq T_\Sigma$. Prove that $\fyield(L^{*a})=(\fyield(L))^{*a}$.
\end{exercise}
\begin{theorem}\label{the.3.41}
	\RECOG{} is closed under tree concatenation closure.
\end{theorem}
\begin{proof}
	Again the proof is a straightforward generalization of the string case. Let $G=(N,\Sigma,R,S)$ be a regular tree grammar in normal form, and let $a\in\Sigma_0$. Construct the regular tree grammar $\mbar{G}=(N\cup\{S_0\},\Sigma,\mbar{R},S_0)$, where $\mbar{R}=R\cup\{A\to S\mid A\to a\text{ is in }R\}\cup\{S_0\to S,S_0\to a\}$. Then $L(\mbar{G})=(L(G))^{*a}$.
\end{proof}
\begin{corollary}\label{cor.3.42}
	\CFL{} is closed under substitution closure.
\end{corollary}
\begin{proof}
	Use Theorem \ref{the.3.28} and Exercise \ref{exe.3.40}.
\end{proof}
It is well known that the class of regular string languages is the smallest class containing the finite languages and closed under union, concatenation and closure. A similar result holds for recognizable tree languages.

\begin{theorem}\label{the.3.43}
	\RECOG{} is the smallest class of tree languages containing the finite tree languages and closed under union, tree concatenation and tree concatenation closure.
\end{theorem}
\begin{proof}
	We have shown that \RECOG{} satisfies the above conditions in Exercise \ref{exe.3.26} and Theorems \ref{the.3.32}, \ref{the.3.35} and \ref{the.3.41}. It remains to show that every recognizable tree language can be built up from the finite tree languages using the operations $\cup$, $\conc_a$ and $^{*a}$. Let $G=(N,\Sigma,R,S)$ be a regular tree grammar (it is easy to think of it as being in normal form). We shall use the elements of $N$ to do tree concatenation at. 
For $A\in N$ and $P,Q\subseteq N$ with $P\cap Q=\emptyset$, 
let us denote by $L_{A,P}^Q$ the set of all trees $t\in T_\Sigma(P)$ for which there is a derivation 
$A\Rightarrow t_1 \Rightarrow t_2 \Rightarrow \cdots \Rightarrow t_n \Rightarrow t_{n+1}=t$ ($n \ge 0$)  
such that, for $1\leq i\leq n$, $t_i\in T_\Sigma(Q\cup P)$ and a rule with left hand side in $Q$ is applied to $t_i$ to obtain $t_{i+1}$. We shall show, by induction on the cardinality of $Q$, that all sets $L_{A,P}^Q$ can be built up from the finite tree languages by the operations $\cup$, $\conc_B$ and $^{*B}$ (for all $B\in N$). For $Q=\emptyset$, $L_{A,P}^\emptyset$ is the set of all those right hand sides of rules with left hand side $A$, that are in $T_\Sigma(P)$. Thus $L_{A,P}^\emptyset$ is a finite tree language for all $A$ and $P$. Assuming now that, for $Q\subseteq N$, all sets $L_{A,P}^Q$ can be built up from the finite tree languages, the same holds for all sets $L_{A,P}^{Q\cup\{B\}}$, where $B\in N-Q$, since 
\[L_{A,P}^{Q\cup\{B\}}=L_{A,P\cup\{B\}}^Q\conc_B(L_{B,P\cup\{B\}}^Q)^{*B}\conc_BL_{B,P}^Q\] 
(a formal proof of this equation is left to the reader).
Thus, since $L(G)=L_{S,\emptyset}^N$, the theorem is proved.
\end{proof}
In other words, each recognizable tree language can be denoted by a ``regular expression" with trees as constants and $\cup$, $\conc_A$ and $^{*A}$ as operators.
\begin{exercise}\label{exe.3.44}
	Try to find a regular expression for the language generated by the regular tree grammar $G=(N,\Sigma,R,S)$ with $N=\{S,T\}$, $\Sigma_0=\{a\}$, $\Sigma_2=\{p\}$ and $R=\{S\to p[TS],\,S\to a,\,T\to p[TT],\,T\to a \}$. Use the algorithm in the proof of Theorem \ref{the.3.43}.
\end{exercise}
As a corollary we obtain the result that all context-free languages can be denoted by ``context-free expressions".
\begin{corollary}\label{cor.3.45}
	\CFL{} is the smallest class of languages containing the finite languages and closed under union, substitution and substitution closure.
\end{corollary}
\begin{proof}
	Exercise.
\end{proof}
\begin{exercise}\label{exe.3.46}
	Define the operation of ``iterated concatenation at $a$" (for tree languages) and ``iterated substitution at $a$" (for string languages) by $\iterated_a(L)=L^{*a}\conc_a\emptyset$. Prove (using Theorem \ref{the.3.43}) that \RECOG{} is the smallest class of tree languages containing the finite tree languages and closed under the operations of union, top concatenation and iterated concatenation. Show that this implies that \CFL{} is the smallest class of languages containing the finite languages and closed under the operations of union, concatenation and iterated substitution (cf. \cite[\rom{6}.11]{Sal}).
\end{exercise}
Let us now turn to another operation on trees: that of relabeling the nodes of a tree.
\begin{definition}\label{def.3.47}
	Let $\Sigma$ and $\Delta$ be ranked alphabets. A \uem{relabeling} $r$ is a family $\{r_k\}_{k\geq0}$ of mappings $r_k:\Sigma_k\to\mathcal{P}(\Delta_k)$. A relabeling determines a mapping $r:T_\Sigma\to\mathcal{P}(T_\Delta)$ by the requirements
	\begin{enumerate}[label=(\roman*)]
		\item for $a\in\Sigma_0$, $r(a)=r_0(a)$,
		\item for $k\geq1$, $a\in\Sigma_k$ and $\li\in T_\Sigma$,\\
		      $r(\listt)=\{\listt[b][s_i]\mid b\in r_k(a)\text{ and }s_i\in r(t_i)\}$.
	\end{enumerate}
	If, for each $k\geq0$ and each $a\in\Sigma_k$, $r_k(a)$ consists of one element only, then $r$ is called a \uem{projection}.
\end{definition}
Obviously, \RECOG{} is closed under relabelings.
\begin{theorem}\label{the.3.48}
	\RECOG{} is closed under relabelings.
\end{theorem}
\begin{proof}
	Let $r$ be a relabeling, and consider some regular tree grammar $G$. By replacing each rule $A\to t$ of $G$ by all rules $A\to s$, $\;s\in r(t)$, one obtains a regular tree grammar for $r(L(G))$. (In order that ``$r(t)$" makes sense, we define $r(B)=\{B\}$ for each nonterminal $B$ of $G$).
\end{proof}
We are now in a position to study the connection between recognizable tree languages and sets of derivation trees of context-free grammars. We shall consider two kinds of derivation trees. First we define the ``{}ordinary" kind of derivation tree (cf. Example \ref{exa.1.1}).
\begin{definition}\label{def.3.49}
	Let $G=(N,\Sigma,R,S)$ be a context-free grammar. Let $\Delta$ be the ranked alphabet such that $\Delta_0=\Sigma\cup\{e\}$ and, for $k\geq1$, $\Delta_k$ is the set of nonterminals $A\in N$ for which there is a rule $A\to w$ with $|w| =k$ (in case $k=1:|w| =1\text{ or }|w| =0$).
	For each $\alpha\in N\cup\Sigma$, the \uem{set of derivation trees} with top $\alpha$, denoted by $D_G^\alpha$, is the tree language over $\Delta$ defined recursively as follows
	\begin{enumerate}[label=(\roman*)]
		\item for each $a$ in $\Sigma$, $a\in D_G^a$;
		\item for each rule $A\to\alpha_1\cdots\alpha_n$ in $R$ ($n\geq1$, $A\in N$, $\alpha_i\in\Sigma\cup N$),\\
		      if $t_i\in D_G^{\alpha_i}$ for $1\leq i\leq n$, then $\listt[A][t_i][n]\in D_G^A$;
		\item for each rule $A\to\lambda$ in $R$, $A[e]\in D_G^A$.\qedhere
	\end{enumerate}
\end{definition}
\begin{definition}\label{def.3.50}
	A tree language $L$ is said to be \uem{local} if, for some context-free grammar $G=(N,\Sigma,R,S)$ and some set of symbols $V\subseteq N\cup\Sigma$, $L=\bigcup\limits_{\alpha\in V}D_G^\alpha$.
\end{definition}
\begin{exercise}\label{exe.3.51}
	Show that each local tree language is recognizable.
\end{exercise}
Note that a local tree language is the set of all derivation trees of a context-free grammar which has a set of initial symbols (instead of one initial nonterminal).

The reason for the name ``local" is that such a tree language $L$ is determined by (1)~a~finite set of trees of height one, (2) a finite set of ``initial symbols", (3) a finite set of ``final symbols", and the requirement that $L$ consists of all trees $t$ such that each node of $t$ together with its direct descendants belongs to (1), the top label of $t$ belongs to (2), and the leaf labels of $t$ to (3).

We now show that the class of local tree languages is properly included in \RECOG.
\begin{theorem}\label{the.3.52}
	There are recognizable tree languages which are not local.
\end{theorem}
\begin{proof}
	Let $\Sigma_0=\{a,b\}$ and $\Sigma_2=\{S\}$. Consider the tree language $L=\{ S[S[ba]S[ab]]\} $. Obviously $L$ is recognizable. Suppose that $L$ is local. Then there is a context-free grammar $G$ such that $D_G^S=L$. Thus $S\to SS$, $S\to ba$ and $S\to ab$ are rules of $G$. But then $ S[S[ab]S[ba]] \in L$. Contradiction.\footnote{Other examples are for instance $\{ S[T[a]T[b]] \}$ and $\{ S[S[a]] \}$.}
\end{proof}
Note that the recognizable tree language $L$ in the above proof can be recognized by a deterministic top-down fta. Note also that the tree language given in the proof of Theorem \ref{the.3.14} is local. Hence the local tree languages and the tree languages recognized by det.\ top-down fta are incomparable.
\begin{exercise}\label{exe.3.53}
	Find a recognizable tree language which is neither local nor recognizable by a det.\ top-down fta.
\end{exercise}
It is clear that, if $\Sigma_0=\{a,b\}$ and $\Sigma_2=\{S_1,S_2,S_3\}$, then $L'=\{ S_1[S_2[ba]S_3[ab]] \}$ is a local language. Hence the language $L$ in Theorem \ref{the.3.52} is the projection of the local language $L'$ (project $S_1$, $S_2$ and $S_3$ on $S$). We will show that this is true in general: each recognizable tree language is the projection of a local tree language. In fact we shall show a slightly stronger fact. To do this we define the second type of derivation tree of a context-free grammar, called ``rule tree".
\begin{definition}\label{def.3.54}
	Let $G=(N,\Sigma,R,S)$ be a context-free grammar. Let $\mbar{R}$ be any set of symbols in one-to-one correspondence with $R$, $\mbar{R}=\{\mbar{r}\mid r\in R\}$. Each element of $\mbar{R}$ is given a rank such that, if $r$ in $R$ is of the form $A\to w_0A_1w_1A_2w_2\cdots A_kw_k$ (for some $k\geq0$, $\li[A_i]\in N$ and $w_0,\li[w_i]\in\Sigma^*$), then $\mbar{r}\in\mbar{R}_k$. The \uem{set of rule trees} of $G$, denoted by $\mathit{RT}(G)$, is defined to be the tree language generated by the regular tree grammar $\mbar{G}=(N,\mbar{R},P,S)$, where $P$ is defined by
	\begin{enumerate}[label=(\roman*)]
		\item if $r=(A\to w_0A_1\cdots w_{k-1}A_kw_k)$, $k\geq1$, is in $R$,\\
		      then $A\to\listt[\mbar{r}][A_i]$ is in $P$;
		\item if $r=(A\to w_0)$ is in $R$,\\
		      then $A\to\mbar{r}$ is in $P$.\qedhere
	\end{enumerate}
\end{definition}
\begin{definition}\label{def.3.55}
	We shall say that a tree language $L$ is a \uem{rule tree language} if $L=\mathit{RT}(G)$ for some context-free grammar $G$.
\end{definition}
Thus, a rule tree is a derivation tree in which the nodes are labeled by the rules applied during the derivation. It should be obvious, that for each context-free grammar $G=(N,\Sigma,R,S)$ there is a one-to-one correspondence between the tree languages $\mathit{RT}(G)$ and $D_G^S$.
\begin{example}\label{exa.3.56}
	Consider Example \ref{exa.1.1}. For each rule $r$ in that example, let $(r)$ stand for a new symbol. The rule tree ``corresponding" to the derivation tree displayed in Example~\ref{exa.1.1} is
	\begin{figure}[H]
	\centering
	\begin{tikzpicture}[tree,level distance=1.3cm,outer sep=-5pt]
	\tikzstyle{level 1}=[sibling distance=6.5cm]
	\tikzstyle{level 2}=[sibling distance=4cm]
	\node {(S\to AD)}
		child {node {(A\to AA)}
			child {node {(A\to bAa)}
				child {node {(A\to \lambda)}}}
			child {node {(A\to aAb)}
				child {node {(A\to \lambda)}}}}
		child {node {(D\to Ddd)}
			child {node {(D\to d)}}};
	\end{tikzpicture}
	\end{figure}
\noindent
Note that this tree is obtained from the other one by viewing the building blocks (trees of height one) of the local tree as the nodes of the rule tree.
\end{example}
The following theorem shows the relationship of the rule languages to those defined before.
\begin{theorem}\label{the.3.57}
	The class of rule tree languages is properly included in the intersection of the class of local tree languages and the class of tree languages recognizable by a det.\ top-down fta.
\end{theorem}
\begin{proof}
	We first show inclusion in the class of local tree languages. Let $G=(N,\Sigma,R,S)$ be a context-free grammar, and $\mbar{R}=\{\mbar{r}\mid r\in R\}$. Consider the context-free grammar $G_1=(\mbar{R}-\mbar{R}_0,\mbar{R}_0,P,-)$ where $P$ is defined as follows: 
if $r=(A\to w_0\li[A_iw_i][k][][\cdots])$, $k\geq1$, is in $R$, 
then $\mbar{r}\to\mbar{r}_1\cdots\mbar{r}_k$ is in $P$ for all rules $\li[r_i]\in R$ such that the left hand side of $r_i$ is $A_i$ ($1\leq i\leq k$). Let $\mbar{V}=\{\mbar{r}\mid r\in R\text{ has left hand side }S\}$. Then $\mathit{RT}(G)=\bigcup\limits_{\alpha\in\mbar{V}}D_{G_1}^\alpha$, and hence $\mathit{RT}(G)$ is local.
	
	To show that $\mathit{RT}(G)$ can be recognized by a deterministic top-down fta, consider $M=(Q,\mbar{R},\delta,q_0,F)$, where $Q=N\cup \{W\}$, $q_0=S$, for $\mbar{r}\in\mbar{R}_0$, $F_{\mbar{r}}$ consists of the left hand side of $r$ only, and for $\mbar{r}\in\mbar{R}_k$, $r$ of the form $A\to w_0A_1w_1A_2w_2\cdots A_kw_k$, $\delta_{\mbar{r}}(A)=(\li[A_i])$ and $\delta_{\mbar{r}}(B)=(\li[W])$ for all other $B\in Q$. Then $L(M)=\mathit{RT}(G)$.
	
	To show proper inclusion, let $H$ be the context-free grammar with rules $S\to SS$, $S\to aS$, $S\to Sb$ and $S\to ab$. Then $D_H^S$ is a local tree language. It is easy to see that $D_H^S$ can be recognized by a det.\ top-down fta. Now suppose that $D_H^S=\mathit{RT}(G)$ for some context-free grammar $G$. Since $S$ has rank 2 and since the configuration
	\raisebox{-0.5\height}{\begin{tikzpicture}[lighttree]
	\node {S}
		child {node {S}}
		child {node {S}};
	\end{tikzpicture}}
	occurs in $D_H^S$, $S$ is the name of a rule of $G$ of the form $A\to w_0Aw_1Aw_2$. Now, since $a$ and $b$ are of rank 0 and since
	\raisebox{-0.5\height}{\begin{tikzpicture}[lighttree]
	\node {S}
		child {node {a}}
		child {node {b}};
	\end{tikzpicture}}
	is in $D_H^S$, $a$ and $b$ are names of rules $A\to w_3$ and $A\to w_4$. Hence $S[ba]$ is a rule tree of $G$. Contradiction.
\end{proof}
We now characterize the recognizable tree languages in terms of rule tree languages.

\begin{theorem}\label{the.3.58}
	Every recognizable tree language is the projection of a rule tree language.
\end{theorem}
\vspace*{-1cm}
\begin{proof}
	Let $G=(N,\Sigma,R,S)$ be a regular tree grammar in normal form. We shall define a regular tree grammar $\mbar{G}$ and a projection $p$ such that $L(G)=p(L(\mbar{G}))$ and $L(\mbar{G})$ is a rule tree language.
	$\mbar{G}$ will simulate $G$, but $\mbar{G}$ will put all information about the rules, applied during the derivation of a tree $t$, into the tree itself. This is a useful technique.
	
	Let $\mbar{R}$ be a set of symbols in one-to-one correspondence with $R$, and let $\mbar{G}=(N,\mbar{R},P,S)$. The ranking of $\mbar{R}$, the set $P$ of rules and the projection $p$ are defined simultaneously as follows:
	\begin{enumerate}[label=(\roman*)]
		\item if $r\in R$ is the rule $A\to\listt[a][B_i]$, then $\mbar{r}$ has rank $k$, $A\to\listt[\mbar{r}][B_i]$ is in $P$ and $p_k(\mbar{r})=a$;
		\item if $r\in R$ is the rule $A\to a$, then $\mbar r$ has rank 0, $A\to\mbar{r}$ is in $P$ and $p_0(\mbar{r})=a$.
	\end{enumerate}
	It is obvious that $p(L(\mbar{G}))=L(G)$. Now note that $G$ may be viewed as a context-free grammar (over $\Sigma\cup\{[\,,]\}$). In fact, $\mbar{G}$ is the same as the one constructed in Definition~\ref{def.3.54}! Thus $L(\mbar{G})$ is a rule tree language.
\end{proof}
Since \RECOG{} is closed under projections (Theorem \ref{the.3.48}), we now easily obtain the following corollary.
\begin{corollary}\label{cor.3.59}
	For each tree language $L$ the following four statements are equivalent:
	\begin{enumerate}[label=(\roman*)]
		\item $L$ is recognizable
		\item $L$ is the projection of a rule tree language
		      \item\label{cor.3.59ii} $L$ is the projection of a local tree language
		\item $L$ is the projection of a tree language recognizable by a det.\ top-down fta.\qed
	\end{enumerate}
\end{corollary}
\begin{exercise}\label{exe.3.60}
	Show that, in the case of local tree languages, the projection involved in the above corollary \ref{cor.3.59ii} can be taken as the identity on symbols of rank 0 (thus the yields are preserved).
\end{exercise}
As a final operation on trees we consider the notion of tree homomorphism. For strings, a homomorphism $h$ associates a string $h(a)$ with each symbol $a$ of the alphabet, and transforms a string $\li*[a_i][n][][\cdots]$ into the string $h(a_1)\cdot h(a_2)\cdots h(a_n)$.
Generalizing this to trees, a tree homomorphism $h$ associates a tree $h(a)$ with each symbol $a$ of the ranked alphabet (actually, one tree for each rank). The application of $h$ to a tree $t$ consists in replacing each symbol $a$ of $t$ by the tree $h(a)$, and tree concatenating all the resulting trees. Note that, if $a$ is of rank $k$, then $h(a)$ should be tree concatenated with $k$ other trees; therefore, since tree concatenation happens \uem{at} symbols of rank 0, the tree $h(a)$ should contain at least $k$ different symbols of rank 0. Since, in general, the number of symbols of rank 0 in some alphabet may be less than the rank of some other symbol, we allow for the use of an arbitrary number of auxiliary symbols of rank 0, called ``variables" (recall the use of nonterminals as auxiliary symbols of rank 0 in Theorem \ref{the.3.43}).
\begin{definition}\label{def.3.61}
	Let $x_1,x_2,x_3,\dots$ be an infinite sequence of different symbols, called \uem{variables}. Let $\X=\{x_1,x_2,x_3,\dots\}$, for $k\geq1$, $\X_k=\{\li*[x_i]\}$, and $\X_0=\emptyset$. Elements of $\X$ will also be denoted by $x$, $y$ and $z$.
\end{definition}
\begin{definition}\label{def.3.62}
	Let $\Sigma$ and $\Delta$ be ranked alphabets. A \uem{tree homomorphism} $h$ is a family $\{h_k\}_{k\geq0}$ of mappings $h_k:\Sigma_k\to T_\Delta(\X_k)$.
	A tree homomorphism determines a mapping $h:T_\Sigma\to T_\Delta$ as follows:
	\begin{enumerate}[label=(\roman*)]
		\item for $a\in\Sigma_0$, $h(a)=h_0(a)$;
		\item for $k\geq1$, $a\in\Sigma_k$ and $\li\in T_\Sigma$,\\[1mm]
		      $h(\listt)=h_k(a)\langle\li[x_i\gets h(t_i)]\rangle$.
	\end{enumerate}
	In the particular case that, for each $a\in\Sigma_k$, $h_k(a)$ does not contain two occurrences of the same $x_i$ ($i=1,2,3,\dots$), $h$ is called a \uem{linear} tree homomorphism.
\end{definition}
A general tree homomorphism $h$ has the abilities of \uwave{deleting} ($h_k(a)$ does not contain $x_i$), \uwave{copying} ($h_k(a)$ contains $\geq2$ occurrences of $x_i$) and \uwave{permuting} (if $i<j$, then $x_j$ may occur before $x_i$ in $h_k(a)$) subtrees. Moreover, at each node, it can add pieces of tree (the frontier of $h_k(a)$ need not to be an element of $\X^*$). A linear homomorphism cannot copy. Note that, to obtain the proper generalization of the monadic case, one should also forbid deletion, and require that $h_0(a)=a$ for all $a\in\Sigma_0$ (moreover, no pieces of tree should be added).
\begin{exercise}\label{exe.3.63}
	Let $\Sigma_0=\{a,b\}$ and $\Sigma_2=\{p\}$. Consider the tree homomorphism $h$ such that $h_0(a)=a$, $h_0(b)=b$ and $h_2(p)= p[x_2x_1] $. Show that, for every $t$ in $T_\Sigma$, $\fyield(h(t))$ is the mirror image of $\fyield(t)$.
\end{exercise}
It is easy to see that recognizable tree languages are not closed under arbitrary tree homomorphisms; they are closed under linear tree homomorphisms. This is shown in the following two theorems.
\begin{theorem} \label{the.3.64}
	\RECOG{} is not closed under arbitrary tree homomorphisms.
\end{theorem}
\begin{proof}
	Let $\Sigma_0=\{a\}$ and $\Sigma_1=\{b\}$. Let $h$ be the tree homomorphism defined by $h_0(a)=a$ and $h_1(b)= b[x_1x_1] $. Consider the recognizable tree language $T_\Sigma$. It is easy to prove that $\fyield(h(T_\Sigma))=\{a^{2^n}\mid n\geq0\}$. Since $\{a^{2^n}\mid n\geq0\}$ is not a context-free language, Theorem \ref{the.3.28} implies that $h(T_\Sigma)$ is not recognizable.
\end{proof}
\begin{theorem}\label{the.3.65}
	\RECOG{} is closed under linear tree homomorphisms.
\end{theorem}
\begin{proof}
	The idea of the proof is obvious. Given some regular tree grammar generating a recognizable tree language, we change the right hand sides of all rules into their homomorphic images. The resulting grammar generates homomorphic images of ``{}sentential forms" of the original grammar (note that this wouldn't work in the nonlinear case). The only thing we should worry about is that the homomorphism may be deleting. In that case superfluous rules in the original grammar might be transformed into useful rules in the new grammar. This is solved by requiring that the original grammar does not contain any superfluous rule. The formal construction is as follows.
	
	Let $G=(N,\Sigma,R,S)$ be a regular tree grammar in normal form, such that for each nonterminal $A$ there is at least one $t\in T_\Sigma$ such that $A\xRightarrow{*}t$ (since $G$ is a context-free grammar, it is well known that each regular tree grammar has an equivalent one satisfying this condition). Let $h$ be a homomorphism from $T_\Sigma$ into $T_\Delta$, for some ranked alphabet $\Delta$. Extend $h$ to trees in $T_\Sigma(N)$ by defining $h_0(A)=A$ for all $A$ in $N$. Thus $h$ is now a homomorphism from $T_\Sigma(N)$ into $T_\Delta(N)$. Construct the regular tree grammar $H=(N,\Delta,\mbar{R},S)$, where $\mbar{R}=\{A\to h(t)\mid A\to t\text{ is in }R\}$.
	
	To show that $L(H)=h(L(G))$ we shall prove that
	\begin{enumerate}[label=(\arabic*)]
		\item if $A\xRightarrow[G]{*}t$, then $A\xRightarrow[H]{*}h(t)$ \hfill{($t\in T_\Sigma$); and}
		\item if $A\xRightarrow[H]{*}s$, then there exists $t$ such that\\
		\hphantom{ if $A\xRightarrow[H]{*}s$,}$h(t)=s$ and $A\xRightarrow[G]{*}t$ \hfill{($s\in T_\Delta$, $t\in T_\Sigma$).}
	\end{enumerate}
	Let us give the straightforward proofs as detailed as possible.
	
	(1) The proof is by induction on the number of steps in the derivation $A\xRightarrow[G]{*}t$. If, in one step, $A\xRightarrow[G]{}t$, then $t\in\Sigma_0$ and $A\to t$ is in $R$. Hence $A\to h(t)$ is in $\mbar{R}$, and so $A\xRightarrow[H]{*}h(t)$. Now suppose that the first step in $A\xRightarrow[G]{*}t$ results from the application of a rule of the form $A\to\listt[a][B_i]$. Then $A\xRightarrow[G]{*}t$ is of the form $A\xRightarrow[G]{}\listt[a][B_i]\xRightarrow[G]{*}t$. It follows that $t$ is of the form $\listt$ such that $B_i\xRightarrow[G]{*}t_i$ for all $1\leq i\leq k$. Hence, by induction, $B_i\xRightarrow[H]{*}h(t_i)$. Now, since the rule $A\to h_k(a)\langle\li[x_i\gets B_i]\rangle$ is in $\mbar{R}$ by definition, we have (prove this!)
	\begin{align*}
		A\xRightarrow[H]{} h_k(a)\langle\li[x_i\gets B_i]\rangle\xRightarrow[H]{*} & \;h_k(a)\langle\li[x_i\gets h(t_i)]\rangle \\
		                                                                  & =h(\listt)=h(t).                               
	\end{align*}
	
	(2) The proof is by induction on the number of steps in $A\xRightarrow[H]{*}s$. For zero steps the statement is trivially true. Suppose that the first step in $A\xRightarrow[H]{*}s$ results from the application of a rule $A\to h_0(a)$ for some $a$ in $\Sigma_0$. Then $h(a)=s$ and $A\xRightarrow[G]{*}a$. Now suppose that the first step results from the application of a rule $A\to h_k(a)\langle\li[x_i\gets B_i]\rangle$, where $A\to\listt[a][B_i]$ is a rule of $G$. Then the derivation is $A\xRightarrow[H]{}h_k(a)\langle\li[x_i\gets B_i]\rangle\xRightarrow[H]{*}s$. At this point we need both linearity of $h$ (to be sure that each $B_i$ in $h_k(a)\langle\li[x_i\gets B_i]\rangle$ produces at most one subtree of $s$) and the condition on $G$ (to deal with deletion: since $h_k(a)\langle\li[x_i\gets B_i]\rangle$ need not contain an occurrence of $B_i$, we need an arbitrary tree generated, in $G$, by $B_i$ to be able to construct the tree $t$ such that $h(t)=s$). There exist trees $\li[s_i]$ in $T_\Delta$ such that $s=h_k(a)\langle\li[x_i\gets s_i]\rangle$ and
	\begin{enumerate}[label=(\roman*)]
		\item if $x_i$ occurs in $h_k(a)$, then $B_i\xRightarrow[H]{*}s_i$;
		\item if $x_i$ does not occur in $h_k(a)$, then $s_i=h(t_i)$ for some arbitrary $t_i$ such that $B_i\xRightarrow[G]{*}t_i$.
	\end{enumerate}
	Hence, by induction and (ii), there are trees $\li$ such that $h(t_i)=s_i$ and $B_i\xRightarrow[G]{*}t_i$ for all $1\leq i\leq k$. Consequently, if $t=\listt$, then $A\xRightarrow[G]{}\listt[a][B_i]\xRightarrow[G]{*}\listt=t$, and $h(t)=h_k(a)\langle\li[x_i\gets h(t_i)]\rangle=s$.
	
	This proves the theorem.
\end{proof}
\begin{exercise}\label{exe.3.66}
	In the string case one can also prove that the regular languages are closed under homomorphisms by using the Kleene characterization theorem. Give an alternative proof of Theorem \ref{the.3.65} by using Theorem \ref{the.3.43} (use the fact, which is implicit in the proof of that theorem, that each regular tree language over the ranked alphabet $\Sigma$ can be built up from finite tree languages using operations $\cup$, $\conc_A$ and $^{*A}$, \dashuline{where $A\notin\Sigma$}).
\end{exercise}
As an indication how one could use theorems like Theorem \ref{the.3.65}, we prove the following theorem, which is (slightly!) stronger than Theorem \ref{the.3.28}.
\begin{theorem}\label{the.3.67}
	Each context-free language over $\Delta$ is the yield of a recognizable tree language over $\Sigma$, where $\Sigma_0=\Delta\cup\{e\}$ and $\Sigma_2=\{*\}$.
\end{theorem}
\begin{proof}
	Let $L$ be a context-free language over $\Delta$. By Theorem \ref{the.3.28}, there is a recognizable tree language $U$ over some ranked alphabet $\Omega$ with $\Omega_0=\Delta\cup\{e\}$, such that $\fyield(U)=L$. Let $h$ be the linear tree homomorphism from $T_\Omega$ into $T_\Sigma$ such that
	\begin{enumerate}[label=]
		\item $h_0(a)=a$ for all $a$ in $\Delta\cup\{e\}$,
		\item $h_1(a)=x_1$ for all $a$ in $\Omega_1$, and
		\item $h_k(a)=*[x_1*[x_2*[\cdots*[x_{k-1}x_k]\cdots]]]$ for all $a$ in $\Omega_k$, $k\geq2$.
	\end{enumerate}
	By Theorem \ref{the.3.65}, $h(U)$ is a recognizable tree language over $\Sigma$. It is easy to show that, for each $t$ in $T_\Omega$, $\fyield(h(t))=\fyield(t)$. Hence $\fyield(h(U))=\fyield(U)=L$.
\end{proof}
Note that Theorem \ref{the.3.67} is ``equivalent" to the fact that each context-free language can be generated by a context-free grammar in Chomsky normal form.
\begin{exercise}\label{exe.3.68}
	Try to show that \RECOG{} is closed under inverse (not necessarily linear) homomorphisms; that is, if $L\in\RECOG$ and $h$ is a tree homomorphism, then $h^{-1}(L)=\{t\mid h(t)\in L\}$ is recognizable. (Represent $L$ by a det.\ bottom-up fta).
\end{exercise}
We have now discussed all \AFL{} operations (see \cite[\rom{4}]{Sal}) generalized to trees: union, tree concatenation, tree concatenation closure, (linear) tree homomorphism, inverse tree homomorphism and intersection with a recognizable tree language. Thus, according to previous results, \RECOG{} is a ``tree \AFL".
\begin{exercise}\label{exe.3.69}
	Generalize the operation of string substitution (see Definition \ref{def.2.29}) to trees, and show that \RECOG{} is closed under ``linear tree substitution".
\end{exercise}
\begin{exercise}\label{exe.3.70}
	Suppose you don't know about context-free grammars. Consider the notion of regular tree grammar. Give a recursive definition of the relation $\xRightarrow{}$ for such a grammar. Show that, if $\listt\xRightarrow{*}s$, then there are $\li[s_i]$ such that $s=\listt[a][s_i]$ and $t_i\xRightarrow{*}s_i$ for all $1\leq i\leq k$. Which of the two definitions of regular tree grammar do you prefer?
\end{exercise}

\subsection{Decidability}
Obviously the membership problem for recognizable tree languages is solvable: given a tree $t$ and an fta $M$, just feed $t$ into $M$ and see whether $t$ is recognized or not.

We now want to prove that the emptiness and finiteness problems for recognizable tree languages are solvable. To do this we generalize the pumping lemma for regular string languages to recognizable tree languages: for each regular string language $L$ there is an integer $p$ such that for all strings $z$ in $L$, if $|z|\geq p$, then there are strings $u$, $v$ and $w$ such that $z=uvw$, $|vw|\leq p$, $|v|\geq1$ and, for all $n\geq0$, $uv^nw\in L$.
\begin{theorem}\label{the.3.71}
	Let $\Sigma$ be a ranked alphabet, and $x$ a symbol not in $\Sigma$. For each recognizable tree language $L$ over $\Sigma$ we can find an integer $p$ such that for all trees $t$ in~$L$, if $\fheight(t)\geq p$, then there are trees $u,v,w\in T_\Sigma(\{x\})$ such that
	\begin{enumerate}[label=(\roman*)]
		\item\label{the.3.71i} $u$ and $v$ contain exactly one occurrence of $x$, and $w\in T_\Sigma$;
		\item $t=u\conc_xv\conc_xw$;
		\item $\fheight(v\conc_xw)\leq p$;
		\item $\fheight(v)\geq1$, and
		\item for all $n\geq0$, $u\conc_xv^{nx}\conc_xw\in L$, 
                                where $v^{nx}=\li*[v][n][\conc_x][\cdots]$ ($n$ times).
	\end{enumerate}
\end{theorem}
\vspace{-1cm}
\begin{proof}
	Let $M=(Q,\Sigma,\delta,s,F)$ be a deterministic bottom-up fta recognizing $L$. Let $p$ be the number of states of $M$. Consider a tree $t\in L(M)$ such that $\fheight(t)\geq p$. Considering some path of maximal length through $t$, it is clear that there are trees $\li*[t_i][n]\in T_\Sigma(\{x\})$ such that $n\geq p+1$, $t=t_n\conc_xt_{n-1}\conc_x\cdots\conc_xt_2\conc_xt_1$, the trees $t_2,\dots, t_n$ contain exactly one occurrence of $x$ and have $\fheight\geq1$, and $t_1\in T_\Sigma$ (this is a ``linearization" of $t$ according to some path). Now consider the states $q_i=\delbu(t_i\conc_x\cdots\conc_xt_1)$ for $1\leq i\leq n$. Then, among $\li[q_i][p+1]$ there are two equal states: there are $i,j$ such that $q_i=q_j$ and $1\leq i<j\leq p+1$. Let $u=t_n\conc_x\cdots\conc_xt_{j+1}$, $v=t_j\conc_x\cdots\conc_xt_{i+1}$ and $w=t_i\conc_x\cdots\conc_xt_1$. Then requirements (i)-(iv) in the statement of the theorem are obviously satisfied. Furthermore, in general, if $\delbu(s_1)=\delbu(s_2)$, then $\delbu(s\conc_xs_1)=\delbu(s\conc_xs_2)$. Hence, since $\delbu(v\conc_xw)=\delbu(w)$, requirement (v) is also satisfied.
\end{proof}
As a corollary to Theorem \ref{the.3.71} we obtain the pumping lemma for context-free languages.
\begin{corollary}\label{cor.3.72}
	For each context-free language $L$ over $\Delta$ we can find an integer $q$ such that for all strings $z\in L$, if $|z|\geq q$, then there are strings $u_1$, $v_1$, $w_0$, $v_2$ and $u_2$ in $\Delta^*$, such that $z=u_1v_1w_0v_2u_2$, $|v_1w_0v_2|\leq q$, $|v_1v_2|>0$, and, for all $n\geq0$, $u_1v_1^nw_0v_2^nu_2\in L$.
\end{corollary}
\begin{proof}
	By Theorem \ref{the.3.67}, and the fact that each context-free language can be generated by a $\lambda$-free context-free grammar, there is a recognizable tree language $U$ over $\Sigma$ such that $\fyield(U)=L$, where $\Sigma_0=\Delta$ and $\Sigma_2=\{*\}$. Let $p$ be the integer corresponding to $U$ according to Theorem \ref{the.3.71}, and put $q=2^p$. Obviously, if $z\in L$ and $|z|\geq q$, then there is a $t$ in $U$ such that $\fyield(t)=z$ and $\fheight(t)\ge p$. Then, by Theorem \ref{the.3.71} there are trees $u$, $v$ and $w$ such that (i)-(v) in that theorem hold. Thus $t=u\conc_xv\conc_xw$. Let $\fyield(u)=u_1xu_2$, $\fyield(v)=v_1xv_2$ and $\fyield(w)=w_0$ (see \ref{the.3.71i}). Then $z=\fyield(t)=\fyield(u\conc_xv\conc_xw)=\fyield(u)\conc_x\fyield(v)\conc_x\fyield(w)=u_1v_1w_0v_2u_2$. It is easy to see that all other requirements stated in the corollary are also satisfied.
\end{proof}
\begin{exercise}\label{exe.3.73}
	Let $\Sigma$ be a ranked alphabet such that $\Sigma_0=\{a,b\}$. 
Show that the tree language $\{t\in T_\Sigma\mid\fyield(t)\text{ has an equal number of $a$'s and $b$'s}\}$ is not recognizable.
\end{exercise}
From the pumping lemma the decidability of both emptiness and finiteness problem for \RECOG{} follows.
\begin{theorem}\label{the.3.74}
	The emptiness problem for recognizable tree languages is decidable.
\end{theorem}
\begin{proof}
	Let $L$ be a recognizable tree language, and let $p$ be the integer of Theorem \ref{the.3.71}. Obviously (using $n=0$ in point (v)), $L$ is nonempty if and only if $L$ contains a tree of $\fheight<p$.
\end{proof}
\begin{theorem}\label{the.3.75}
	The finiteness problem for recognizable tree languages is decidable.
\end{theorem}
\begin{proof}
	Let $L$ be a recognizable tree language, and let $p$ be the integer of Theorem \ref{the.3.71}. Obviously, $L$ is finite if and only if all trees in $L$ are of $\fheight<p$. Thus, $L$ is finite iff $L\cap\{t\mid\fheight(t)\geq p\}=\emptyset$. Since $L\cap\{t\mid\fheight(t)\geq p\}$ is recognizable (Exercise \ref{exe.3.26} and Theorem \ref{the.3.32}), this is decidable by the previous theorem.
\end{proof}
Note that the decidability of emptiness and finiteness problem for context-free languages follows from these two theorems together with the ``yield-theorem" (with $e\notin\Sigma_0$, $\Sigma_1=\emptyset$).

As in the string case we now obtain the decidability of inclusion of recognizable tree languages (and hence of equality).
\begin{theorem}\label{the.3.76}
	It is decidable, for arbitrary recognizable tree languages $U$ and $V$, whether $U\subseteq V$ (and also, whether $U=V$).
\end{theorem}
\begin{proof}
	Since $U$ is included in $V$ iff the intersection of $U$ with the complement of $V$ is empty, the theorem follows from Theorems \ref{the.3.32} and \ref{the.3.74}.
\end{proof}
Note again that each regular tree language is a special kind of context-free language. Note also that inclusion of context-free languages is not decidable (neither equality). Therefore it is nice that we have found a subclass of \CFL{} for which inclusion and equality are decidable. Note also that \CFL{} is not closed under intersection, but \RECOG{} is. We shall now relate these facts to some results in the literature concerning ``parenthesis languages" and ``{}structural equivalence" of  context-free grammars (see \cite[\rom{8}.3]{Sal}).
\begin{definition}\label{def.3.77}
	A \uem{parenthesis grammar} is a context-free grammar $G=(N,\Sigma\cup\{[\,,]\},$ $R,S)$ such that each rule in $R$ is of the form $A\to[w]$ with $A\in N$ and $w\in (\Sigma\cup N)^*$. The language generated by $G$ is called a \uem{parenthesis language}.
\end{definition}
To relate parenthesis languages to recognizable tree languages, let us restrict attention to ranked alphabets $\Delta$ such that, for $k\geq1$, if $\Delta_k\neq\emptyset$, then $\Delta_k=\{*\}$, where $*$ is a fixed symbol. Suppose that in our recursive definition of ``tree" we change $\listt$ into $\underset{a}{[}t_1\cdots t_k]$ (see Definition \ref{def.2.5} and Remark \ref{rem.2.35}). Then, obviously, all our results about the class \RECOG{} are still valid. Furthermore, since $*$ is the only symbol of rank $\geq1$, we may as well replace $\underset{*}{[}$ by $[$. In this way, each parenthesis language is in \RECOG{} (in fact, each parenthesis grammar is a regular tree grammar). It is also easy to see that, if $L$ is a recognizable tree language (over a restricted ranked alphabet $\Delta$), then $L-\Delta_0$ is a parenthesis language. From these considerations we obtain the following theorem.
\begin{theorem}\label{the.3.78}
	The class of parenthesis languages is closed under union, intersection and subtraction. The inclusion problem for parenthesis languages is decidable.
\end{theorem}
\begin{proof}
	The first statement follows directly from Theorem \ref{the.3.32} and the last remark above. The second statement follows directly from Theorem \ref{the.3.76}.
\end{proof}
A paraphrase of this theorem is obtained as follows.
\begin{definition}\label{def.3.79}
	For any ranked alphabet, let $p$ be the projection such that $p(a)=a$ for all symbols of rank 0, and $p(a)=*$ for all symbols of rank $\geq1$. Let $G$ be a context-free grammar. The \uem{bare tree language} of $G$, denoted by $\BT(G)$, is $p(D_G^S)$, where $S$ is the initial nonterminal of $G$. We say that two context-free grammars $G_1$ and $G_2$ are \uem{structurally equivalent} iff they generate the same bare tree language (i.e., $\BT(G_1)=\BT(G_2)$).
\end{definition}
Thus, $G_1$ and $G_2$ are structurally equivalent if their sets of derivation trees are the same after ``erasing" all nonterminals.
\begin{theorem}\label{the.3.80}
	It is decidable for arbitrary context-free grammars whether they are structurally equivalent.
\end{theorem}
\begin{proof}
	For any context-free grammar $G=(N,\Sigma,R,S)$, let $[G]$ be the parenthesis grammar $(N,\Sigma\cup\{[\,,]\},\mbar{R},S)$, where $\mbar{R}=\{A\to[w]\mid A\to w\text{ is in }R\}$. Obviously, $L([G])=\BT(G)$. Hence, by Theorem \ref{the.3.78}, the theorem holds.
\end{proof}
\begin{exercise}\label{exe.3.81}
	Show that, for any two context-free grammars $G_1$ and $G_2$ there exists a context-free grammar $G_3$ such that $\BT(G_3)=\BT(G_1)\cap\BT(G_2)$.
\end{exercise}
\begin{exercise}\label{exe.3.82}
	Show that each context-free grammar has a structurally equivalent context-free grammar that is invertible (cf. Exercise \ref{exe.3.29}).
\end{exercise}
\begin{exercise}\label{exe.3.83}
	Consider the ``bracketed context-free languages" of Ginsburg and Harrison, and show that some of their results follow easily from results about \RECOG{} (show first that each recognizable tree language is a deterministic context-free language).
\end{exercise}
\begin{exercise}\label{exe.3.84}
	Investigate whether it is decidable for an arbitrary recognizable tree language $R$
	\begin{enumerate}[label=(\roman*)]
		\item whether $R$ is local;
		\item whether $R$ is a rule tree language;
		\item whether $R$ is recognizable by a det.\ top-down fta.\qedhere
	\end{enumerate}
\end{exercise}

\section{Finite state tree transformations}\label{sec.4}
\subsection{Introduction: Tree transducers and semantics} \label{sec.4.1}
In this part we will be concerned with the notion of a tree transducer: a machine that takes a tree as input and produces another tree as output. In all generality we may view a tree transducer as a device that gives meaning to structured objects (i.e., a semantics defining device). Let us try to indicate this aspect of tree transducers. 

Consider a ranked alphabet $\Sigma$. The elements of $\Sigma$ may be viewed as ``operators", i.e., symbols denoting operations (functions of several arguments). The rank of an operator stands for the number of arguments of the operation (note therefore that one operator may denote several operations). The operators of rank $0$ have no arguments: they are (denote) constants. As an example, the ranked alphabet $\Sigma$ with $\Sigma_0 = \{ e,a,b \}$ and $\Sigma_2 = \{ f \}$ may be viewed as consisting of three constants $e$, $a$ and $b$ and one binary operator $f$. From operators we may form ``terms" or ``expressions", like for instance $f(a,f(e,b))$, or perhaps, denoting $f$ by $\ast$, $(a \ast (e \ast b))$. Obviously the terms are in one-to-one correspondence with the set $T_\Sigma$ of trees over $\Sigma$ . Thus the notions tree and term may be identified. Intuitively, terms denote structured objects, obtained by applying the operations to the constants.

Formally, meaning is given to operators and terms by way of an ``interpretation". An interpretation of $\Sigma$ consists of a ``domain" $B$, for each element $a \in \Sigma_0$ an element $h_0(a) $ of $B$, and for each $k\ge1$ and operator $a \in \Sigma_k$ an operation $h_k(a):B^k \to B$. An interpretation of $\Sigma$ is also called a ``$\Sigma$-algebra" or ``{}algebra of type $\Sigma$". An interpretation ($B, \{h_0(a)\}_{a\in\Sigma_0}, \{h_k(a)\}_{a \in \Sigma_k}$) determines a mapping $h: T_\Sigma \to B$ (giving an interpretation to each term as an element of $B$) as follows: 
\begin{enumerate}[label=(\roman*)]
	\item for $a\in\Sigma_0$, $h(a) = h_0(a)$;
	\item for $k\ge1$ and $a \in \Sigma_k$,\\
	$h(\listt) = h_k(a)(\li[h(t_i)])$.
\end{enumerate}
(Such a mapping is also called a ``homomorphism" from $T_\Sigma$ into $B$). Thus the meaning of a tree is uniquely determined by the meaning of its subtrees and the interpretation of the operator applied to these subtrees. In general we can say that the meaning of a structured object is a function of the meanings of its substructures, the function being determined by the way the object is constructed from its substructures.

As an example, an interpretation of the above-mentioned ranked alphabet $\Sigma = \{e,a,b,f\}$ might for instance consist of a group $B$ with unity $h_0(e)$, multiplication $h_2(f)$ and two specific elements $h_0(a)$ and $h_0(b)$. Or it might consist of $B = \{a,b\}^\ast$, $h_0(e) = \lambda$, $h_0(a) = a$, $h_0(b) = b$ and $h_2(f)$ is concatenation. Note that in this case the mapping $h: T_\Sigma \to B$ is the yield!

It is now easy to see that a deterministic bottom-up fta with input alphabet $\Sigma$ is nothing else but a $\Sigma$-algebra with a finite domain (its set of states). Such an automaton may therefore be used as a semantic defining device in case there are only a finite number of possible semantical values. Obviously, in general, one needs an infinite number of semantical values. However, it is not attractive to consider arbitrary infinite domains $B$ since this provides us with no knowledge about the structure of the elements of $B$. We therefore assume that the elements of $B$ are structured objects: trees (or interpretations of them). Thus we consider $\Sigma$-algebras with domain $T_\Delta$ for some ranked alphabet $\Delta$. Our complete semantics of $T_\Sigma$ may then consist of two parts: an interpretation of $T_\Sigma$ into $T_\Delta$ and an interpretation of $T_\Delta$ in some $\Delta$-algebra. The interpretation of $T_\Sigma$ into $T_\Delta$ may be realized by a tree transducer. 

An example of an interpretation of $T_\Sigma$ into $T_\Delta$ is 
the tree homomorphism of Definition~\ref{def.3.62}. In fact each tree $s \in T_\Delta(\X_k)$ may be viewed as an operation $\widetilde{s} : T_\Delta^k \to T_\Delta$, defined by \[\widetilde{s}(\li[s_i]) = s\langle \li[x_i\gets s_i]\rangle.\]
A tree homomorphism is then the same thing as an interpretation of $\Sigma$ with domain $T_\Delta$, where the allowable interpretations of the elements of $\Sigma$ are the mappings $\widetilde{s}$ above. Note that these interpretations are very natural, since the interpretation of a tree is obtained by ``{}applying a finite number of $\Delta$-operators to the interpretations of its subtrees". To show the relevance of tree homomorphisms (and therefore tree transducers in general) to the semantics of context-free languages we consider the following very simple example.

\begin{example}\label{exa.4.1}
	Consider a context-free grammar generating expressions by the rules $E \to E + T$, $T \to T \ast F$, $E \to a$, $T \to a$, $F \to a$ and $F \to (E)$. Suppose we want to translate each expression into the equivalent post-fix expression. To do this we consider the rule tree language corresponding to this grammar and apply to the rule trees in this language the tree homomorphism $h$ defined by
	$h_2(E \to E + T) =  E[x_1x_2+] $, $h_2(T \to T \ast F) =  T[x_1x_2\ast] $, $h_0(E \to a) =  E[a] $, $h_0(T \to a) =  T[a] $, $h_0(F \to a) =  F[a] $ and $h_1(F \to (E)) =  F[x_1] $. Then the rule tree corresponding to an expression is translated into a tree whose yield is the corresponding post-fix expression. For instance,
		\begin{align*}&
		\raisebox{-0.5\height}{\begin{tikzpicture}[lighttree, level distance=1.1cm, outer sep=-5pt]
			\tikzstyle{level 1}=[sibling distance=2.5cm]
			\tikzstyle{level 2}=[sibling distance=2cm]
			\node {(E\to E + T)}
			child {node {(E\to a)}}
			child {node {(T\to T \ast F)}
				child {node {(T\to a)}}
				child {node {(F\to a)}}};
			\end{tikzpicture}}&&
		\text{goes into~}&&
		\raisebox{-0.5\height}{\begin{tikzpicture}[lighttree, level distance=0.8cm]
			\tikzstyle{level 1}=[sibling distance=1.3cm]
			\tikzstyle{level 2}=[sibling distance=0.6cm]
			\node {E}
			child {node {E}
				child {node {a}}}
			child {node {T}
				child {node {T}
					child {node {a}}}
				child {node {F}
					child {node {a}}}
				child {node {\ast}}}
			child {node {+}};
			\end{tikzpicture}}\text{~,}&
		\end{align*}
	so that $a + a \ast a$ is translated into $aaa\ast +$. Note that, moreover, the transformed tree is the derivation tree of the post-fix expression in the context-free grammar with the rules $E \to ET+$, $T \to TF\ast$, $E \to a$, $T \to a$, $F \to a$ and $F \to E$. Instead of interpreting this derivation tree as its yield (the post-fix expression), one might also interpret it as, for instance, a sequence of machine code instructions, like ``load $a$; load $a$; load $a$; multiply; add".
\end{example} 

It is not difficult to see that the syntax-directed translation schemes of \cite[\rom{1}.3]{AU}
correspond in some way to linear, nondeleting homomorphisms working on rule tree languages.

To arrive at our general notion of tree transducer, we combine the finite tree automaton and the tree homomorphism into a ``tree homomorphism with states" or a ``finite tree automaton with output". This tree transducer will not any more be an interpretation of $T_\Sigma$ into $T_\Delta$, but involves a generalization of this concept (although, by replacing $T_\Delta$ by some other set, it can again be formulated as an interpretation of $T_\Sigma$). Two ideas occur in this generalization.

\begin{idea} \label{sta.4.2}
	The translation (meaning) of a tree may depend not only on the translation of its subtrees but also on certain properties of these subtrees. Assuming that these properties are recognizable (that is, the set of all trees having the property is in \RECOG), they may be represented as states of a (deterministic) bottom-up fta. Thus we can combine the deterministic bottom-up fta and the tree homomorphism by associating to each symbol $a \in \Sigma_k$ a mapping $f_a : Q^k \to Q \times T_\Delta(\X_k)$. This $f_a$ may be split up into two mappings $\delta_a : Q^k \to Q$ and $h_a : Q^k \to T_\Delta(\X_k)$. The $\delta$-functions determine a mapping $\delbu : T_\Sigma \to Q$, as for the bottom-up fta, and the $h$-functions determine an output-mapping $\widehat{h} : T_\Sigma \to T_\Delta$ by the formula (cf. the corresponding tree homomorphism formula): 
$$\widehat{h}(\listt) = h_a(\li[\delbu(t_i)]) \langle \li[x_i \gets \widehat{h} (t_i)]\rangle.$$ 
Thus our tree transducer works through the tree in a bottom-up fashion just like the bottom-up fta, but at each step, it produces output by combining the output trees, already obtained from the subtrees, into one new output tree. Note that, if we allow our bottom-up tree transducer to be nondeterministic, then the above formula for $\widehat{h}$ is intuitively wrong (we need ``deterministic substitution").
\end{idea}

\begin{idea}\label{sta.4.3}
	To obtain the translation of the input tree one may need several different translations of each subtree. Suppose that one needs $m$ different kinds of translation of each tree (where one of them is the ``main meaning" and the others are ``{}auxiliary meanings"), then these may be realized by $m$ states of the transducer, say $\li[q_i][m]$. The $i$\textsuperscript{th} translation may then be specified by associating to each $a \in \Sigma_k$ a tree $h_{q_i}(a)\in T_\Delta(Y_{m,k})$, where
	$Y_{m,k} = \{y_{i,j} \mid 1 \le i \le m, 1 \le j \le k\}$. 
The $i$\textsuperscript{th} translation of a tree $\listt$ may then be defined by the formula $$\widehat{h}_{q_i}(\listt) = h_{q_i}(a) \langle y_{r,s} \gets \widehat{h}_{q_r}(t_s)\rangle_{1 \le r \le m, 1 \le s \le k}.$$ 
Thus the $i$\textsuperscript{th} translation of a tree is expressed in terms of all possible translations of its subtrees. Realizing such a translation in a bottom-up fashion would mean that we should compute all $m$ possible translations of each tree in parallel, whereas working in a top-down way we know exactly from $h_{q_i}(a)$ which translations of which subtrees are needed (note that, in general, not all elements of $Y_{m,k}$ appear in $h_{q_i}(a)$). Therefore, such a translation seems to be realized best by a top-down tree transducer. We note that the generalized syntax-directed translation scheme of \cite[\rom{2}.9.3]{AU}
	corresponds to such a top-down tree transducer working on a rule tree language.
\end{idea}

As already indicated in Example \ref{exa.4.1}, tree transducers are of interest to the translation of context-free languages (in particular the context-free part of a programming language). For this reason we often restrict the tree transducer to a rule tree language, a local tree language or a recognizable tree language (the difference being slight: a projection). This restriction is also of interest from a linguistical point of view: a natural language may be described by a context-free set of kernel sentences to which transformations may be applied, working on the derivation trees (as for instance the transformation active~$\rightarrow$~passive). The language then consists of all transformations of kernel sentences. We note that if derivation tree $d_1$ of sentence $s_1$ is transformed into tree $d_2$ with yield $s_2$, then the sentence $s_2$ is said to have ``deep structure" $d_1$ and ``{}surface structure" $d_2$.

\subsection{Top-down and bottom-up finite tree transducers}\label{sec.4.2}
Since tree transducers define tree transformations (recall Definition \ref{def.2.10}), we start by recalling some terminology concerning relations. We note first that, for ranked alphabets $\Sigma$ and $\Delta$, we shall identify any mapping $f : T_\Sigma \to T_\Delta$ with the tree transformation $\{(s,t) \mid f(s) = t \}$, and we shall identify any mapping $f : T_\Sigma \to \mathcal{P}(T_\Delta)$ with the tree transformation $\{(s,t) \mid t \in f(s) \}$.

\begin{definition}\label{def.4.4}
	Let $\Sigma$, $\Delta$ and $\Omega$ be ranked alphabets. If $M_1 \subseteq T_\Sigma \times T_\Delta$ and $M_2 \subseteq T_\Delta \times T_\Omega$, then the \uem{composition} of $M_1$ and $M_2$, denoted by $M_1 \comp M_2$, is the tree transformation $\{(s,t) \in T_\Sigma \times T_\Omega \mid (s,u) \in M_1$ and $(u,t) \in M_2 $ for some $ u \in T_\Delta \}$. If $F$ and $G$ are classes of tree transformations, then $F \comp G$ denotes the class $\{ M_1 \comp M_2 \mid M_1 \in F$ and $M_2 \in G \}$.
\end{definition}

\begin{definition}\label{def.4.5}
	Let $M$ be a tree transformation from $T_\Sigma$ into $T_\Delta$. The \uem{inverse} of $M$, denoted by $M^{-1}$, is the tree transformation $\{(t,s) \in T_\Delta \times T_\Sigma \mid (s,t) \in M \}$.
\end{definition}

\begin{definition}\label{def.4.6}
	Let $M$ be a tree transformation and $L$ a tree language. The \uem{image} of $L$ under $M$, denoted by $M(L)$, is the tree language $M(L) = \{ t \mid (s,t) \in M$ for some $s$ in $L \}$. If $M$ is a tree transformation from $T_\Sigma$ into $T_\Delta$, then the \uem{domain} of $M$, denoted by $\dom(M)$, is $M^{-1}(T_\Delta)$, and the \uem{range} of $M$, denoted by $\text{range}(M)$, is $M(T_\Sigma)$.
\end{definition}

In Part (\ref{sec.3}) we already considered certain simple tree transformations: relabelings and tree homomorphisms.

\begin{notation}\label{not.4.7}
	We shall use \REL~to denote the class of all relabelings, \HOM~to denote the class of all tree homomorphisms, and \LHOM~to denote the class of linear tree homomorphisms. 
\end{notation}

Moreover we want to view each finite tree automaton as a simple ``checking" tree transducer, which, given some input tree, produces the same tree as output if it belongs to the tree language recognized by the fta, and produces no output if not. 

\begin{definition}\label{def.4.8}
	Let $\Sigma$ be a ranked alphabet. A tree transformation $R \subseteq T_\Sigma \times T_\Sigma$ is called a \uem{finite tree automaton restriction} if there is a recognizable tree language $L$ such that $R = \{ (t,t) \mid t \in L\}$. If $M$ is an fta, then we shall denote the finite tree automaton restriction $\{ (t,t) \mid t \in L(M) \}$ by $T(M)$. We shall use \FTA{} to denote the class of all finite tree automaton restrictions. 
\end{definition}

\begin{exercise}\label{exe.4.9}
	Prove that the classes of tree transformations \REL, \HOM~and \FTA~are each closed under composition. Show that \REL~and \FTA~are also closed under inverse. 
\end{exercise}

Before defining tree transducers we first discuss a very general notion of tree rewriting system that can be used to define both tree transducers and tree grammars. The reason to introduce these tree rewriting systems is that recursive definitions like those for the finite tree automata and tree homomorphisms tend to become very cumbersome when used for tree transducers, whereas rewriting systems are more ``machine-like" and therefore easier to visualize. To arrive at the notion of tree rewriting system we first generalize the notion of string rewriting system to allow for the use of rule ``schemes". Recall the set $\X$ of variables from Definition \ref{def.3.61}. 

\begin{definition}\label{def.4.10}
	A \uem{rewriting system with variables} is a pair $G = (\Delta, R)$, where $\Delta$ is an alphabet and $R$ a finite set of ``rule schemes". A \uem{rule scheme} is a triple $(v,w,D)$ such that, for some $k \ge 0$, $v$ and $w$ are strings over $\Delta \cup \X_k$ and $D$ is a mapping from $\X_k$ into $\mathcal{P}(\Delta^\ast)$. Whenever $D$ is understood, $(v,w,D)$ is denoted by $v \to w$. For $1 \le i \le k$, the language $D(x_i)$ is called the \uem{range} (or domain) of the variable $x_i$.
	
	A relation $\xRightarrow[G]{}$ on $\Delta^\ast$ is defined as follows. For strings $s,t \in \Delta^ \ast$, $s \xRightarrow[G]{} t$ if and only if there exists a rule scheme $(v,w,D) $ in $R$, strings $\li[\phi_i]$ in $\li[D(x_i)]$ respectively (where $\X_k$ is the domain of $D$), and strings $\alpha$ and $\beta$ in $\Delta^ \ast$ such that
	\begin{align*}
	s &= \alpha \cdot v \langle\li[x_i \gets \phi_i]\rangle \cdot \beta \hspace{1cm}\text{and}\\ 
	t &= \alpha \cdot w \langle\li[x_i \gets \phi_i]\rangle \cdot \beta\hspace{0.3cm}.
	\end{align*}
	As usual $\xRightarrow[G]{\ast}$ denotes the transitive-reflexive closure of $\xRightarrow[G]{}$.
\end{definition}

For convenience we shall, in what follows, use the word ``rule" rather than ``rule scheme".

Of course, in a rewriting system with variables, the ranges of the variables should be specified in some effective way (note that we would like the relation $\Rightarrow$ to be decidable). In what follows we shall only use the case that the variables range over recognizable tree languages.

\newpage
\begin{examples}\label{exa.4.11}~
	\begin{enumerate}[label=(\arabic*)]
		\item Consider the rewriting system with variables $G = ( \Delta, R)$, where $\Delta = \{a,b,c\}$ and $R$ consists of the one rule $ax_1c \to aax_1bcc$, where $D(x_1) = b^\ast$.
		Then, for instance, $aabbcc \Rightarrow aaabbbccc$ (by application of the ordinary rewriting rule $abbc \to aabbbcc$ obtained by substituting $bb$ for $x_1$ in the rule above). It is easy to see that $\{w \in \Delta^\ast \mid abc \xRightarrow{\ast} w \} = \{a^n b^n c^n \mid n \ge 1\}$.
		\item Consider the rewriting system with variables $G = (\Delta, R)$, where $\Delta = \{ [\,,],\ast,1 \}$ and $R$ consists of the rules $[x_1 \ast x_2 1] \to [x_1 \ast x_2 ] x_1$ and $[x_1 \ast 1] \to x_1$, where in both rules $D(x_1) = D(x_2) = 1^\ast$. It is easy to see that, for arbitrary $u,v,w \in 1^\ast$, $[u \ast v] \xRightarrow{\ast} w$ iff $w$ is the product of $u$ and $v$ (in unary notation).
		\item The two-level grammar used to describe Algol 68 may be viewed as a rewriting system with variables. The variables ( = meta notions) range over context-free languages, specified by the meta grammar.\qedhere
	\end{enumerate}
\end{examples}

By specializing to trees we obtain the notion of tree rewriting system. 

\begin{definition}\label{def.4.12}
	A rewriting system with variables $G = (\Delta,R) $ is called a \uem{tree rewriting system} if 
	\begin{enumerate}[label=(\roman*)]
		\item $\Delta = \Sigma \cup \{ [\,,] \}$ for some ranked alphabet $\Sigma$;
		\item for each rule $(v,w,D)$ in $R$, $v$ and $w$ are trees in $T_\Sigma(\X_k)$ and, for $1 \le i \le k$, $D(x_i) \subseteq T_\Sigma$ (where $\X_k$ is the domain of $D$).\qedhere
	\end{enumerate}
\end{definition}

It should be clear that, for a tree rewriting system $G = (\Sigma \cup \{ [\,,] \}, R)$, if $s \in T_\Sigma$ and $s \xRightarrow[G]{} t$, then $t \in T_\Sigma$. In fact, the application of a rule to a tree consists of replacing some piece in the middle of the tree by some other piece, where the variables indicate how the subtrees of the old piece should be connected to the new one. As an example, if we have a rule $a[b[x_1x_2]b[x_3d]] \to b[x_2a[x_1dx_2]]$, then the application of this rule to a tree $t$ (if possible) consists of replacing a subtree of $t$ of the form 
	\begin{align*}
	&
	\raisebox{-0.5\height}{
		\begin{tikzpicture}[shorten >=3pt, shorten <=2pt, execute at begin node=$,
		execute at end node=$,
		node/.style={inner sep=1},
		leaf/.style={isosceles triangle,draw,shape border rotate=90,inner sep=1, opacity=1},
		bbleaf/.style={shorten >=0pt, level distance=0.75cm},
		bleaf/.style={level distance=0.64cm,opacity=0}]
		\tikzstyle{level 1}=[level distance=0.6cm,sibling distance=2.5cm]
		\tikzstyle{level 2}=[level distance=1.2cm,sibling distance=0.8cm]
		\node[node] {a}
		child {node[node] {b} 
			child[bbleaf] {node[scale=0.01] {} child[bleaf] {node[leaf] {t_1}}}
			child[bbleaf] {node[scale=0.01] {} child[bleaf] {node[leaf] {t_2}}}}
		child {node[node] {b}
			child[bbleaf] {node[scale=0.01] {} child[bleaf] {node[leaf] {t_3}}}
			child[level distance=0.8cm] {node[node] {d}}}
		;
		\end{tikzpicture}}&
	\text{by~}&&
	\raisebox{-0.5\height}{
		\begin{tikzpicture}[shorten >=3pt, shorten <=2pt, execute at begin node=$,
		execute at end node=$,
		node/.style={inner sep=1},
		leaf/.style={isosceles triangle,draw,shape border rotate=90,inner sep=1, opacity=1},
		bbleaf/.style={shorten >=0pt, level distance=0.75cm},
		bleaf/.style={level distance=0.64cm,opacity=0}]
		\tikzstyle{level 1}=[level distance=0.6cm,sibling distance=2.3cm]
		\tikzstyle{level 2}=[level distance=1.2cm,sibling distance=0.8cm]
		\node[node] {b}
		child[bbleaf] {node[scale=0.01] {} child[bleaf] {node[leaf] {t_2}}}
		child {node[node] {a}
			child[bbleaf] {node[scale=0.01] {} child[bleaf] {node[leaf] {t_1}}}
			child[level distance=0.8cm] {node[node] {d}}
			child[bbleaf] {node[scale=0.01] {} child[bleaf] {node[leaf] {t_2}}}}
		;
		\end{tikzpicture}}\text{,}&
	\end{align*}
where $t_1$, $t_2$ and $t_3$ are in the ranges of $x_1$, $x_2$ and $x_3$. Thus $t$ is of the form $\alpha a [ b [t_1t_2]b[t_3d]] \beta$ and is transformed into $\alpha b [t_2a[t_1dt_2]]\beta$.

\begin{example}\label{exa.4.13}
	Let $\Sigma_0 = \{ a \}$, $\Sigma_1 = \{ b \}$, $\Delta_0 = \{ a \}$, $\Delta_2 = \{ b \}$, $\Omega_0 = \{a\}$, $\Omega_1 = \{\ast,b\}$ and $\Omega_2 = \{b\}$.
	
	\begin{enumerate}[label=(\roman*)]
		\item Consider the tree rewriting system $G = (\Omega \cup \{ [\,,] \},R)$, where $R$ consists of the rules \\ 
		\hspace*{10mm}$ a  \to  \ast[a] $, \\[1mm]
		\hspace*{10mm}$ b[\ast[x_1]]  \to  \ast[b[x_1x_1]] $, \\[1mm]
		and $D(x_1) = T_\Delta$. Then, for instance,

			\begin{center}
			\raisebox{-0.5\height}{\begin{tikzpicture}[lighttree]
				\node {b}
				child {node {b}
					child {node {a}}};
				\end{tikzpicture}}
			$\Rightarrow$
			\raisebox{-0.5\height}{\begin{tikzpicture}[lighttree]
				\node {b}
				child {node {b}
					child {node {\ast}
						child {node {a}}}};
				\end{tikzpicture}}
			$\Rightarrow$
			\raisebox{-0.5\height}{\begin{tikzpicture}[lighttree]
				\node {b}
				child {node {\ast}
					child {node {b}
						child {node {a}}
						child {node {a}}}};
				\end{tikzpicture}}
			$\Rightarrow$
			\raisebox{-0.5\height}{\begin{tikzpicture}[lighttree]
				\tikzstyle{level 2}=[sibling distance=1.2cm]
				\tikzstyle{level 3}=[sibling distance=0.7cm]
				\node {\ast}
				child {node {b}
					child {node {b}
						child {node {a}}
						child {node {a}}}
					child {node {b}
						child {node {a}}
						child {node {a}}}};
				\end{tikzpicture}}.
\end{center}

		It is easy to see that, if $h$ is the tree homomorphism from $T_\Sigma$ into $T_\Delta$ defined by $h_0(a) =  a $ and $h_1(b) =  b[x_1x_1] $, then, for $s \in T_\Sigma$ and $t\in T_\Delta$, $h(s) = t$ iff $s\xRightarrow{\ast}  \ast[t] $.
		
		\item Consider the tree rewriting system $G' = (\Omega \cup \{[\,,]\},R')$, where $R'$ consists of the rules \\
		\hspace*{10mm}$ \ast[b[x_1]]  \to  b[\ast[x_1]\ast[x_1]] $, \\[1mm]
		\hspace*{10mm}$ \ast[a]  \to  a $, \\[1mm]
		and $D(x_1) = T_\Sigma$. Then, for instance, 

			\begin{center}
			\raisebox{-0.5\height}{\begin{tikzpicture}[lighttree]
				\node {\ast}
				child {node {b}
					child {node {b}
						child {node {a}}}};
				\end{tikzpicture}}
			$\Rightarrow$
			\raisebox{-0.5\height}{\begin{tikzpicture}[lighttree]
				\node {b}
				child {node {\ast}
					child {node {b}
						child {node {a}}}}
				child {node {\ast}
					child {node {b}
						child {node {a}}}};
				\end{tikzpicture}}
			$\Rightarrow$
			\raisebox{-0.5\height}{\begin{tikzpicture}[lighttree]
				\tikzstyle{level 2}=[sibling distance=0.4cm]
				\node {b}
				child {node {b}
					child {node {\ast}
						child {node {a}}}
					child {node {\ast}
						child {node {a}}}}
				child {node {\ast}
					child {node {b}
						child {node {a}}}};
				\end{tikzpicture}}
			$\Rightarrow$
			\raisebox{-0.5\height}{\begin{tikzpicture}[lighttree]
				\tikzstyle{level 2}=[sibling distance=0.4cm]
				\node {b}
				child {node {b}
					child {node {\ast}
						child {node {a}}}
					child {node {a}}}
				child {node {\ast}
					child {node {b}
						child {node {a}}}};
				\end{tikzpicture}}
			$\xRightarrow{\ast}$
			\raisebox{-0.5\height}{\begin{tikzpicture}[lighttree]
				\tikzstyle{level 2}=[sibling distance=0.4cm]
				\node {b}
				child {node {b}
					child {node {a}}
					child {node {a}}}
				child {node {b}
					child {node {a}}
					child {node {a}}};
				\end{tikzpicture}}.
\end{center}

		It is easy to see that, if $h$ is the homomorphism defined above, then, for $s \in T_\Sigma$ and $t\in T_\Delta$, $h(s) = t$ iff $ \ast[s]  \xRightarrow{*} t$.\qedhere
	\end{enumerate}
\end{example}

The tree transducers to be defined will be a generalization of the generalized sequential machine working on strings, which is essentially a finite automaton with output. 

A (nondeterministic) generalized sequential machine is a 6-tuple $M = (Q,\Sigma,\Delta,\delta,S,F)$, where $Q$ is the set of states, $\Sigma$ is the input alphabet, $\Delta$ the output alphabet, $\delta$ is a mapping $Q \times \Sigma \to \mathcal{P}(Q\times\Delta^\ast)$, $S$ is a set of initial states and $F$ a set of final states. Intuitively, if $\delta(q,a)$ contains $(q',w)$ then, in state $q$ and scanning input symbol $a$, the machine $M$ may go into state $q'$ and add $w$ to the output. Formally we may define the functioning of $M$ in several ways. As already said, the recursive definition (as for the fta) is too cumbersome, although it is the most exact one (and should be used in very formal proofs).
The other way is to describe the sequence of configurations the machine goes through during the translation of the input string. A configuration is usually a triple $(v,q,s)$, where $v$ is the output generated so far, $q$ is the state and $s$ is the rest of the input. If $s = a s_1$, then the next configuration might be $(vw,q',s_1)$. A useful variation of this is to replace $(v,q,s)$ by the string $vqs \in \Delta^\ast Q \Sigma^\ast$. The next configuration can now be obtained by applying the string rewriting rule $qa \to wq'$, thus $vqas_1 \Rightarrow vwq's_1$. Replacing $\delta$ by a corresponding set of rewriting rules, the string translation realized by $M$ can be defined as $\{ (v_1,v_2) \mid q_0v_1 \xRightarrow{\ast} v_2q_f$ for some $q_0 \in S$ and $q_f \in F \}$.

Let us first consider the bottom-up generalization of this machine to trees, which is conceptually easier than the top-down version, although perhaps less interesting. The bottom-up finite tree transducer goes through the input tree in the same way as the bottom-up fta, at each step producing a piece of output to which the already generated output is concatenated. The transducer arrives at a node of rank $k$ with a sequence of $k$ states and a sequence of $k$ output trees (one state and one output tree for each direct subtree of the node). The sequence of states and the label at the node determine (nondeterministically) a new state and a piece of output containing the variables $\li[x_i]$. The transducer processes the node by going into the new state and computing a new output tree by substituting the $k$ output trees for $\li[x_i]$ in the piece of output. There should be start states and output for each node of rank 0. If the transducer arrives at the top of the tree in a final state, then the computed output tree is the transformation of the input tree. (Cf. the story in Idea~\ref{sta.4.2}).

To be able to put the states of the transducer as labels on trees, we make them into symbols of rank 1. The configurations of the bottom-up tree transducer will be elements of $T_\Sigma(Q[T_\Delta])$,\;\footnote{Note that $Q[T_\Delta] = \{q[t] \mid q \in Q$ and $t \in T_\Delta \}$.} and the steps of the tree transducer (including the start steps) are modelled by the application of tree rewriting rules to these configurations. 

We now give the formal definition.

\begin{definition}\label{def.4.14}
	A \uem{bottom-up (finite) tree transducer} is a structure \\
$M = (Q,\Sigma,\Delta, R, Q_d)$, where\\
	\begin{tabular}{cp{0.85\linewidth}}
		$Q$      & is a ranked alphabet (of \uem{states}), such that all elements of $Q$ have rank 1 and no other ranks;\\
		$\Sigma$ & is a ranked alphabet (of \uem{input symbols});\\
		$\Delta$ & is a ranked alphabet (of \uem{output symbols}),\\
		& $Q \cap (\Sigma \cup \Delta) = \emptyset$;\\
		$Q_d$	 & is a subset of $Q$ (the set of \uem{final states}); and\\
		$R$  & is a finite set of \uem{rules} of one of the forms (i) or (ii): \\
	\end{tabular}
	\begin{enumerate}[leftmargin=1.8cm,label=(\roman*)]
		\item $a \to q[t]$, where $a \in \Sigma_0$, $q\in Q$ and $t \in T_\Delta$;
		\item $\listt[a][q_i[x_i]] \to q[t]$, where $k \ge 1$, $a \in \Sigma_k$, $\li[q_i], q \in Q$ and $t\in T_\Delta(\X_k)$.
	\end{enumerate}
	$M$ is viewed as a tree rewriting system over the ranked alphabet $Q \cup \Sigma \cup \Delta$ with $R$ as the set of rules, such that the range of each variable occurring in $R$ is $T_\Delta$. Therefore the relations $\xRightarrow[M]{}$ and $\xRightarrow[M]{\ast}$ are well defined according to Definition \ref{def.4.10}. 
	
	The \uem{tree transformation realized by $M$}, denoted by $T(M)$ or simply $M$, is 
	\[\{ (s,t) \in T_\Sigma \times T_\Delta \mid s \xRightarrow[M]{\ast} q[t]\text{ for some $q$ in }Q_d \}.\qedhere\] 
\end{definition}

We shall abbreviate ``(finite) tree transducer" by ``ftt".

\begin{remark}\label{rem.4.15}
	Note that $T(M)$ is also denoted by $M$. In general, we shall often make no distinction between a tree transducer and the tree transformation it realizes. Hopefully this will not lead to confusion.
\end{remark}

\begin{definition}\label{def.4.16}
	The class of tree transformations realized by bottom-up ftt will be denoted by \B. An element of \B{} will be called a \uem{bottom-up tree transformation}. 
\end{definition}

\begin{example}\label{exa.4.17}
	An example of a bottom-up ftt realizing a homomorphism was given in Example \ref{exa.4.13}(i) (it had one state $\ast$).
\end{example}

\begin{example}\label{exa.4.18}
	Consider the bottom-up ftt $M = (Q,\Sigma,\Delta,R,Q_d)$, where $Q = \{ q_0,q_1 \}$, $\Sigma_0 = \{a,b\}$, $\Sigma_2 = \{f,g\}$, $\Delta_0 = \{a,b\}$, $\Delta_2 = \{m,n\}$, $Q_d = \{ q_0 \}$ and the rules are 
	\begin{align*}
	a &\to q_0[a], & b &\to q_0[b],\\
	f[q_i[x_1]q_j[x_2]] &\to q_{1-i}[m[x_1x_1]] ,  & g[q_i[x_1]q_j[x_2]] &\to q_{1-i}[n[x_1x_1]] ,\\
	f[q_i[x_1]q_j[x_2]] &\to q_{1-j}[m[x_2x_2]] ,  & g[q_i[x_1]q_j[x_2]] &\to q_{1-j}[n[x_2x_2]] ,
	\end{align*}
for all $i,j \in \{0,1\}$. 
	The transformation realized by $M$ may be described by saying that, given some input tree $t$, $M$ selects some path of even length through $t$, relabels $f$ by $m$ and $g$ by $n$ and then doubles every subtree. For example $f[a[g[ab]]]$ may be transformed into the tree $m[n[bb]n[bb]]$ corresponding to the path $fgb$. The tree $g[ab]$ is not in the domain of $M$.
\end{example}

\begin{exercise}\label{exe.4.19}
	Construct bottom-up tree transducers $M_1$ and $M_2$ such that 
	\begin{enumerate}[label=(\roman*)]
		\item $\{ (\fyield(s),\,\fyield(t)) \mid (s,t) \in T(M_1) \} = \{ (a(cd)^nfe^nb,\,ac^nfd^{2n}b) \mid n \ge 0 \}$;
		\item $M_2$ deletes, given an input tree $t$, all subtrees $t'$ of $t$ such that $\fyield(t') \in a^+b^+$.\qedhere
	\end{enumerate}
\end{exercise}

\begin{exercise}\label{exe.4.20}
	Give a recursive definition of the transformation realized by a bottom-up tree transducer (without using the notion of tree rewriting system).
\end{exercise}

\begin{exercise}\label{exe.4.21}
	Given a bottom-up ftt $M$ with input alphabet $\Sigma$, find a suitable \mbox{$\Sigma$-algebra} such that $M$ may be viewed as an interpretation of $\Sigma$ into this $\Sigma$-algebra (cf. Section~\ref{sec.4.1}).
\end{exercise}

We now define some subclasses of the class of bottom-up tree transformations.

\begin{definition}\label{def.4.22}
	Let $\Sigma$ be a ranked alphabet and $k \ge 0$. A tree $t$ in $T_\Sigma(\X_k)$ is \uem{linear} if each element of $\X_k$ occurs at most once in $t$. The tree $t$ is called \uem{nondeleting with respect to $\X_k$} if each element of $\X_k$ occurs at least once in $t$.
\end{definition}

\begin{definition}\label{def.4.23}
	Let $M = (Q,\Sigma,\Delta,R,Q_d)$ be a bottom-up ftt.\\
	$M$ is called \uem{linear} if the right hand side of each rule in $R$ is linear.\\
	$M$ is called \uem{nondeleting} if the right hand side of each rule in $R$ is nondeleting with respect to $X_k$, where $k$ is the rank of the input symbol in the left hand side.\\
	$M$ is called \uem{one-state} (or \uem{pure}) if $Q$ is a singleton.\\
	$M$ is called (partial) \uem{deterministic} if
	\begin{enumerate}[label=(\roman*)]
		\item for each $a \in \Sigma_0$ there is at most one rule in $R$ with left hand side $a$, and
		\item for each $k \ge 1$, $a \in \Sigma_k$ and $\li[q_i] \in Q$ there is at most one rule in $R$ with left hand side $\listt[a][q_i[x_i]]$.
	\end{enumerate}
	$M$ is called \uem{total deterministic} if (i) and (ii) hold with ``{}at most one" replaced by ``exactly one" and $Q_d = Q$.
\end{definition}

\begin{notation}\label{not.4.24}
	The same terminology will be applied to the transformations realized by such transducers. Thus, for instance, a linear deterministic bottom-up tree transformation is one that can be realized by a linear deterministic bottom-up ftt. 
	
	The classes of tree transformations obtained by putting one or more of the above restrictions on the bottom-up tree transducers will be denoted by adding the symbols \LIN, \class{N}, \class{P}, \D{} and $\class{D}_t$ (standing for linear, nondeleting, pure, deterministic, and total deterministic respectively) to the letter \B. Thus the class of linear deterministic bottom-up tree transformations is denoted by \LDB.
\end{notation}

\begin{example}\label{exa.4.25}
	Let $\Sigma_0 = \{ e \}$, $\Sigma_1 = \{ a, f \}$, $\Delta_0 = \{e\}$, $\Delta_1 = \{a,b\}$ and $\Delta_2 = \{f\}$. Consider the bottom-up ftt $M = (Q,\Sigma,\Delta,R,Q_d)$, where $Q = Q_d = \{\ast\}$ and $R$ consists of the rules 
	\begin{align*}
	e &\to \ast[e],\\
	a [\ast[x_1]] & \to \ast[a[x_1]], &a [\ast[x_1]] \to \ast[b[x_1]],\\
	f[\ast[x_1]] & \to \ast[f[x_1x_1]] \text{.}
	\end{align*}
	Then $M \in \PNB$.
\end{example}

Let us make the following remarks about the concepts defined in Definition \ref{def.4.23}.
\begin{remarks}~\label{rem.4.26}
	\begin{enumerate}[label=(\arabic*)]
		\item Deletion is different from erasing. A rule may be called erasing if its right hand side belongs to $Q[\X]$. Thus, symbols of rank 0 cannot be erased. Symbols of rank~1 can be erased without any deletion, but symbols of rank $k\ge 2$ can only be erased by deleting also $k-1$ subtrees. Thus a nondeleting tree transducer is still able to erase symbols of rank 1.
		\item The one-state bottom-up tree transformations correspond intuitively to the finite substitutions in the string case. 
		\item The total deterministic bottom-up ftt realize tree transformations which are total functions.\qedhere
	\end{enumerate}
\end{remarks}

\begin{exercise}\label{exe.4.27}
	Show that, in the definition of ``(partial) deterministic", we may replace the phrase ``at most one" by ``exactly one" without changing the corresponding class \DB~of deterministic bottom-up tree transformations.
\end{exercise}

In the next theorem we show that all relabelings, finite tree automaton restrictions and tree homomorphisms are realizable by bottom-up ftt.

\begin{theorem}~\label{the.4.28}
	\begin{enumerate}[label=(\arabic*)]
		\item $\REL\subseteq \PNLB$,
		\item $\FTA\subseteq \NLDB$,
		\item $\HOM= \PDTB$ and $\LHOM= \PLDTB$.
	\end{enumerate}
\end{theorem}
\vspace*{-1cm}
\begin{proof}~
(1) Let $r$ be a relabeling from $T_\Sigma$ into $T_\Delta$. Thus $r$ is determined by a family of mappings $r_k : \Sigma_k \to \mathcal{P}(\Delta_k)$. Obviously the following bottom-up ftt realizes $r$: $M = (\{\ast\}, \Sigma, \Delta, R, \{\ast\})$, where $R$ is constructed as follows:
		\begin{enumerate}[label=(\roman*)]
			\item for $a \in \Sigma_0$, if $b \in r_0(a)$, then $a \to \ast[b]$ is in $R$;
			\item for $k \ge 1$ and $a \in \Sigma_k$, if $b \in r_k(a)$, then\\
			$\listt[a][\ast[x_i]] \to \ast[\listt[b][x_i]]$ is in $R$. 
		\end{enumerate}
		Clearly $M \in \PNLB$.

(2) From the definition of \FTA~and from Part (\ref{sec.3}) it follows that we need only consider a deterministic bottom-up fta $M = (Q,\Sigma,\delta,s,F)$ and show that $T(M) = \{(t,t) \mid t \in L(M)\}$ is realized by a bottom-up ftt. Consider the bottom-up ftt $\widetilde{M} = (Q,\Sigma,\Sigma,R,F)$, where $R$ is constructed as follows: 
		\begin{enumerate}[label=(\roman*)]
			\item for $a \in \Sigma_0$, $a \to q[a]$ is in $R$, where $q = s_a$;
			\item for $k \ge 1$ and $a \in \Sigma_k$, if $\delta_a^k(\li[q_i]) = q$,\\then $\listt[a][q_i[x_i]] \to q[\listt[a][x_i]]$ is in $R$. 
		\end{enumerate}
		Clearly $\widetilde{M}$ realizes $T(M)$ and $\widetilde{M} \in \NLDB$ (the determinism of $\widetilde{M}$ follows from that of~$M$).

(3) We first show that $\HOM\subseteq\PDTB$ (and $\LHOM\subseteq \PLDTB$). An example of this was already given in Example \ref{exa.4.13}(i). Let $h$ be a tree homomorphism from $T_\Sigma$ into $T_\Delta$ determined by the mappings $h_k : \Sigma_k \to T_\Delta(\X_k)$. Consider the bottom-up ftt $M = (\{\ast\},\Sigma,\Delta,R,\{\ast\})$, where $R$ contains the following rules:
		\begin{enumerate}[label=(\roman*)]
			\item for $a \in \Sigma_0$, $a \to \ast[h_0(a)]$ is in $R$;
			\item for $k \ge 1$ and $a \in \Sigma_k$, the rule $\listt[a][\ast[x_i]] \to  \ast[h_k(a)] $ is in $R$. 
		\end{enumerate}			
		Obviously $M$ is in \PDTB~(and linear, if $h$ is linear). Let us prove that $M$ realizes $h$. Thus we have to show that, for $s \in T_\Sigma$ and $t\in T_\Delta$, $h(s) = t$ iff $s \xRightarrow{\ast} \ast[t]$.
		
		The proof is by induction on $s$. The case $s\in \Sigma_0$ is clear. Now let $s = \listt[a][s_i]$. Suppose that $h(s) = t$. Then, by definition of $h$, $t = h_k(a) \langle \li[x_i \gets h(s_i)]\rangle$. By induction, $s_i \xRightarrow{\ast} \ast[h(s_i)]$ for all $i$, $1 \le i \le k$. Hence (but note that formally this needs a proof) $\listt[a][s_i] \xRightarrow{\ast} \listt[a][\ast[h(s_i)]]$. But, by rule (ii) above, $\listt[a][\ast[h(s_i)]] \Rightarrow \ast[h_k(a) \langle \li[x_i \gets h(s_i)]\rangle]$. Consequently $s \xRightarrow{\ast} \ast[t]$.
		
		Now suppose that $s = \listt[a][s_i] \xRightarrow{\ast} \ast[t]$. Then (and again this needs a formal proof) there are trees $\li[t_i]$ such that $s_i \xRightarrow{\ast} \ast[t_i]$ for $1 \le i \le k$ and $\listt[a][s_i] \xRightarrow{\ast} \listt[a][\ast[t_i]] \Rightarrow \ast[h_k(a) \langle \li[x_i \gets t_i]\rangle]=\ast[t]$.		
		By induction, $t_i = h(s_i)$ for all $i$, $1 \le i \le k$. Hence 
$t=h_k(a)\langle x_1\gets h(s_1),\dots,x_k\gets h(s_k)\rangle = h(s)$.
		
		This proves that $(\LIN)\HOM{} \subseteq \class{P}(\LIN)\D_t\B$. To show the converse, consider a one-state total deterministic bottom-up tree transducer $M = (\{\ast\},\Sigma,\Delta,R,\{\ast\})$.
		Define the tree homomorphism $h$ from $T_\Sigma$ into $T_\Delta$ as follows:
		\begin{enumerate}[label=(\roman*)]
			\item for $a \in \Sigma_0$, $h_0(a)$ is the tree $t$ occurring in the (unique) rule $a \to \ast[t]$ in $R$;
			\item for $k\ge 1$ and $a \in \Sigma_k$, $h_k(a)$ is the tree $t$ occurring in the (unique) rule  $\listt[a][\ast[x_i]] \to \ast[t]$ in $R$.
		\end{enumerate}
		Then, obviously, by the same proof as above, $h=T(M)$.\qedhere
\end{proof}

\begin{exercise}\label{exe.4.29}
	Prove that the domain of a bottom-up tree transformation is a recognizable tree language, and vice-versa.
\end{exercise}

Let us now consider the top-down generalization of the generalized sequential machine.

The top-down finite tree transducer goes through the input tree in the same way as the top-down fta, at each step producing a piece of output to which the (unprocessed) rest of the input is concatenated. Note therefore that the transducer not really ``goes through" the input tree in the same way as the bottom-up ftt does, since in the top-down case the rest of the input may be modified (deleted, permuted, copied) during translation, whereas in the bottom-up case the rest of the input is unmodified during translation. The top-down transducer arrives at a node of rank $k$ in a certain state; on that moment the configuration is an element of $T_\Delta(Q[T_\Sigma])$, where $\Sigma$ and $\Delta$ are the input and output alphabet, and $Q$ the set of states. The state and the label at the node determine (nondeterministically) a piece of output containing the variables $\li[x_i]$, and states with which to continue the translation of the subtrees. These states are also specified in the piece of output, which is in fact a tree in $T_\Delta(Q[\X_k])$, where an occurrence of $q[x_i]$ means that the processing of the $i\textsuperscript{th}$ subtree should, at this point, be continued in state $q$. The transducer processes the node by replacing it and its direct subtrees by the piece of output, in which the $k$ subtrees are substituted for the variables $\li[x_i]$. The processing of (all copies of) the subtrees is continued as indicated above.

The transducer starts at the root of the input tree in some initial state. There should be final states and output for each node of rank 0. If the transducer arrives in a final state at each leaf, then it replaces each leaf by the final output, and the resulting tree is the transformation of the input tree. (Cf. the story in Idea~\ref{sta.4.3}). The steps of the transducer, including the final steps, are modelled by the application of rewriting rules to the elements of $T_\Delta(Q[T_\Sigma])$.

We now give the formal definition.

\begin{definition}\label{def.4.30}
	A \uem{top-down (finite) tree transducer} is a structure \\
$M = (Q,\Sigma,\Delta, R, Q_d)$, 
where $Q$, $\Sigma$ and $\Delta$ are as for the bottom-up ftt,\\
	\begin{tabular}{cp{0.8\linewidth}}
		$Q_d$	 & is a subset of $Q$ (the set of \uem{initial states}), and\\
		$R$  & is a finite set of \uem{rules} of one of the forms (i) or (ii): \\
	\end{tabular}
	\begin{enumerate}[leftmargin=1.8cm,label=(\roman*)]
		\item $q[\listt[a][x_i]] \to t$, where $k\ge1$, $a \in \Sigma_k$, $q \in Q$ and $t \in T_\Delta(Q[\X_k])$;
		\item $q[a] \to t$, where $q \in Q$, $a \in \Sigma_0$ and $t \in T_\Delta$.
	\end{enumerate}	
$M$ is viewed as a tree rewriting system over the ranked alphabet $Q \cup \Sigma \cup \Delta$ with $R$ as the set of rules, such that the range of each variable in $\X$ is $T_\Sigma$.
	
	The \uem{tree transformation realized by $M$}, denoted by $T(M)$ or simply $M$, is 
	$\{ (s,t) \in T_\Sigma \times T_\Delta \mid q[s] \xRightarrow[M]{\ast} t$ for some $q$ in $Q_d \}$.
\end{definition}

\begin{definition}\label{def.4.31}
	The class of tree transformations realized by top-down ftt will be denoted by \T. An element of \T{} will be called a \uem{top-down tree transformation}.
\end{definition}

\begin{example}\label{exa.4.32}
	An example of a top-down ftt realizing a homomorphism was given in Example \ref{exa.4.13}(ii) (it had one state $\ast$).	
\end{example}

The next example is a top-down ftt computing the formal derivative of an arithmetic expression.

\begin{example}\label{exa.4.33}
	Consider the top-down ftt $M = (Q,\Sigma,\Delta,R,Q_d)$, where $\Sigma_0 = \{a,b\}$, $\Delta_0 = \{a,b,0,1\}$, $\Sigma_1 = \Delta_1 = \{-,\uem{\sin},\uem{\cos}\}$, $\Sigma_2 = \Delta_2 = \{+,\ast\}$, $Q = \{q,i\}$, $Q_d = \{q\}$, 
	\begin{align*}
	\text{the rules for $q$ are}\\[1mm]
	q[+[x_1x_2]] 		&\to +[q[x_1]q[x_2]],\\
	q[\ast[x_1x_2]] 	&\to +[\ast[q[x_1]i[x_2]]\ast[i[x_1]q[x_2]]],\\
	q[-[x_1]] 			&\to -[q[x_1]],\\
	q[\uem{\sin}[x_1]]	&\to \ast[\uem{\cos}[i[x_1]]q[x_1]],\\
	q[\uem{\cos}[x_1]]	&\to \ast[-[\uem{\sin}[i[x_1]]]q[x_1]],\\
	q[a] 		\to 1	&,~q[b] \to 0,\\[1mm]
    \text{and the rules for $i$ are}\\[1mm]
	i[+[x_1x_2]] 		&\to +[i[x_1]i[x_2]],\\
	i[\ast[x_1x_2]] 	&\to \ast[i[x_1]i[x_2]],\\
	i[-[x_1]] 			&\to -[i[x_1]],\\
	i[\uem{\sin}[x_1]]	&\to \uem{\sin}[i[x_1]],\\
	i[\uem{\cos}[x_1]]	&\to \uem{\cos}[i[x_1]],\\
	i[a] 		\to a	~&\text{and}~i[b] \to b.\\
	\end{align*}
Then $(t,s) \in T(M)$ iff $s$ is the formal derivative of $t$ with respect to $a$. For instance, $q[\ast[+[ab]-[a]]] \xRightarrow{\ast} +[\ast[+[10]-[a]]\ast[+[ab]-[1]]]$. Note that $i[t_1] \xRightarrow{\ast} t_2$ iff $t_1 = t_2$ $(t_1,t_2 \in T_\Sigma)$.
\end{example}

\begin{exercise}\label{exe.4.34}
	Let $\Sigma_0 = \{a,b\}$, $\Sigma_1 =\{\sim\}$ and $\Sigma_2 = \{ \land,\lor \}$. $T_\Sigma$ may be viewed as the set of all boolean expressions over the boolean variables $a$ and $b$, using negation, conjunction and disjunction. Write a top-down tree transducer which transforms every boolean expression into an equivalent one in which $a$ and $b$ are the only subexpressions which may be negated.
\end{exercise}

\begin{exercise}\label{exe.4.35}
(i) Give a recursive definition of the transformation realized by a top-down ftt. 
(ii) Find a suitable $\Sigma$-algebra such that the top-down ftt may be viewed as an interpretation of $\Sigma$ into this $\Sigma$-algebra. 
\end{exercise}

As in the bottom-up case we define some subclasses of \T.

\begin{definition}\label{def.4.36}
	Let $M = (Q,\Sigma,\Delta,R,Q_d)$ be a top-down tree transducer. \\
The definitions of \uem{linear}, \uem{nondeleting} and \uem{one-state} are identical to the bottom-up ones in Definition \ref{def.4.23}. \\
	$M$ is called (partial) \uem{deterministic} if
	\begin{enumerate}[label=(\roman*)]
		\item $Q_d$ is a singleton;
		\item for each $q \in Q$, $k \ge 1$, and $a \in \Sigma_k$, there is at most one rule in $R$ with left hand side $q[\listt[a][x_i]]$;
		\item for each $q \in Q$ and $a \in \Sigma_0$ there is at most one rule in $R$ with left hand side $q[a]$.
	\end{enumerate}
	$M$ is called \uem{total deterministic} if (i), (ii) and (iii) hold with ``{}at most one" replaced by ``exactly one".
\end{definition}

Notation \ref{not.4.24} also applies to the top-down case. Thus, \PLT{} is the class of one-state linear top-down tree transformations.

\begin{example} \label{exa.4.37}
	Let $\Sigma_0=\{e\}$, $\Sigma_1 = \{a,f\}$, $\Delta_0 = \{e\}$, $\Delta_1 = \{a,b\}$ and $\Delta_2 = \{f\}$. Consider the top-down tree transducer $M = (Q,\Sigma,\Delta,R,Q_d)$ with $Q = Q_d = \{\ast\}$ and $R$ consists of the rules 
	\begin{align*}
	&\ast[f[x_1]]  \to f[\ast[x_1]\ast[x_1]],\\
	&\ast[a[x_1]]  \to a[\ast[x_1]], ~~~\ast[a[x_1]] \to b[\ast[x_1]],\\
	&\ast[e] \to e.
	\end{align*}
	Then $M \in $ \PNT.
\end{example}

Remarks \ref{rem.4.26} also apply to the top-down case.

\begin{exercise}\label{exe.4.38}
	Show that, in the definition of ``(partial) deterministic" top-down ftt, we may replace in (ii) the phrase ``{}at most one" by ``exactly one" without changing \DT.
\end{exercise}

The next theorem shows that all relabelings, finite tree automaton restrictions and tree homomorphisms are realizable by top-down tree transducers (cf. Theorem \ref{the.4.28}). Note therefore that these tree transformations are not specifically bottom-up or top-down.

\begin{theorem}\label{the.4.39}~
	\begin{enumerate}[label=(\arabic*)]
		\item $\REL\subseteq\PNLT$,
		\item $\FTA\subseteq\NLT$,
		\item $\HOM=\PDTT$ and $\LHOM=\PLDTT$.
	\end{enumerate}
\end{theorem}
\vspace*{-1cm}
\begin{proof}Exercise.
\end{proof}

In what follows we shall need one other type of tree transformation which corresponds to the ordinary sequential machine in the string case (which translates each input symbol into one output symbol). It is a combination of an fta and a relabeling.

\begin{definition}\label{def.4.40}
	A \uem{top-down} (resp. \uem{bottom-up}) \uem{finite state relabeling} is a (tree transformation realized by a) top-down (resp. bottom-up) tree transducer $M = (Q,\Sigma,\Delta,R,Q_d)$ in which all rules are of the form $q[\listt[a][x_i]] \to \listt[b][q_i[x_i]]$ with $q,\li[q_i] \in Q$, $a\in \Sigma_k$ and $b \in \Delta_k$, or of the form $q[a]\to b$ with $q\in Q$, $a\in \Sigma_0$ and $b\in \Delta_0$ (resp. of the form $\listt[a][q_i[x_i]] \to q[\listt[b][x_i]]$ with $q,\li[q_i] \in Q$, $a \in \Sigma_k$ and $b \in \Delta_k$, or of the form $a\to q[b]$ with $q\in Q$, $a\in \Sigma_0$ and $b\in \Delta_0$).
\end{definition}

It is clear that the classes of top-down and bottom-up finite state relabelings coincide. This class will be denoted by \QREL. The classes of deterministic top-down and deterministic bottom-up finite state relabelings obviously do not coincide. They will be denoted by \DTQREL~and \DBQREL~respectively. Note that $\FTA\cup\REL\subseteq\QREL\subseteq\NLB\cap\NLT$.

Apart from the tree transformation realized by a tree transducer we will also be interested in the image of a recognizable tree language under a tree transformation and the yield of that image.

\begin{definition}\label{def.4.41}
	Let \Xclass~be a class of tree transformations. An \uem{\Xclass-surface tree language} is a language $M(L)$ with $M \in \Xclass$ and $L \in \RECOG$. An \uem{\Xclass-target language} is the yield of an \Xclass-surface language. An \uem{\Xclass-translation} is a string relation $\{(\fyield(s),\,\fyield(t)) \mid (s,t) \in M$ and $s \in L\}$ for some $M \in \Xclass$ and $L \in \RECOG$.
	
	The classes of \Xclass-surface and \Xclass-target languages will be denoted by \Xclass-Surface and \Xclass-Target respectively.
\end{definition}

It is clear that, for all classes \Xclass{} discussed so far, since the identity transformation is in~\Xclass, $\RECOG\subseteq\text{\Xclass-Surface}$, and so $\CFL\subseteq\text{\Xclass-Target}$. Moreover it is clear from the proof of Theorem \ref{the.3.64} that if $\HOM\subseteq\Xclass$, then the above inclusions are proper.

\subsection{Comparison of B and T, the nondeterministic case}\label{sec.4.3}
The main differences between the bottom-up and the top-down tree transducer are the following.
\parskip 8pt

\uline{Property (B)}. Nondeterminism followed by copying.\\
A bottom-up ftt has the ability of first processing an input subtree nondeterministically and then copying the resulting output tree.

\uline{Property (T)}. Copying followed by different processing (by nondeterminism or by different states).\\
A top-down ftt has the ability of first copying an input subtree  and then treating the resulting copies differently.

\uline{Property (B$'$)}. Checking followed by deletion.\\
A bottom-up ftt has the ability of first processing an input subtree and then deleting the resulting output subtree. In other words, depending on a (recognizable) property of the input subtree, it can decide whether to delete the output subtree or do something else with it.

It should be intuitively clear that top-down ftt do not possess properties (B) and (B$'$), whereas bottom-up ftt do not have property (T). We now show that these differences also result in differences in the corresponding classes of tree transformations.
\parskip 0pt
\begin{notation}\label{not.4.42}
	For any alphabet $\Sigma$, not containing the brackets [ and ], we define a function $m : \Sigma^+ \to (\Sigma \cup \{[\,,]\})^+$ as follows: for $a \in \Sigma$ and $w \in \Sigma^+$, $m(a) = a$ and $m(aw) = a[m(w)]$. For instance, $m(aab) = a[a[b]]$. 
	
	Note that $m$ is a kind of converse to the mapping $f_\text{td}$ discussed after Definition \ref{def.2.21}. A tree of the form $m(w)$ will also be called a monadic tree.
\end{notation}

\begin{theorem} \label{the.4.43}
	The classes of bottom-up and top-down tree transformations are incomparable. In particular, there are tree transformations in $\PNB-\T$ and $\PNT-\B$.
\end{theorem}
\begin{proof}~\\
	(1)~ Consider the bottom-up ftt $M$ of Example \ref{exa.4.25}. $M$ is in \PNB{} and is a typical example of an ftt having property (B). It is intuitively clear that $M$ is not realizable by a top-down ftt. In fact, consider for each $n \ge 1$ the specific input tree $f[m(a^ne)]$. This tree is nondeterministically transformed by $M$ into all trees of the form $f[m(we)m(we)]$, where $w$ is a string over $\{a,b\}$ of length $n$. Suppose that the top-down ftt $N = (Q', \Sigma', \Delta', R', Q'_d)$ could do the same transformation of these input trees. Then, roughly speaking, $N$ would first have to make a copy of $m(a^ne)$ and would then have to relabel the two copies in an arbitrary but identical way, which is clearly impossible. A formal proof goes as follows. If $N$ realizes the same transformation, then, for each $n \ge 1$ and each $w \in \{a,b\}^\ast$ of length~$n$, there is a derivation $q_0[f[m(a^ne)]] \xRightarrow[N]{\ast} f[m(we)m(we)]$ for some $q_0$ in $Q_d'$. 
	
	Let us consider a fixed $n$. Consider, in each of these $2^n$ derivations, the first string of the form $f[t_1t_2]$; that is, consider the moment that $f$ is produced as output. Note that this is not necessarily the second string of the derivation, since the transducer may first erase the input symbol $f$ and some of the $a$'s before producing any output (thus the derivation may look like $q_0[f[m(a^ne)]] \xRightarrow{\ast} q[a[m(a^ke)]] \Rightarrow f[t_1t_2] \xRightarrow{\ast} f[m(we)m(we)]$ for some $q\in Q'$ and some $k$, $0\le k<n$, or even like $q_0[f[m(a^ne)]] \xRightarrow{\ast} q[e] \Rightarrow f[t_1t_2] =
f[m(we)m(we)]$ for some $q\in Q$). Obviously, for different derivations these strings have to be different: if, for $w \ne w'$, both $t_1 \xRightarrow{\ast} m(we)$, $t_2 \xRightarrow{\ast} m(we)$ and $t_1 \xRightarrow{\ast} m(w'e)$, $t_2 \xRightarrow{\ast} m(w'e)$, then also $f[t_1t_2] \xRightarrow{\ast} f[m(we)m(w'e)]$, which is an invalid output. Therefore there are $2^n$ of such strings $f[t_1t_2]$. However it is clear that $f[t_1t_2]$ is of the form $f[\,\mbar{t}_1\mbar{t}_2] \langle x_1 \gets m(a^ke)\rangle$, where $0 \le k \le n$ and $f[\,\mbar{t}_1\mbar{t}_2]$ is the right hand side of a rule in $R'$. Therefore the number of possible $f[t_1t_2]$'s is less than $(n+1)r$, where $r = \#(R')$. For $n$ sufficiently large this is a contradiction.
	
	(2)~ Consider now the top-down ftt $M$ of Example \ref{exa.4.37}. $M$ is in \PNT{} and is a typical example of an ftt having property (T). Suppose that $M$ can be realized by a bottom-up ftt $N = (Q',\Sigma',\Delta',R',Q'_d)$. Consider again for each $n \ge 1$ the specific input tree $f[m(a^ne)]$. This tree should be transformed by $N$ into all trees of the form $f[m(w_1e)m(w_2e)]$ for $w_1,w_2 \in \{a,b\}^\ast$ of length $n$. Let us consider, in each of the derivations realizing this transformation, the first string which contains the output tree. Note that this is not necessarily the last string since $N$ may end its computation by erasing a number of $a$'s and the input $f$. Obviously, this string is obtained from the previous one by application of a rule with right hand side of the form $q[f[\,\mbar{t}_1\mbar{t}_2]]$, where $q \in Q'$, $\mbar{t}_1, \mbar{t}_2 \in T_\Delta(\{x_1\})$ and there are $s_1$ and $s_2$ such that $\mbar{t}_1\langle x_1 \gets s_1\rangle = m(w_1e)$ and $\mbar{t}_2\langle x_1 \gets s_2\rangle = m(w_2e)$. Obviously, if $f[\,\mbar{t}_1\mbar{t}_2]$ contains no $x_1$ or only one $x_1$, then the rule can only be used for exactly one input tree $f[m(a^ne)]$. Thus we may choose $n$ such that in all derivations starting with $f[m(a^ne)]$ the right hand side $q[f[\,\mbar{t}_1\mbar{t}_2]]$ contains two $x_1$'s (it cannot contain more). Thus $q[f[\,\mbar{t}_1\mbar{t}_2]]$ is of the form $q[f[m(v_1x_1)m(v_2x_1)]]$ for certain $v_1,v_2 \in \{a,b\}^\ast$. By choosing $n$ larger than the length of all such $v_1$'s and $v_2$'s occurring in right hand sides of rules in $R'$, we see that the output tree always has two equal subtrees $\ne e$: it has to be of the form $f[m(v_1we)m(v_2we)]$ for some $w \in \{a,b\}^+$. Thus, for such an $n$, not all possible outputs are produced. This is a contradiction.
\end{proof}

An important property of a class \Fclass{} of tree transformations is whether it is closed under composition. If so, then we know that each sequence of transformations from~\Fclass{} can be realized by one tree transducer (corresponding to the class \Fclass). We then also know that the class of \Fclass-surface tree languages is closed under the transformations of~\Fclass. The next theorem shows that unfortunately the classes of top-down and bottom-up tree transformations are not closed under composition. This nonclosure is caused by the failure of property (B) for top-down transformations (property (T) for bottom-up transformations).

\begin{theorem}\label{the.4.44}
	\T~and \B~are not closed under composition. In particular, there are tree transformations in $(\REL\comp\HOM)-\T$ and in $(\HOM\comp\REL) - \B$.
\end{theorem}
\begin{proof}~\\
	(1)~ The bottom-up ftt $M$ of Example \ref{exa.4.25} can be realized by the composition of a relabeling and a homomorphism. Let $\Omega_0 = \{e\}$ and $\Omega_1 = \{a,b,f\}$. Let $r$ be the relabeling from $\Sigma$ into $\Omega$ defined by $r_0(e) = \{e\}$, $r_1(a) = \{a,b\}$ and $r_1(f) = \{f\}$. Let $h$ be the tree homomorphism from $\Omega$ into $\Delta$ defined by $h_0(e) = e$, $h_1(a) = a[x_1]$, $h_1(b) = b[x_1]$ and $h_1(f) = f[x_1x_1]$. Then, for all $s \in T_\Sigma$ and $t \in T_\Delta$, $(s,t) \in M$ iff there exists $u$ in $T_\Omega$ such that $u \in r(s)$ and $h(u) = t$. 
	
	Thus, by the first part of the proof of Theorem \ref{the.4.43}, $M$ is in $(\REL\comp\HOM)-\T$.
	
	(2)~ The top-down ftt $M$ of Example \ref{exa.4.37} can be realized by the composition of a homomorphism and a relabeling. Let $\Pi$ be the ranked alphabet with $\Pi_0 = \{e\}$, $\Pi_1 = \{a\}$ and $\Pi_2=\{f\}$. Let $h$ be the tree homomorphism from $\Sigma$ into $\Pi$ defined by $h_0(e) = e$, $h_1(a) = a[x_1]$ and $h_1(f) = f[x_1x_1]$. Let $r$ be the relabeling from $\Pi$ into $\Delta$ defined by $r_0(e) = \{e\}$, $r_1(a) = \{a,b\}$ and $r_2(f) = \{f\}$. Then, for all $s \in T_\Sigma$ and $t \in T_\Delta$, $(s,t) \in M$ iff there exists $u$ in $T_\Pi$ such that $h(s) = u$ and $t \in r(u)$. 
	
	Thus, by the second part of the proof of Theorem \ref{the.4.43}, $M$ is in $(\HOM\comp\REL)-\B$.
\end{proof}

\begin{exercise}\label{exe.4.45}
	Prove the statements in the above proof.
\end{exercise}

One might get the impression that each bottom-up (resp. top-down) tree transformation can be realized by two top-down (resp. bottom-up) tree transducers (i.e., $\B\subseteq\T\comp\T$, resp. $\T\subseteq\B\comp\B$). We shall show later that this is true. 

Let us now consider the linear case. Since properties (B) and (T) are now eliminated, the only remaining difference between linear top-down and bottom-up tree transducers is caused by property (B$'$).

\begin{lemmawithoutqed}\label{lem.4.46}
	There is a tree transformation $M$ that belongs to \LDB, but not to \T. $M$ can be realized by the composition of a deterministic top-down fta with a linear homomorphism.
\end{lemmawithoutqed}
\begin{proof}
	Let $\Sigma_0 = \{c\}$, $\Sigma_1 = \{b\}$, $\Sigma_2 = \{a\}$, $\Delta_0 = \{c\}$ and $\Delta_1 = \{a,b\}$. Consider the tree transformation $M = \{ (a[tc], a[t]) \mid t = m(b^nc)$ for some $n\ge 0 \}$. We shall show that $M\notin$ \T. The rest of the proof is left as an \uem{exercise}. Suppose that there is a top-down ftt $N = (Q,\Sigma',\Delta',R,Q_d) $ such that $T(N) = M$. Each successful derivation of $N$ has to start with the application of a rule $q_0[a[x_1x_2]] \to s$, where $q_0 \in Q_d$ and $s \in T_{\Delta'}(\X_2)$. Now, if $s$ contains no $x_1$, then we could change the input $a[tc]$ into $a[t'c]$ without changing the output. If $s$ contains no $x_2$, then we could change $a[tc]$ into $a[tb[c]]$ and still obtain (the same) output. But if $s$ contains both $x_1$ and $x_2$ then it has to contain a symbol of rank~2 and so $a[t]$ cannot be derived.
\end{proof}

Since both deterministic top-down fta and linear homomorphisms belong to \LDT{} we can state the following corollary.

\begin{corollarywithqed}\label{cor.4.47}
	Composition of linear deterministic top-down tree transformations leads out of the class of top-down tree transformations; in a formula:\\ $(\LDT\circ\LDT)-\T\ne\emptyset$.
\end{corollarywithqed}

We now show that, in some sense, property (B$'$) is the only cause of difference between linear bottom-up and linear top-down tree transformations. Firstly, all linear top-down tree transformations can be realized linear bottom-up. Secondly, in the nondeleting linear case, all differences between top-down and bottom-up are gone (this can be considered as a generalization of Theorem \ref{the.3.17}).

\begin{theorem}~\label{the.4.48}
	\begin{enumerate}[label=(\arabic*)]
		\item $\LT\subsetneq\LB$,
		\item $\NLT=\NLB$.
	\end{enumerate}
\end{theorem}
\vspace*{-1cm}
\begin{proof}
	We first show part (2). Let us say that a nondeleting linear bottom-up ftt $M=(Q,\Sigma,\Delta,R,Q_d)$ and a nondeleting linear top-down ftt $N=(Q',\Sigma',\Delta',R',Q'_d)$ are ``{}associated" if $Q = Q'$, $\Sigma=\Sigma'$, $\Delta=\Delta'$, $Q_d=Q'_d$ and
	\begin{enumerate}[label=(\roman*)]
		\item for each $a \in \Sigma_0$, $q\in Q$ and $t\in T_\Delta$, \\[1mm]
              $a \to q[t]$ is in $R$ iff $q[a] \to t$ is in $R'$;		
		\item\label{the.4.48ii} for each $k \ge 1$, $a \in \Sigma_k$, $\li[q_i],q\in Q$ and 
                                $t \in T_\Delta(\X_k)$ linear and nondeleting w.r.t.\ $\X_k$, \\[1mm]
		      $\listt[a][q_i[x_i]] \to q[t]$ is in $R$ iff \\[1mm] 
              $q[\listt[a][x_i]] \to t\langle \li[x_i\gets q_i[x_i]]\rangle$ is in $R'$.			
	\end{enumerate}
Note that each tree $r \in T_\Delta(Q[\X_k])$, which is linear and nondeleting w.r.t.\ $\X_k$, is of the form $t\langle \li[x_i\gets q_i[x_i]]\rangle$, where $t\in T_\Delta(\X_k)$ is linear and nondeleting w.r.t.\ $\X_k$ (in fact, $t$ is the result of replacing $q_i[x_i]$ by $x_i$ in $r$). Therefore it is clear that for each $M \in\NLB$ there exists an associated $N \in\NLT$ and vice versa. Hence it suffices to prove that associated ftt realize the same tree transformation. Let $M$ and $N$ be associated as above. We shall prove, by induction on $s$, that for every $q\in Q$, $s \in T_\Sigma$ and $u\in T_\Delta$, 
\[
\begin{array}{lll}
	s \xRightarrow[M]{\ast} q[u]  &\text{ iff } & q[s] \xRightarrow[N]{\ast} u\text{.} \tag{$\ast$}
\end{array}
\]
	For $s\in\Sigma_0$, $(\ast)$ is obvious. Suppose now that $s = \listt[a][s_i]$ for some $k \ge 1$, $a \in \Sigma_k$ and $\li[s_i] \in T_\Sigma$. The only-if part of $(\ast)$ is left to the reader (it is similar to the proof of Theorem \ref{the.4.28}(3)). The if-part of $(\ast)$ is proved as follows (it is similar to the proof of Theorem \ref{the.3.65}). Let the first rule applied in the derivation $q[\listt[a][s_i]] \xRightarrow[N]{\ast} u$ be $q[\listt[a][x_i]] \to r$, and let $r = t\langle\li[x_i \gets q_i[x_i]]\rangle$ for certain $t \in T_\Delta(\X_k)$ and $\li[q_i] \in Q$. Thus $q[\listt[a][s_i]] \xRightarrow[N]{} t\langle\li[x_i \gets q_i[s_i]]\rangle \xRightarrow[N]{\ast} u$. Since $t$ is linear and nondeleting, there exist $\li[u_i] \in T_\Delta$ such that $u=t\langle\li[x_i \gets u_i]\rangle$ and $q_i[s_i] \xRightarrow[N]{\ast} u_i$ for all $i$, $1 \le i \le k$. Hence, by induction, $s_i \xRightarrow[M]{\ast} q_i[u_i]$ for all $i$, $1\le i \le k$. Also, by associatedness, the rule $\listt[a][q_i[x_i]] \to q[t]$ is in $R$. Consequently, $\listt[a][s_i] \xRightarrow[M]{\ast}\listt[a][q_i[u_i]]\xRightarrow[M]{} q[t\langle\li[x_i \gets u_i]\rangle] = q[u]$.
	
	We now show part (1). By Lemma \ref{lem.4.46}, it suffices to show that $\LT\subseteq\LB$. In principle we can use the construction used above to show $\NLT\subseteq\NLB$. The only problem is that the top-down transducer $N$ may delete subtrees, whereas a bottom-up transducer is forced to process a subtree before deleting it. The solution is to add an ``identity state" $d$ to the set of states of $M$ which allows $M$ to process any subtree which has to be deleted ($d$ is such that for all $t\in T_\Sigma$, $t \xRightarrow[M]{\ast} d[t]$). The formal construction is as follows. Let $N=(Q,\Sigma,\Delta,R,Q_d)$ be a linear top-down ftt. Construct the linear bottom-up ftt $M = (Q \cup \{d\},\Sigma,\Delta \cup \Sigma, R_M,Q_d)$, where $R_M$ is obtained as follows. 
	\begin{enumerate}[label=(\roman*)]
		\item For each $a \in \Sigma_0$ the rule $a \to d[a]$ is in $R_M$, and for each $k \ge 1$ and $a \in \Sigma_k$ the rule $\listt[a][d[x_i]] \to d[\listt[a][x_i]]$ is in $R_M$.
		
		\item For $q \in Q$, $a \in \Sigma_0 $ and $t \in T_\Delta$, if $q[a] \to t $ is in $R$, then $a \to q[t]$ is in $R_M$.
		
		\item Let $q[\listt[a][x_i]] \to t$ be in $R$, where $q\in Q$, $k\ge1$, $a\in \Sigma_k$ and $t$ is a linear tree in $T_\Delta(Q[\X_k])$. Determine the (unique) states $\li[q_i] \in Q \cup \{d\}$ such that, for $1 \le i \le k$, either $q_i[x_i]$ occurs in $t$ or ($x_i$ does not occur in $t$ and) $q_i = d$. Determine $t' \in T_\Delta(\X_k)$ such that $t'\langle\li[x_i \gets q_i[x_i]]\rangle = t$. Then the rule $\listt[a][q_i[x_i]] \to q[t']$ is in $R_M$. 
	\end{enumerate}
	Again $(\ast)$ can be proved, and since the proof only slightly differs from the previous one, it is left to the reader. 
\end{proof}

\begin{exercise}\label{exe.4.49}
	Find an example of a tree transformation in $\LDTB-\T$.
\end{exercise}

\begin{exercise}\label{exe.4.50}
	Compare the classes \PLT{} and \PLB.
\end{exercise}

\begin{exercise}\label{exe.4.51}
	Let a deterministic top-down ftt be called ``{}simple" if it is not allowed to make different translations of the same input subtree (if $q[\listt[a][x_i]] \to t$ is a rule and $q_1[x_i]$, $q_2[x_i]$ occur in $t$, then $q_1 = q_2$). Prove that the class of simple deterministic top-down tree transformations is included in \B. (This result should be expected from the fact that property (T) is eliminated. Similarly, one can prove that $\NDB\subseteq\T$, because properties (B) and (B$'$) are eliminated.)
\end{exercise}

\subsection{Decomposition and composition of bottom-up tree transformations}
Since bottom-up tree transformations are theoretically easier to handle than top-down tree transformations, we start investigating the former.

We have seen that a bottom-up ftt can copy after nondeterministic processing (property~(B)). The next theorem shows that these two things can in fact be taken apart into different phases of the transformation: each bottom-up ftt can be decomposed into two transducers, the first doing the nondeterminism (linearly) and the second doing the copying (deterministically).
\begin{theorem}\label{the.4.52}
	Each bottom-up tree transformation can be realized by a finite state relabeling followed by a homomorphism. In formula:
	\begin{flalign*}
	&&\B&\subseteq\QREL\comp\HOM.& \hspace{7cm}\\
	&\text{Moreover} & \LB & \subseteq\QREL\comp\LHOM~\text{ and}& \\
 	&&\DB & \subseteq\DBQREL\comp\HOM. &         
	\end{flalign*}
\end{theorem}
\vspace*{-1cm}
\begin{proof}
	Let $M=(Q,\Sigma,\Delta,R,Q_d)$ be a bottom-up ftt. To simulate $M$ in two phases we apply a technique similar to the one use in the proof of Theorem \ref{the.3.58}: a finite state relabeling is used to put information on each node indicating by which piece of the tree the node should be replaced; then a homomorphism is used to actually replace each node by that piece of tree. The formal construction is as follows.
	
	We simultaneously construct a ranked alphabet $\Omega$, the set of rules $R_N$ of a bottom-up ftt $N=(Q,\Sigma,\Omega,R_N,Q_d)$ and a homomorphism $h:T_\Omega\to T_\Delta$ as follows.
	\begin{enumerate}[label=(\roman*)]
		\item If $a\to q[t]$ is a rule in $R$, then $d_t$ is a (new) symbol in $\Omega_0$, $a\to q[d_t]$ is in $R_N$ and $h_0(d_t)=t$.
		\item If $\listt[a][q_i[x_i]]\to q[t]$ is a rule in $R$, then $d_t$ is a (new) symbol in $\Omega_k$, $\listt[a][q_i[x_i]]\to q[\listt[d_t][x_i]]$ is in $R_N$ and $h_k(d_t)=t$.
	\end{enumerate}
	
\noindent
The only requirement on the symbols of $\Omega$ is that if $t_1\neq t_2$ then $d_{t_1}\neq d_{t_2}$.
	
	Obviously $N$ is a (bottom-up) finite state relabeling. Also, if $M$ is linear then $h$ is linear, and if $M$ is deterministic then so is $N$.
	
	It can easily be shown (by induction on $s$) that, for $s\in T_\Sigma$, $q\in Q$ and $t\in T_\Delta$,
\[
\begin{array}{lll}
s\xRightarrow[M]{*}q[t] & \text{ iff } & \exists u\in T_\Omega:s\xRightarrow[N]{*}q[u] \text{ and } h(u)=t. 
\end{array}
\]
From this it follows that $M=N\comp h$, which proves the theorem.
\end{proof}
\begin{example}\label{exa.4.53}
	Consider the bottom-up ftt $M$ of Example \ref{exa.4.18}. It can be decomposed as follows. Firstly, $\Omega_0=\{a,b\}$ and $\Omega_2=\{m_1,m_2,n_1,n_2\}$, where $d_a=a$, $d_b=b$, $d_{m[x_1x_1]}=m_1$, $d_{m[x_2x_2]}=m_2$, $d_{n[x_1x_1]}=n_1$ and $d_{n[x_2x_2]}=n_2$. Secondly, $N=(Q,\Sigma,\Omega,R_N,Q_d)$, where $R_N$ consists of the rules 
	\begin{align*}
	a & \to q_0[a], & b & \to q_0[b] \\
	f[q_i[x_1]q_j[x_2]]  & \to q_{1-i}[m_1[x_1x_2]], &  g[q_i[x_1]q_j[x_2]]  & \to q_{1-i}[n_1[x_1x_2]], \\
	f[q_i[x_1]q_j[x_2]] &\to q_{1-j}[m_2[x_1x_2]],&
	g[q_i[x_1]q_j[x_2]] &\to q_{1-j}[n_2[x_1x_2]]
	\end{align*}
	for all $i,j\in\{0,1\}$.
	Finally, $h$ is defined by $h_0(a)=a$, $h_0(b)=b$, $h_2(m_1)=m[x_1x_1]$, $h_2(m_2)=m[x_2x_2]$, $h_2(n_1)=n[x_1x_1]$ and $h_2(n_2)=n[x_2x_2]$.
	For example, $ f[a[g[ab]]] \xRightarrow[N]{*} q_0[m_2[a[n_2[ab]]]] $ and $h( m_2[a[n_2[ab]]] )= m[n[bb]n[bb]] $.
\end{example}
Note that Theorem \ref{the.4.52} means (among other things) that each bottom-up tree transformation can be realized by the composition of two top-down tree transformations (cf. Theorem \ref{the.4.43}).

We now show that, in the nondeterministic case, the finite state relabeling can still be decomposed further into a relabeling followed by a finite tree automaton restriction.
\begin{theorem}\label{the.4.54}
	\begin{align*}
	\B  & \subseteq\REL\comp\FTA\comp\HOM~~\text{ and} \hspace{7cm}\\
	\LB & \subseteq\REL\comp\FTA\comp\LHOM.           
	\end{align*}
\end{theorem}
\vspace*{-1cm}
\begin{proof}
	By the previous theorem and the fact that \HOM{} and \LHOM{} are closed under composition (see Exercise \ref{exe.4.9}) it clearly suffices to show that $\QREL\subseteq\REL\comp\FTA\comp\LHOM$. Let $M=(Q,\Sigma,\Delta,R,Q_d)$ be a bottom-up finite state relabeling. We shall actually show that $M$ can be simulated by a relabeling, followed by an fta, followed by a projection (which is in \LHOM). The relabeling guesses which rule is applied by $M$ at each node (and puts that rule as a label on the node), the bottom-up fta checks whether this guess is in accordance with the possible state transitions of $M$, and finally the projection labels the node with the right label.
	
	Formally we construct a ranked alphabet $\Omega$, a relabeling $r$ from $T_\Sigma$ into $T_\Omega$, a (nondeterministic) bottom-up fta $N=(Q,\Omega,\delta,S,Q_d)$ and a projection $p$ from $T_\Omega$ into $T_\Delta$ as follows.
	\begin{enumerate}[label=(\roman*)]
		\item If rule $m$ in $R$ is of the form $a\to q[b]$, then $d_m$ is a (new) symbol in $\Omega_0$, $d_m\in r_0(a)$, $q\in S_{d_m}$ and $p_0(d_m)=b$.
		\item If rule $m$ in $R$ is of the form $\listt[a][q_i[x_i]]\to q[\listt[b][x_i]]$, then $d_m$ is a (new) symbol in $\Omega_k$, $d_m\in r_k(a)$, $q\in\delta_{d_m}^k(\li[q_i])$ and $p_k(d_m)=b$.
	\end{enumerate}

\noindent
We require that if $m$ and $n$ are different rules, then $d_m\neq d_n$.
	
	It is left to the reader to show that, for $s\in T_\Sigma$, $q\in Q$ and $t\in T_\Delta$,
\[
\begin{array}{lll}
s\xRightarrow[M]{*}q[t] & \text{ iff } & \exists u\in T_\Omega:u\in r(s), u\in L(N) \text{ and } t=p(u).
\end{array}
\]
This proves the theorem.
\end{proof}
These decomposition results are often very helpful when proving something about bottom-up tree transformations: the proof can often be split up into proofs about \REL{}, \FTA{} and \HOM{} only. As an example, we immediately have the following result from Theorem \ref{the.4.54} and Theorems \ref{the.3.32}, \ref{the.3.48} and \ref{the.3.65} (note that a class of tree languages is closed under fta restrictions if and only if it is closed under intersection with recognizable tree languages!).
\begin{corollarywithqed}\label{cor.4.55}
	\RECOG{} is closed under linear bottom-up tree transformations.
\end{corollarywithqed}
(This expresses that the image of a recognizable tree language under a linear bottom-up tree transformation is again recognizable. In other words, $\text{\LB-Surface}=\RECOG$.)
\begin{exercise}\label{exe.4.56}
	Prove, using Theorem \ref{the.4.54}, that $\text{\B-Surface}=\text{\HOM-Surface}$. Prove that, in fact, each \B-Surface tree language is the homomorphic image of a rule tree language.
\end{exercise}
We now prove that under certain circumstances the composition of two elements in \B{} is again in \B. Recall from Section \ref{sec.4.3} that the non-closure of \B{} under composition was caused by the failure of property (T) for \B{}: in general, in \B{}, we can't compose a copying transducer with a nondeterministic one. We now show that if either the first transducer is noncopying or the second one is deterministic, then their composition is again in \B{}. Thus, when eliminating (the failure of) property (T), closure results are obtained.
\begin{theorem}\label{the.4.57}
\[
\begin{array}{rrlllrll}
(1) & \LB\comp\B & \!\!\!\subseteq & \!\!\!\!\B & \text{and} & \LB\comp\LB & \!\!\!\subseteq & \!\!\!\!\LB.  \hspace{6cm}\\[2mm]
(2) & \B\comp\DB & \!\!\!\subseteq & \!\!\!\!\B & \text{and} & \DB\comp\DB & \!\!\!\subseteq & \!\!\!\!\DB.
\end{array}
\]	
\end{theorem}
\vspace*{-1cm}
\begin{proof}
	Because of our decomposition of \B{} we only need to look at special cases. These are treated in three lemmas, concerning composition with homomorphisms, fta restrictions and relabelings respectively. Detailed induction proofs of these lemmas are left to the reader.
	\begin{lemma*}
		$\B\comp\HOM\subseteq\B$, \;$\LB\comp\LHOM\subseteq\LB$ and $\DB\comp\HOM\subseteq\DB$.
	\end{lemma*}
	\begin{proof}
		Let $M=(Q,\Sigma,\Delta,R,Q_d)$ be a bottom-up ftt and $h$ a tree homomorphism from $T_\Delta$ into $T_\Omega$. We have to show that $M\comp h$ can be realized by a bottom-up ftt $N$. The idea is the same as that of Theorem \ref{the.3.65}: $N$ simulates $M$ but outputs at each step the homomorphic image of the output of $M$. Note that, contrary to the proof of Theorem~\ref{the.3.65} (which was concerned with regular tree grammars, a top-down device), we need not require linearity of $h$.
		
		The construction is as follows. Extend $h$ to trees in $T_\Delta(\X)$ by defining $h_0(x_i)=x_i$ for all $x_i$ in $\X$. Thus $h$ is now a homomorphism from $T_\Delta(\X)$ into $T_\Omega(\X)$. Define $N=(Q,\Sigma,\Omega,R_N,Q_d)$ such that
		\begin{enumerate}[label=(\roman*)]
			\item if $a\to q[t]$ is in $R$, then $a\to q[h(t)]$ is in $R_N$;
			\item if $\listt[a][q_i[x_i]]\to q[t]$ is in $R$,
			then $\listt[a][q_i[x_i]]\to q[h(t)]$ is in $R_N$.
		\end{enumerate}
		Obviously, if $M$ and $h$ are linear then so is $N$ (the linear homomorphism transforms a linear tree in $T_\Delta(\X)$ into a linear tree in $T_\Omega(\X)$), and if $M$ is deterministic then so is~$N$.
	\end{proof}
	\begin{lemma*}
		$\B\comp\FTA\subseteq\B$, \;$\LB\comp\FTA\subseteq\LB$ and $\DB\comp\FTA\subseteq\DB$.
	\end{lemma*}
	\begin{proof}
		The idea of the proof is similar to the one used to solve Exercise \ref{exe.3.68}. Let $M=(Q,\Sigma,\Delta,R,Q_d)$ be a bottom-up ftt and let $\widetilde{N}=(Q_N,\Delta,\Delta,R_N,Q_{dN})$ be a deterministic bottom-up ftt corresponding to a deterministic bottom-up fta as in the proof of Theorem \ref{the.4.28}(2). We have to show that $M\comp\widetilde{N}$ can be realized by a bottom-up ftt $K$. $K$ will have $Q\times Q_N$ as its set of states and it will simultaneously simulate $M$ and keep track of the state of $\widetilde{N}$ at the computed output of $M$.
		
		Extend $\widetilde{N}$ by expanding its alphabet to $\Delta\cup\X$ (or, better, to $\Delta\cup\X_m$ where $m$ is the highest subscript of a variable occurring in the rules of $M$) and by allowing the variables in its rules to range over $T_\Delta(\X)$. Thus the computation of the finite tree automaton $\widetilde{N}$ may now be started with an element of $T_\Delta(Q_N[\X])$, which means that, at certain places in the tree, $\widetilde{N}$ has to start in prescribed start states.
		
		Construct $K=(Q\times Q_N,\Sigma,\Delta,R_K,Q_d\times Q_{dN})$ such that
		\begin{enumerate}[label=(\roman*)]
			\item if $a\to q[t]$ is in $R$ and if $t\xRightarrow[\widetilde{N}]{*}q'[t]$, \\
                  then $a\to(q,q')[t]$ is in $R_K$;
			\item if $\listt[a][q_i[x_i]]\to q[t]$ is in $R$, and 
			      if $t\langle\li[x_i\gets q'_i[x_i]]\rangle\xRightarrow[\widetilde{N}]{*}q'[t]$, \\
                  then the rule $\listt[a][(q_i,q'_i)[x_i]]\to(q,q')[t]$ is in $R_K$.
		\end{enumerate}
		Note that if $M$ is linear, then so is $K$. Moreover, since $\widetilde{N}$ is deterministic, if $M$ is deterministic then so is $K$.
	\end{proof}
	\begin{lemma*}
		$\B\comp\DBQREL\subseteq\B$, \;$\DB\comp\DBQREL\subseteq\DB$ and $\LB\comp\REL\subseteq\LB$.
	\end{lemma*}
	\begin{proof}
		The proofs of the first two statements are easy generalizations of the proof of the previous lemma. The proof of the third statement is left as an easy exercise.
	\end{proof}
	We now finish the proof of Theorem \ref{the.4.57}.\\
	Firstly
	\begin{align*}
	\LB\comp\B & \subseteq\LB\comp\REL\comp\FTA\comp\HOM &   & \text{(Thm. \ref{the.4.54})}    \\
	& \subseteq\LB\comp\FTA\comp\HOM          &   & \text{(3\textsuperscript{d} lemma)}  \\
	& \subseteq\LB\comp\HOM                   &   & \text{(2\textsuperscript{d} lemma)}  \\
	& \subseteq\B                             &   & \text{(1\textsuperscript{st} lemma)} 
	\end{align*}
	and similarly for $\LB\comp\LB\subseteq\LB$.\\
	Secondly
	\begin{align*}
	\B\comp\DB & \subseteq\B\comp\DBQREL\comp\HOM &   & \text{(Thm. \ref{the.4.52})}      \\
	& \subseteq\B\comp\HOM             &   & \text{(3\textsuperscript{d} lemma)}  \\
	& \subseteq\B                      &   & \text{(1\textsuperscript{st} lemma)} 
	\end{align*}
	and similarly for $\DB\comp\DB\subseteq\DB$.
\end{proof}
Note that the ``right hand side" of Theorem \ref{the.4.57} states that \LB{} and \DB{} are closed under composition.
\begin{exercise}\label{exe.4.58}
	Is \LDB{} closed under composition?
\end{exercise}
\begin{exercise}\label{exe.4.59}
	Prove that \B-Surface is closed under intersection with recognizable tree languages.
\end{exercise}
A consequence of Theorem \ref{the.4.57} (or in fact its second lemma) is the following useful fact.
\begin{corollary}\label{cor.4.60}
	\RECOG{} is closed under inverse bottom-up tree transformations (in particular under inverse tree homomorphisms, cf. Exercise \ref{exe.3.68}).
\end{corollary}
\begin{proof}
	Let $L\in\RECOG$ and $M\in B$. We have to show that $M^{-1}(L)$ is in the class \RECOG{}. Obviously, if $R$ is the finite tree automaton restriction $\{(t,t)\mid t\in L\}$, then $M^{-1}(L)=\dom(M\comp R)$. By Theorem \ref{the.4.57}, $M\comp R$ is in $B$ and so, by Exercise \ref{exe.4.29}, its domain is in \RECOG{}.
\end{proof}
\begin{remark}\label{rem.4.61}
	Note that, since, for arbitrary tree transformations $M_1$ and $M_2$, $(M_1\comp M_2)^{-1}=M_2^{-1}\comp M_1^{-1}$, Corollary \ref{cor.4.60} implies that if $M$ is the composition of any finite number of bottom-up tree transformations (in particular, elements of \REL{}, \FTA{} and \HOM{}), then \RECOG{} is closed under $M^{-1}$; moreover, since $\dom(M)=M^{-1}(T_\Delta)$, the domain of any such tree transformation $M$ is recognizable.
\end{remark}
We finally note that, because of the composition results of Theorem \ref{the.4.57}, the inclusion signs in Theorems \ref{the.4.52} and \ref{the.4.54} may be replaced by equality signs. It is also easy to prove equations like $\B=\LB\comp\HOM$, $\B=\LB\comp\DB$ or $\B=\REL\comp\DB$ (cf. property~(B)). The first equation has some importance of its own, since it characterizes \B{} without referring to the notion ``bottom-up"; this is so because \LB{} has such a characterization as shown next (we first need a definition).
\begin{definition}\label{def.4.62}
	Let $F$ be a class of tree transformations containing all identity transformations. For each $n\geq1$ we define $F^n$ inductively by $F^1=F$ and $F^{n+1}=F^n\comp F$. Moreover, the \uem{closure of $F$ under composition}, denoted by $F^*$, is defined to be $\bigcup\limits_{n\geq1}F^n$.
	
	Thus $F^*$ consists of all tree transformations of the form $\li*[M_i][n][\comp][\cdots]$, where $n\geq1$ and $M_i\in F$ for all $i$, $1\leq i\leq n$.
\end{definition}
\begin{corollary}\label{cor.4.63}
\[
\begin{array}{rrll}
(1) & \LB & \!\!\!= & \!\!\!(\REL\cup\FTA\cup\LHOM)^*.  \hspace{7.5cm}\\[2mm]
(2) & \B & \!\!\!= & \!\!\!\LB\comp\HOM.
\end{array}
\]
\end{corollary}
\begin{proof}
	The inclusions $\subseteq$ follow from Theorems \ref{the.4.52} and \ref{the.4.54}; the inclusions $\supseteq$ follow from Theorem \ref{the.4.57}.
\end{proof}

\subsection{Decomposition of top-down tree transformations}
We now show that, analogously to the bottom-up case, we can decompose each top-down ftt into two transducers, the first doing the copying and the second doing the rest of the work (cf. property (T)).
\begin{theorem}\label{the.4.64}
	$\T\subseteq\HOM\comp\LT$ and $\DT\subseteq\HOM\comp\LDT$.
\end{theorem}
\begin{proof}
	Let $M=(Q,\Sigma,\Delta,R,Q_d)$ be a top-down ftt. While processing an input tree, $M$ generally makes a lot of copies of input subtrees in order to get different translations of these subtrees. To simulate $M$ we can first use a homomorphism which simply makes as many copies of subtrees as are needed by $M$, and then we can simulate $M$ linearly (since all copies are already there). The formal construction is as follows.
	
	We first determine the ``degree of copying" of $M$, that is the maximal number of copies needed by $M$ in any step of its computation. Thus, for $x\in\X$ and $r\in R$, let $r_x$ be the number of occurrences of $x$ in the right hand side of $r$. Let $n=\max\{r_x\mid x\in\X,\,r\in R\}$. We now let $\Omega$ be the ranked alphabet obtained from $\Sigma$ by multiplying the ranks of all symbols by $n$ (so that a node may be connected to $n$ copies of each of its subtrees). Thus $\Omega_{kn}=\Sigma_k$ for all $k\geq0$. The ``copying" homomorphism $h$ from $T_\Sigma$ into $T_\Omega$ is now defined by
	\begin{enumerate}[label=(\roman*)]
		\item for $a\in\Sigma_0$, $h_0(a)=a$
		\item for $k\geq1$ and $a\in\Sigma_k$, $h_k(a)=\listt*[a][x_i^n]$.
	\end{enumerate}
	(For example, if $k=2$ and $n=3$, then $h_2(a)= a[x_1x_1x_1x_2x_2x_2] $). \\
	Finally the top-down ftt $N=(Q,\Omega,\Delta,R_N,Q_d)$ is defined as follows.
	\begin{enumerate}[label=(\roman*)]
		\item If $q[a]\to t$ is a rule in $R$, then it is also in $R_N$.
		\item Suppose that $q[\listt[a][x_i]]\to t$ is a rule in $R$. Let us denote the variables $\li*[x_i][kn]$ by $x_{1,1},x_{1,2},\dots,x_{1,n},x_{2,1},\dots,x_{2,n},\dots,x_{k,1},\dots,x_{k,n}$ respectively. Then the rule $q[a[x_{1,1}\cdots x_{1,n}\cdots x_{k,1}\cdots x_{k,n}]]\to t'$ is in $R_N$, where $t'$ is taken such that it is linear and such that $t'\langle x_{i,j}\gets x_i\rangle_{\substack{1\leq i\leq k\\1\leq j\leq n}}=t$
	\end{enumerate}
	($t'$ can be obtained by putting different second subscripts on different occurrences of the same variable in $t$. For instance, if $ q[a[x_1x_2]] \to b[q_1[x_2]cd[q_2[x_1]q_3[x_2]]] $ is in $R$ and $n=3$, then we can put the rule $q[a[x_{1,1}x_{1,2}x_{1,3}x_{2,1}x_{2,2}x_{2,3}]]\to b[q_1[x_{2,1}]cd[q_2[x_{1,1}]q_3[x_{2,2}]]]$ in $R_N$.)
	
	Obviously, if $M$ is deterministic, then so is $N$. A formal proof of the fact that $M=h\comp N$ is left to the reader.
\end{proof}
\begin{example}\label{exa.4.65}
	Consider the top-down ftt $M$ of Example \ref{exa.4.33} and let us consider its decomposition according to the above proof. Clearly $n=2$ and therefore the definition of $h$ for, for example, $+$, $*$, $-$, $a$ and $b$ is $h_2(+)= +[x_1x_1x_2x_2] $, $h_2(*)= *[x_1x_1x_2x_2] $, $h_1(-)= -[x_1x_1] $, $h_0(a)=a$ and $h_0(b)=b$. Thus, for example, $h( *[+[ab]-[a]] )= *[+[aabb]+[aabb]-[aa]-[aa]] $. For instance the first three rules of $M$ turn into the following three rules for $N$:
	\begin{align*}
	& q[+[x_{1,1}x_{1,2}x_{2,1}x_{2,2}]]\to +[q[x_{1,1}]q[x_{2,1}]],                           \\
	& q[*[x_{1,1}x_{1,2}x_{2,1}x_{2,2}]]\to +[*[q[x_{1,1}]i[x_{2,1}]]*[i[x_{1,2}]q[x_{2,2}]]], \\
	& q[-[x_{1,1}x_{1,2}]]\to -[q[x_{1,1}]].                                                   
	\end{align*}
	It is left to the reader to see how $N$ processes $h( *[+[ab]-[a]] )$.
\end{example}
Since $\LT\subseteq\LB$ (Theorem \ref{the.4.48}), we know already how to decompose \LT{}. This gives us the following result.
\begin{corollarywithqed}\label{cor.4.66}
	$\T\subseteq\HOM\comp\LB=\HOM\comp\REL\comp\FTA\comp\LHOM$.
\end{corollarywithqed}
Notice that the inclusion is proper by Lemma \ref{lem.4.46}.
From this corollary we see that each top-down tree transformation can be realized by the composition of two bottom-up tree transformations (cf. Theorem \ref{the.4.43}). Another way of expressing our decomposition results concerning \B{} and \T{} is by the equation $\B^{\,*}=\T^{\,*}=(\REL\cup\FTA\cup\HOM)^*$.

By the above corollary and Remark \ref{rem.4.61} we obtain that \RECOG{} is closed under inverse top-down tree transformations, and in particular
\begin{corollarywithqed}\label{cor.4.67}
	The domain of a top-down tree transformation is recognizable.
\end{corollarywithqed}
The next theorem says, analogously to Theorem \ref{the.4.52}, that each deterministic element of \LT{} can be decomposed into two simpler (deterministic) ones.
\begin{theorem}\label{the.4.68}
	$\LDT\subseteq\DTQREL\comp\LHOM$.
\end{theorem}
\begin{proof}
	See the proof of Theorem \ref{the.4.52}.
\end{proof}
We now note that we cannot obtain very nice results about closure under composition analogously to those of Theorem \ref{the.4.57}, since for instance \LT{} and \DT{} are not closed under composition (see Corollary \ref{cor.4.47}). The reason is essentially the failure of property~(B$'$) for top-down ftt. One could get closure results by eliminating both properties (B) and~(B$'$). For example, one can show that $\D_tT\comp\T\subseteq\T$, $\T\comp\NLT\subseteq\T$, etc. However, after having compared \DB{} with \DT{} in the next section, we prefer to extend the top-down ftt in such a way that it has the capability of checking before deletion. It will turn out that the so extended top-down transducer has all the nice properties ``dual" to those of \B{}.

The following is an easy exercise in composition.
\begin{exercise}\label{exe.4.69}
	Prove that every \T-surface tree language is in fact the image of a rule tree language under a top-down tree transformation.
\end{exercise}

\subsection{Comparison of B and T, the deterministic case}
Although, by determinism, some differences between \B{} and \T{} are eliminated, \DB{} and \DT{} are still incomparable for a number of reasons. We first discuss the question why \DB{} contains elements not in \DT{}, and then the reverse question.

Firstly, we have seen in Lemma \ref{lem.4.46} that \DB~contains elements not in \T{}. This was caused by property (B$'$).

Secondly, we have seen in Theorem \ref{the.3.14} (together with Theorem \ref{the.3.8}) that there are deterministic bottom-up recognizable tree languages which cannot be recognized deterministically top-down. It is easy to see that the corresponding fta restrictions cannot be realized by a deterministic top-down transducer. In fact the following can be shown.
\begin{exercise}\label{exe.4.70}
	Prove that the domain of a deterministic top-down ftt can be recognized by a deterministic top-down fta.
\end{exercise}
Thirdly, there is a trivial reason that \DB{} is stronger than \DT{}: a bottom-up ftt can, for example, recognize the ``lowest" occurrence of some symbol in a tree (since it is the first occurrence), whereas a deterministic top-down ftt cannot (since for him it is the last occurrence).
\begin{lemmawithoutqed}\label{lem.4.71}
	There is a tree transformation in \DBQREL{} which is not in \DT{}.
\end{lemmawithoutqed}
\begin{proof}
	Let $\Sigma_0=\{b\}$, $\Sigma_1=\{a,f\}$, $\Delta_0=\{b\}$ and $\Delta_1=\{a,\mbar{a},f\}$. Consider the bottom-up ftt $M=(Q,\Sigma,\Delta,R,Q_d)$ where $Q=Q_d=\{q_1,q_2\}$ and $R$ consists of the rules
	\begin{align*}
	b&\to  q_1[b] ,\\
	a[q_1[x]]  & \to q_1[\mbar{a}[x]], &  f[q_1[x]]  & \to q_2[f[x]] , \\
	a[q_2[x]]  & \to q_2[a[x]] ,        &  f[q_2[x]]  & \to q_2[f[x]] , 
	\end{align*}
	where $x$ denotes $x_1$.
	
	Thus $M$ is a deterministic bottom-up finite state relabeling which bars all $a$'s below the lowest $f$. Obviously a deterministic top-down ftt cannot give the right translation for both $m(a^nb)$ and $m(a^nfb)$.
\end{proof}
Let us now consider \DT{}.

Firstly, we note that property (T) has not been eliminated: a deterministic top-down ftt still has the ability to copy an input subtree and continue the translation of these copies in different states. Consider for example the tree transformation $M=\{(m(ab^nc),a[m(p^nc)m(q^nc)])\mid n\geq0\}$.
Obviously $M$ is in \DT{}. It can be shown, similarly to the proof of Theorem \ref{the.4.43}(2), that $M$ is not in \B{} (see also Exercise \ref{exe.4.51}).

Secondly, deterministic bottom-up ftt cannot distinguish between left and right (because they start at the bottom!), whereas deterministic top-down ftt can. Consider for example the tree transformation
$M=\{(a[m(b^nc)m(b^kc)],a[m(b^nd)m(b^ke)])\mid n,k\geq0\}$.
This is obviously an element of \DT{} (even \DTQREL{}) and can be shown not to be in \DB{} (of course it \uem{is} in \uem{\B}: a nondeterministic bottom-up ftt can guess whether it is left or right and check its guess when arriving at the top). Thus we have the following lemma.
\begin{lemma}\label{lem.4.72}
	There is a tree transformation in \DTQREL{} which is not in \DB{}.
\end{lemma}
Thirdly, there is a proof of this lemma which is analogous to the one of Lemma~\ref{lem.4.71}: there is a deterministic top-down finite state relabeling that bars all $a$'s above the highest~$f$, and this cannot be done by an element of \DB{}.

\subsection{Top-down finite tree transducers with regular look-ahead}
One way to take away the advantages of \DB{} over \DT{} is to allow the top-down tree transducer to have a look-ahead: that is, the ability to inspect an input subtree and, depending on the result of that inspection, decide which rule to apply next. Moreover it seems to be sufficient (and natural) that this look-ahead ability should consist of inspecting whether the input subtree belongs to a certain recognizable tree language or not (in other words, checking whether it has a certain ``recognizable property"). As a result of this capability the top-down tree transducer would first of all have property~(B$'$): it can check a recognizable property of a subtree and decide whether to delete it or not. Secondly, the domain of a deterministic top-down tree transducer would be arbitrary recognizable (it just starts by checking whether the whole input tree belongs to the recognizable tree language). And, thirdly, a deterministic top-down tree transducer would for instance be able to see the ``lowest" occurrence of some symbol in a tree (it just checks whether the subtree beneath the symbol contains another occurrence of the same symbol, and that is a recognizable property).

We now formally define the top-down transducer with regular (= recognizable) look-ahead. It turns out that the look-ahead feature can be expressed easily in a tree rewriting system (see Definition \ref{def.4.12}): for each rule we specify the ranges of the variables in the rule to be certain recognizable tree languages (such a rule is then applicable only if the corresponding input subtrees belong to these recognizable tree languages).
\begin{definition}\label{def.4.73}
	A \uem{top-down (finite) tree transducer with regular look-ahead} is a structure $M=(Q,\Sigma,\Delta,R,Q_d)$, where $Q$, $\Sigma$, $\Delta$ and $Q_d$ are as for the ordinary top-down ftt and $R$ is a finite set of rules of the form $(t_1\to t_2,\,D)$, where $t_1\to t_2$ is an ordinary top-down ftt rule and $D$ is a mapping from $\X_k$ into $\mathcal{P}(T_\Sigma)$ (where $k$ is the number of variables in $t_1$) such that, for $1\leq i\leq k$, $D(x_i)\in\RECOG$. (Whenever $D$ is understood or will be specified later we write $t_1\to t_2$ rather than $(t_1\to t_2,\,D)$. We call $t_1$ and $t_2$ the left hand side and right hand side of the rule respectively).
	
	$M$ is viewed as a tree rewriting system in the obvious way, 
$(t_1\to t_2,\,D)$ being a ``rule scheme" $(t_1,t_2,D)$.
	
	The \uem{tree transformation realized by $M$}, denoted by $T(M)$ or $M$, is $\{(s,t)\in T_\Sigma\times T_\Delta\mid  q[s] \xRightarrow[M]{*}t\text{ for some }q\in Q_d\}$.
\end{definition}
Thus a top-down ftt with regular look-ahead works in exactly the same way as an ordinary one, except that the application of each of its rules is restricted: the (input sub-)trees substituted in the rule should belong to prespecified recognizable tree languages. Note that for rules of the form $ q[a] \to t$ the mapping $D$ need not be specified.
\begin{notation}\label{not.4.74}
	The phrase ``with regular look-ahead" will be indicated by a prime. Thus the class of top-down tree transformations with regular look-ahead will be denoted by \TR{}. An element of \TR{} is also called a top-down$'$ tree transformation.
\end{notation}
\begin{example}\label{exa.4.75}
	Consider the tree transformation $M$ in the proof of Lemma \ref{lem.4.46}. It can be realized by the top-down$'$ ftt $N=(Q,\Sigma,\Delta,R,Q_d)$ where $Q=\{q_0,q\}$, $Q_d=\{q_0\}$ and $R$ consists of the following rules:
	\begin{align*}
	q_0[a[x_1x_2]]  & \to a[q[x_1]] \text{ with ranges }D(x_1)=\{m(b^nc)\mid n\geq0\}\text{ and }D(x_2)=\{c\}, \\
	q[b[x_1]]       & \to b[q[x_1]] \text{ with }D(x_1)=T_\Sigma\text{, and }                                  \\
	q[c]            & \to c~.
	\end{align*}
	
	Note that, in the first rule, $D(x_1)$ could as well be $T_\Sigma$ since it is checked later by $N$ that the left subtree contains no $a$'s. The essential use of regular look-ahead in this example is the restriction of the right subtree to $\{c\}$.
\end{example}
We now define some subclasses of \TR.
\begin{definition}\label{def.4.76}
	Let $M=(Q,\Sigma,\Delta,R,Q_d)$ be a top-down$'$ ftt.
	
	The definitions of \uem{linear}, \uem{nondeleting} and \uem{one-state} are identical to the bottom-up and top-down ones (see Definition \ref{def.4.23}).
	
	$M$ is called (partial) \uem{deterministic} if the following holds.
	\begin{enumerate}[label=(\roman*)]
		\item $Q_d$ is a singleton.
		\item If $(s\to t_1,D_1)$ and $(s\to t_2,D_2)$ are different rules in $R$ (with the same left hand side), then $D_1(x_i)\cap D_2(x_i)=\emptyset$ for some $i$, $1\leq i\leq k$ (where $k$ is the number of variables in $s$).\qedhere
	\end{enumerate}
\end{definition}
Since the ranges of the variables are recognizable, it can effectively be determined whether a top-down$'$ ftt is deterministic (if, of course, these ranges are effectively specified, which we always assume).

Notation \ref{not.4.24} also applies to \TR{}. Thus \LDTR{} is the class of linear deterministic top-down tree transformations with regular look-ahead.

Observe that, obviously, $\T\subseteq\TR$ since each top-down ftt can be transformed trivially 
into a top-down$'$ ftt by specifying all ranges of all variables in all rules to be the (recognizable) tree language $T_\Sigma$ (where $\Sigma$ is the input alphabet). Moreover, if \Z{} is a modifier, then $\ZT\subseteq\ZTR$.
\\

It can easily be seen that Theorems \ref{the.4.43} and \ref{the.4.44} still hold with \T{} replaced by \TR{}. In fact, the proofs of the theorems are true without further change.

In the next theorem we show that the regular look-ahead can be ``taken out" of a top-down$'$ ftt.
\begin{theorem}\label{the.4.77}
	$\TR\subseteq\DBQREL\comp\T$ and $\ZTR\subseteq\DBQREL\comp\ZT$ for $\Z\in\{\LIN,\D,\LD\}$.
\end{theorem}
\begin{proof}
	Let $M=(Q,\Sigma,\Delta,R,Q_d)$ be in \TR{}. Consider all recognizable properties which $M$ needs for its look-ahead (that is, all recognizable tree languages $D(x_i)$ occurring in the rules of $M$). We can use a total deterministic bottom-up finite state relabeling to check, for a given input tree $t$, whether the subtrees of $t$ have these properties or not, and to put at each node a (finite) amount of information telling us whether the direct subtrees of that node have the properties or not. After this relabeling we can use an ordinary top-down ftt to simulate $M$, because the look-ahead information is now contained in the label of each node.
	
	The formal construction might look as follows. Let $\li[L_i][n]$ be all the recognizable tree languages occurring as ranges of variables in the rules of $M$. Let $U$ denote the set $\{0,1\}^n$, that is, the set of all sequences of 0's and 1's of length $n$. The $j$\textsuperscript{th} element of $u\in U$ will be denoted by $u^j$. An element $u$ of $U$ will be used to indicate whether a tree belongs to $\li[L_i][n]$ or not ($u^j=1$ iff the tree is in $L_j$).
	
	We now introduce a new ranked alphabet $\Omega$ such that $\Omega_0=\Sigma_0$ and, for $k\geq1$, $\Omega_k=\Sigma_k\times U^k$. Thus an element of $\Omega_k$ is of the form $(a,(\li[u_i]))$ with $a\in\Sigma_k$ and $\li[u_i]\in U$. If a node is labeled by such a symbol, it will mean that $u_i$ contains all the information about the $i$\textsuperscript{th} subtree of the node.
	
	Next we define the mapping $f:T_\Sigma\to T_\Omega$ as follows:
	\begin{enumerate}[label=(\roman*)]
		\item for $a\in\Sigma_0$, $f(a)=a$;
		\item for $k\geq1$, $a\in\Sigma_k$ and $\li\in T_\Sigma$,\\
		$f(\listt)=\listt[b][f(t_i)]$, where $b=(a,(\li[u_i]))$ and, for $1\leq i\leq k$ and $1\leq j\leq n$, $u_i^j=1$ iff $t_i\in L_j$.
	\end{enumerate}
	It is left as an \uem{exercise} to show that $f$ can be realized by a (total) deterministic bottom-up finite state relabeling (given the deterministic bottom-up fta's recognizing $\li[L_i][n]$).
	
	We now define a top-down ftt $N=(Q,\Omega,\Delta,R_N,Q_d)$ such that
	\begin{enumerate}[label=(\roman*)]
		\item if $ q[a] \to t$ is in $R$, then it is in $R_N$;
		\item if $( q[\listt[a][x_i]] \to t,D)$ is in $R$, then each rule of the form $q[\listt[(a,\mbar{u})][x_i]]\to t$ is in $R_N$, where $\mbar{u}=(\li[u_i])\in U^k$ and $\mbar{u}$ satisfies the condition: if $D(x_i)=L_j$ then $u_i^j=1$ (for all $i$ and $j$, $1\leq i\leq k$, $1\leq j\leq n$).
	\end{enumerate}
	This completes the construction. It is obvious from this construction that $M=f\comp N$. Moreover, if $M$ is linear, then so is $N$. It is also clear that, in the construction above, the set $U$ may be replaced by the set $\{u\in U\mid$ for all $j_1\text{ and }j_2$, if $L_{j_1}\cap L_{j_2}=\emptyset$, then $u^{j_1}\cdot u^{j_2}\neq1 \}$ (note that this influences $R_N$: rules containing elements not in this set are removed). One can now easily see that if $M$ is deterministic, then so is $N$.
\end{proof}
An immediate consequence of this theorem and previous decomposition results is that each element of \TR{} is decomposable into elements of \REL{}, \FTA{} and \HOM{}.
\begin{corollary}\label{cor.4.78}
	The domain of a top-down ftt with regular look-ahead is recognizable.
\end{corollary}
\begin{proof}
	For instance by Remark \ref{rem.4.61}.
\end{proof}
Another consequence of Theorem \ref{the.4.77} is that the addition of regular look-ahead has no influence on the surface tree languages.
\begin{corollary}\label{cor.4.79}
	$\text{\TR-Surface}=\text{\T-Surface}$ and $\text{\DTR-Surface}=\text{\DT-Surface}$.
\end{corollary}
\begin{proof}
	Let $L$ be a (\D)\TR{}-surface tree language, so $L=M(L_1)$ for some $M\in (\D)\TR$ and $L_1\in\RECOG$.
	Now, by Theorem \ref{the.4.77}, $M=R\comp N$ for some $R\in\DBQREL$ and $N\in(\D)\T$. Hence $L=N(R(L_1))$. Since \RECOG{} is closed under linear bottom-up tree transformations (Corollary \ref{cor.4.55}), $N(R(L_1))$ is a (\D)\T{}-surface tree language.
\end{proof}
We now show that, in the linear case, there is no difference between bottom-up and top-down$'$ tree transformations (all properties (B), (T) and (B$'$) are ``eliminated"), cf. Theorem \ref{the.4.48}.
\begin{theorem}\label{the.4.80}
	$\LTR=\LB$.
\end{theorem}
\begin{proof}
	First of all we have
	\begin{align*}
	\LTR & \subseteq\DBQREL\comp\LT &   & \text{(Theorem \ref{the.4.77})} \\
	& \subseteq\DBQREL\comp\LB &   & \text{(Theorem \ref{the.4.48})} \\ 
	& \subseteq\LB             &   & \text{(Theorem \ref{the.4.57})}.   
	\end{align*}
	
	Let us now prove that $\LB\subseteq\LTR$. The construction is the same as in the proof of Theorem \ref{the.4.48}(2), but now we use look-ahead in case the bottom-up ftt is deleting. Let $M=(Q,\Sigma,\Delta,R,Q_d)$ be a linear bottom-up ftt. Define, for each $q$ in~$Q$, $M_q$ to be the bottom-up ftt $(Q,\Sigma,\Delta,R,\{q\})$. Construct the linear top-down$'$ ftt $N=(Q,\Sigma,\Delta,R_N,Q_d)$ such that
	\begin{enumerate}[label=(\roman*)]
		\item if $a\to q[t] $ is in $R$, then $ q[a] \to t$ is in $R_N$;
		\item if $\listt[a][q_i[x_i]]\to q[t] $ is in $R$, then the rule\\ 
$ q[\listt[a][x_i]] \to t\langle\li[x_i\gets q_i[x_i]]\rangle$ is in $R_N$, where, for $1\leq i\leq k$, \\
if $x_i$ does not occur in $t$, then $D(x_i)=\dom(M_{q_i})$, and $D(x_i)=T_\Sigma$ otherwise.
	\end{enumerate}
Note that $\dom(M_{q_i})$ is recognizable by Exercise \ref{exe.4.29}.
	
	It should be clear that $T(N)=T(M)$.
\end{proof}
From the proof of Theorem \ref{the.4.80} it follows that each element of \LTR{} can be realized by a linear top-down$'$ ftt with the property that look-ahead is only used on subtrees which are deleted. This property corresponds precisely to property (B$'$).

Let us now consider composition of top-down$'$ ftt. Analogously to the bottom-up case, we can now expect the results in the next theorem from property (B).
\begin{theorem}\label{the.4.81}
\[
\begin{array}{rrlllrll}
(1) & \TR\comp\LTR & \!\!\!\subseteq & \!\!\!\!\TR. & & & & \hspace{5cm}  \\[2mm]
(2) & \DTR\comp\TR & \!\!\!\subseteq & \!\!\!\!\TR & \text{and} & \DTR\comp\DTR & \!\!\!\subseteq & \!\!\!\!\DTR.
\end{array}
\]
\end{theorem}
\vspace*{-1cm}
\begin{proof}
	As in the bottom-up case, we only consider a number of special cases.
	\begin{lemma*}
		$\TR\comp\LHOM\subseteq\TR$ and $\DTR\comp\HOM\subseteq\DTR$.
	\end{lemma*}
	\begin{proof}
		Let $M=(Q,\Sigma,\Delta,R,Q_d)$ be a top-down$'$ ftt and $h$ a tree homomorphism from $T_\Delta$ into $T_\Omega$. We construct a new top-down$'$ ftt using the old idea of applying the homomorphism to the right hand sides of the rules of $M$. Therefore we extend $h$ to trees in $T_\Delta(Q[\X])$ by defining $h_0(x)=x$ for $x$ in $\X$ and $h_1(q)= q[x_1] $ for $q\in Q$. Let, for $q\in Q$, $M_q=(Q,\Sigma,\Delta,R,\{q\})$. Note that, by Corollary \ref{cor.4.78}, $\dom(M_q)$ is recognizable.
		
		Construct now $N=(Q,\Sigma,\Omega,R_N,Q_d)$ such that
		\begin{enumerate}[label=(\roman*)]
			\item if $ q_0[a] \to t$ is in $R$, then $ q_0[a] \to h(t)$ is in $R_N$;
			\item if $( q_0[\listt[a][x_i]] \to t,D)$ is in $R$, then $( q_0[\listt[a][x_i]] \to h(t),\mbar{D})$ is in $R_N$, where, for $1\leq i\leq k$, $\mbar{D}(x_i)$ is the intersection of $D(x_i)$ and all tree languages $\dom(M_q)$ such that $q[x_i]$ occurs in $t$ but not in $h(t)$.
		\end{enumerate}
		
		Thus $N$ simultaneously simulates $M$ and applies $h$ to the output of $M$. But, whenever $M$ starts making a translation of a subtree $t$ starting in a state $q$, and this translation is deleted by $h$, $N$ checks that $t$ is translatable by $M_q$.
		
		If $h$ is linear or $M$ is deterministic, then $N=M\comp h$. Obviously, if $M$ is deterministic, then so is $N$.
	\end{proof}
	\begin{lemma*}
		\begin{align*}
		\TR\comp\QREL    & \subseteq\TR,   \hspace{9cm}         \\
		\DTR\comp\DTQREL & \subseteq\DTR~~\text{and} \\
		\DTR\comp\DBQREL & \subseteq\DTR.           
		\end{align*}
	\end{lemma*}
	\begin{proof}
		It is not difficult to see that, given $M\in\TR$ and a top-down finite state relabeling $N$, the composition of these two can be realized by one top-down$'$ ftt $K$, which at each step simultaneously simulates $M$ and $N$ by transforming the output of $M$ according to $N$. The construction is basically the same as that in the second lemma of Theorem \ref{the.4.57}. It is also clear that if $M$ and $N$ are both deterministic, then so is $K$. This gives us the first two statements of the lemma. The third one is more difficult since the finite state relabeling is now bottom-up. However the same kind of construction can be applied again, using the look-ahead facility to make the resulting top-down$'$ ftt deterministic.
	
		Let $M=(Q,\Sigma,\Delta,R,Q_d)$ with $Q_d=\{q_d\}$ be in \DTR{} and 
let $N=(Q_N,\Delta,\Omega,R_N,Q_{dN})$ be in \DBQREL{}. 
Without loss of generality we assume that $q_d$ does not occur in the right hand sides of the rules in $R$. 
Let, as usual, for $q\in Q$, $M_q=(Q,\Sigma,\Delta,R,\{q\})$, and, for $p\in Q_N$, $N_p=(Q_N,\Delta,\Omega,R_N,\{p\})$. We shall realize $M\comp N$ by a deterministic top-down$'$  ftt $K=(Q,\Sigma,\Omega,R_K,Q_d)$, where the set $R_K$ of rules is determined as follows.

Let $q\in Q_M$ and $p\in Q_N$, such that if $q=q_d$ then $p\in Q_{dN}$. 
		\begin{enumerate}[label=(\roman*)]
			\item If $ q[a] \to t$ is in $R$ and $t\xRightarrow[N]{*}p[t']$, 
then $q[a] \to t'$ is in $R_K$.
			\item Let $( q[\listt[a][x_i]] \to t,\,D)$ be a rule in $R$. Then $t$ can be written as $t=s\langle x_1\gets q_1[x_{i_1}],\dots,x_m\gets q_m[x_{i_m}]\rangle$ for certain $m\geq0$, $s\in T_\Delta(\X_m)$, $\li[q_i][m]\in Q$ and $x_{i_1},\dots,x_{i_m}\in\X_k$, 
such that $s$ is nondeleting w.r.t.\ $\X_m$. 
			
			Let $\li[p_i][m]$ be a sequence of $m$ states of $N$ such that $s\langle\li[x_i\gets p_i[x_i]][m]\rangle\xRightarrow[N]{*}p[s']$, where $s'\in T_\Omega(\X_m)$. (Of course $N$ was first extended in the usual way).
			
			Then the rule $q[\listt[a][x_i]] \to s'\langle x_1\gets q_1[x_{i_1}],\dots,x_m\gets q_m[x_{i_m}]\rangle$ is in $R_K$, where \\[0.5mm]
the ranges of the variables are specified by $\mbar{D}$ as follows. For $1\leq u\leq k$, $\mbar{D}(x_u)$ is the intersection of $D(x_u)$ and all tree languages $\dom(M_{q_j}\comp N_{p_j})$ such that $x_{i_j}=x_u$.
		\end{enumerate}
		This ends the construction. 
Intuitively, when $K$ arrives at the root of an input subtree $a[t_1\cdots t_k]$ (in the same state as $M$), 
it first uses its regular look-ahead to determine the rule applied by $M$ and to determine,
for every $1\leq j\leq m$, the (unique) state $p_j$ in which $N$ will arrive after translation of the $M_{q_j}$-translation of $t_{i_j}$. It then runs $N$ on the piece of output of $M$, starting in the states $p_1,\ldots,p_m$, and produces the output of $N$ as its piece of output.
It is straightforward to check formally that $K$ is a deterministic top-down$'$ ftt (using the determinism of both $M$ and $N$).
	\end{proof}
	
	We now complete the proof of Theorem \ref{the.4.81}.\\
	Firstly
	\begin{align*}
	\TR\comp\LTR & =\TR\comp\LB                    &   & \text{(Theorem \ref{the.4.80})} \\
	& \subseteq\TR\comp\QREL\comp\LHOM &   & \text{(Theorem \ref{the.4.52})}  \\
	& \subseteq\TR\comp\LHOM          &   & \text{(second lemma)}               \\
	& \subseteq\TR                    &   & \text{(first lemma)}.                
	\end{align*}
	Secondly
	\begin{align*}
	\DTR\comp(\D)\TR & \subseteq\DTR\comp\DBQREL\comp(\D)\T  &   & \text{(Theorem \ref{the.4.77})} \\
	& \subseteq\DTR\comp(\D)\T              &   & \text{(second lemma)}                 \\
	& \subseteq\DTR\comp\HOM\comp\LIN(\D)\T &   & \text{(Theorem \ref{the.4.64})}   \\
	& \subseteq\DTR\comp\LIN(\D)\T          &   & \text{(first lemma)}.                  
	\end{align*}
	Now $\DTR\comp\LT\subseteq\TR$ (by (1) of this theorem) and
	\begin{align*}
	\DTR\comp\LDT & \subseteq\DTR\comp\DTQREL\comp\LHOM &   & \text{(Theorem \ref{the.4.68})} \\
	& \subseteq\DTR\comp\LHOM             &   & \text{(second lemma)}               \\
	& \subseteq\DTR                       &   & \text{(first lemma)}.                
	\end{align*}
	This proves Theorem \ref{the.4.81}.
\end{proof}
Note that the ``right hand side" of Theorem \ref{the.4.81}(2) states that \DTR{} is closed under composition. Clearly we know already that \LTR{} is closed under composition (since $\LTR=\LB$). It can also easily be checked from the proof of Theorem \ref{the.4.81} that \LDTR{} is closed under composition.

We can now show that indeed regular look-ahead has made \DT{} stronger than \DB{}.
\begin{corollary}\label{cor.4.82}
	$\DB\subsetneq\DTR$.
\end{corollary}
\begin{proof}
	By Theorem \ref{the.4.52}, $\DB\subseteq\DBQREL\comp\HOM$. Hence, since the identity tree transformation is in \DTR{}, we trivially have $\DB\subseteq\DTR\comp\DBQREL\comp\HOM$. But, by the second and the first lemma in the proof of Theorem \ref{the.4.81}, $\DTR\comp\DBQREL\comp\HOM\subseteq\DTR$. Hence $\DB\subseteq\DTR$. Proper inclusion follows from Lemma \ref{lem.4.72}.
\end{proof}
\begin{exercise}\label{exe.4.83}
	Show that each \TR{}-surface tree language is the range of some element of \TR{}.
\end{exercise}
\begin{exercise}\label{exe.4.84}
	Prove that \T{}-Surface is closed under linear tree transformations (recall Corollary \ref{cor.4.79}).
	
	Prove that \DT{}-Surface is closed under deterministic top-down and bottom-up tree transformations.
\end{exercise}
It follows from Theorem \ref{the.4.81} that the inclusion signs in Theorem \ref{the.4.77} may be replaced by equality signs. Hence, for example, $\DTR=\DBQREL~\comp~\DT=\DB~\comp~\DT$. We finally show a result ``dual" to Corollary \ref{cor.4.63}(2) (recall that $\LTR=\LB$).
\begin{theorem}\label{the.4.85}
	$\TR=\HOM\comp\LTR$.
\end{theorem}\begin{proof}
The inclusion $\HOM\comp\LTR\subseteq\TR$ is immediate from Theorem \ref{the.4.81}. The inclusion $\TR\subseteq\HOM\comp\LTR$ can be shown in exactly the same way as in the proof of $\T\subseteq\HOM\comp\LT$ (Theorem \ref{the.4.64}). The only problem is the regular look-ahead: the image of a recognizable tree language under the homomorphism $h$ need not be recognizable (we use the notation of the proof of Theorem \ref{the.4.64}). The solution is to consider a homomorphism $g$ from $T_\Omega$ into $T_\Sigma$ such that, for all $t$ in $T_\Sigma$, $g(h(t))=t$; $g$ is easy to find. Now, whenever, in a rule of $M$, we have a recognizable tree language $L$ as look-ahead (for some variable), we can use $g^{-1}(L)$ as the look-ahead in the corresponding rule of $N$. The details are left to the reader.
\end{proof}

\subsection{Surface and target languages}
In this section we shall consider a few properties of the tree (and string) languages which are obtained from the recognizable tree languages by application of a finite number of tree transducers. In other words we shall consider the classes

$(\REL\cup\FTA\cup\HOM)^*\Surface*$, briefly denoted by \uem{\Surface}, and

$(\REL\cup\FTA\cup\HOM)^*\Target*$, briefly denoted by \uem{\Target}.

\noindent
Note that $\Target = \fyield(\Surface)$. Note also that, by various composition results, $(\REL\cup\FTA\cup\HOM)^*=\T^{\,*}=(\TR)^*=\B^{\,*}=(\TR\cup\B)^*=\dots$ etc.

Let us first consider some classes of tree languages obtained by restricting the number of transducers applied. In particular, let us consider, for each $k\geq1$, the classes $\T^{\,k}$\Surface*, $(\TR)^k$\Surface* and $\B^{\,k}$\Surface*. Obviously, by the above remark, $\text{Surface} =\bigcup\limits_{k\geq1} \text{$\T^{\,k}$\Surface*}=
\bigcup\limits_{k\geq1} \text{$(\TR)^k$\Surface*}=
\bigcup\limits_{k\geq1} \text{$\B^{\,k}$\Surface*}$.

As a corollary to previous results we can show that regular look-ahead has no influence on the class of surface languages (cf. Corollary \ref{cor.4.79}).

\begin{corollary}\label{cor.4.86}
	For all $k\geq1$,
	\begin{enumerate}[label=(\roman*)]
		\item $(\TR)^k=\DBQREL\comp\T^{\,k}$,
		\item $(\TR)^k$\Surface*${}=\T^{\,k}$\Surface*, and
		\item $\T^{\,k}$\Surface* is closed under linear tree transformations.
	\end{enumerate}
\end{corollary}
\begin{proof}~

	(i)~ By Theorem \ref{the.4.77}, $\TR\subseteq\DBQREL\comp\T$. Also, by Corollary \ref{cor.4.82} and Theorem \ref{the.4.81}(2), $\DBQREL\comp\T\subseteq\TR$. Hence $\TR=\DBQREL\comp\T$.
	
	We now show that $\TR\comp\TR=\TR\comp\T$. Trivially, $\TR\comp\T\subseteq\TR\comp\TR$. Also $\TR\comp\TR=\TR\comp\DBQREL\comp\T$ and, by Theorem \ref{the.4.81}(1), $\TR\comp\DBQREL\subseteq\TR$. Hence $\TR\comp\TR\subseteq\TR\comp\T$.
	
	From this it is straightforward to see that $(\TR)^{k+1}=\TR\comp\T^{\,k}=\DBQREL\comp\T\comp\T^{\,k}=$ $\DBQREL\comp\T^{\,k+1}$. 
	
	(ii)~ This is an immediate consequence of (i) and the fact that \RECOG{} is closed under linear tree transformations.
	
	(iii)~ This follows easily from (ii) and the fact that $(\TR)^k$ is closed under composition with linear tree transformations (Theorem \ref{the.4.81}(1), recall also Theorem \ref{the.4.80}).
\end{proof}

We mention here that it can also be shown that $\T^{\,k}\Surface*$ is closed under union, tree concatenation and tree concatenation closure. From that a lot of closure properties of $\T^{\,k}\Target*$ and \Target{} follow.

The relation between the top-down and the bottom-up surface tree languages is easy.

\begin{corollary}\label{cor.4.87}
	For all $k\geq1$,
	\begin{enumerate}[label=(\roman*)]
		\item $\T^{\,k}\Surface*=(\B^{\,k}\comp\LB)\Surface*$ and \\[1mm]
		      $\B^{\,k+1}\Surface*=(\T^{\,k}\comp\HOM)\Surface*$;
		\item $\B^{\,k}\Surface*\subseteq\T^{\,k}\Surface*\subseteq\B^{\,k+1}\Surface*$.
	\end{enumerate}
\end{corollary}
\begin{proof}
	(i) follows from the fact that $\B=\LB\comp\HOM$ (Corollary \ref{cor.4.63}(2)) and that $\TR=\HOM\comp\LB$ (Theorem \ref{the.4.85}).
	
	(ii) is an obvious consequence of (i).
	
	Note that $\B\Surface*=\HOM\Surface*$ (Exercise \ref{exe.4.56}).
\end{proof}

It is not known, but conjectured, that for all $k$ the inclusions in Corollary \ref{cor.4.87}(ii) are proper.

Note that, by taking yields, Corollary \ref{cor.4.87} also holds for the corresponding target languages. Again it is not known whether the inclusions are proper.\\

In the rest of this section we show that the emptiness-, the membership- and the finiteness-problem are solvable for \Surface{} and \Target.

\begin{theorem}\label{the.4.88}
	The emptiness- and membership-problem are solvable for \Surface.
\end{theorem}
\begin{proof}
	Let $M\in(\REL\cup\FTA\cup\HOM)^*$ and $L\in\RECOG$. Consider the tree language $M(L)\in \Surface$. Obviously, $M(L)=\emptyset$ iff $L\cap\dom(M)=\emptyset$. But, by Remark \ref{rem.4.61}, $\dom(M)$ is recognizable. Hence, by Theorem \ref{the.3.32} and Theorem \ref{the.3.74}, it is decidable whether $L\cap\dom(M)=\emptyset$.
	
	To show solvability of the membership-problem note first that \Surface{} is closed under intersection with a recognizable tree language (if $R\in\RECOG$ and $\widehat{R}$ is the fta restriction such that $\dom(\widehat{R})=R$, then $M(L)\cap R=(M\comp \widehat{R})(L)$). Now, for any tree $t$, $t\in M(L)$ iff $M(L)\cap\{t\}\neq\emptyset$. Since $\{t\}$ is recognizable, $M(L)\cap\{t\}\in\Surface$, and we just showed that it is decidable whether a surface tree language is empty or not.
\end{proof}

To show decidability of the finiteness-problem we shall use the following result.
\begin{lemmawithoutqed}\label{lem.4.89}
	Each monadic tree language in \Surface{} is recognizable (and hence regular as a set of strings).
\end{lemmawithoutqed}
\begin{proof}
	Let $L$ be a monadic tree language. Obviously it suffices to show that, for each $k\geq1$, if $L\in(\TR)^k$\Surface*, then $L\in(\TR)^{k-1}$\Surface* (where, by definition, $(\TR)^0\Surface*=\RECOG$). Suppose therefore that $L\in(\TR)^k$\Surface*, so $L=(M_1\comp\cdots\comp M_{k-1}\comp M_k)(R)$ for certain $M_i\in\TR$ and $R\in\RECOG$. Consider all right hand sides of rules in $M_k$. Obviously, since $L$ is monadic, these right hand sides do not contain elements of rank${}\geq2$, that is, they are monadic in the sense of Notation~\ref{not.4.42} (rules which have nonmonadic right hand sides may be removed). But from this it follows that $M_k$ is linear. It now follows from Theorem \ref{the.4.81}(1) (and Corollary \ref{cor.4.55} and Theorem \ref{the.4.80} in the case $k=1$), that $L\in(\TR)^{k-1}$\Surface*.
\end{proof}

\begin{theorem}\label{the.4.90}
	The finiteness-problem is solvable for \Surface.
\end{theorem}
\begin{proof}
	Intuitively, a tree language is finite if and only if the set of paths through this tree language is finite. For $L\subseteq T_\Sigma$ we define $\fpath(L)\subseteq\Sigma^*$ recursively as follows
	\begin{enumerate}[label=(\roman*)]
		\item for $a\in\Sigma_0$, $\fpath(a)=\{a\}$;
		\item for $k\geq1$, $a\in\Sigma_k$ and $\li\in T_\Sigma$,\\
		      $\fpath(\listt)=\{a\}\cdot(\li[\fpath(t_i)][k][\cup])$;
		\item for $L\subseteq T_\Sigma$, $\fpath(L)=\bigcup\limits_{t\in L}\fpath(t)$.
	\end{enumerate}
	Thus, $\fpath(t)$ consists of all strings which can be obtained by following a path through $t$ (for instance, if $t=a[bb[cc]]$, then $\fpath(t)=\{ab,abc\}$). We remark here that any other similar definition of ``path" would also satisfy our purposes.
	It is left to the reader to show that, for any tree language $L$, $L$ is finite iff $\fpath(L)$ is finite.
	
	We now show that, given $L\subseteq T_\Sigma$, the set $\fpath(L)$ can be obtained by feeding $L$ into a top-down tree transducer $M_p$: $M_p(L)$ will be equal to $\fpath(L)$, modulo the correspondence between strings and monadic trees (see Definition \ref{def.2.21}). In fact, $M_p=(Q,\Sigma,\Delta,R,Q_d)$, where $Q=Q_d=\{p\}$, $\Delta_0=\{e\}$, $\Delta_1=\Sigma$ and $R$ consists of the following rules:
	\begin{enumerate}[label=(\roman*)]
		\item for each $a\in\Sigma_0$, $p[a]\to a[e]$ is in $R$;
		\item for every $k\geq1$, $a\in\Sigma_k$ and $1\leq i\leq k$, 
              $p[\listt[a][x_i]]\to a[p[x_i]]$ is in $R$.
	\end{enumerate}
	Consequently, $L$ is finite iff $M_p(L)$ is finite.
	Now, let $L\in\Surface$. Then, obviously, $M_p(L)\in \Surface$. Moreover $M_p(L)$ is monadic. Hence, by Lemma \ref{lem.4.89}, $M_p(L)$ is recognizable. Thus, by Theorem \ref{the.3.75}, it is decidable whether $M_p(L)$ is finite.
\end{proof}

To show that the above mentioned problems are solvable for \Target{}, we need the following lemma.
\begin{lemmawithoutqed}\label{lem.4.91}
	Each language in \Target{} is (effectively) of the form $\fyield(L)$ or $\fyield(L)\cup\{\lambda\}$, where $L\subseteq T_\Sigma$ for some $\Sigma$ such that $\Sigma_1=\emptyset$ and $e\notin\Sigma$, $L\in\Surface$.
\end{lemmawithoutqed}
\begin{proof}
	It is left as an \uem{exercise} to show that, for any $L'\subseteq T_{\Sigma'}$, there exists a bottom-up tree transducer $M$ such that $M(L')\subseteq T_\Sigma$ for some $\Sigma$ satisfying the requirements, and such that $\fyield(M(L'))=\fyield(L')-\{\lambda\}$.
	
	It is also left as an \uem{exercise} to show that it is decidable whether $\lambda\in\fyield(L')$. From these two facts the lemma follows.
\end{proof}

\begin{theorem}\label{the.4.92}
	The emptiness-, membership- and finiteness-problem are solvable for \Target.
\end{theorem}
\begin{proof}
	It is obvious from the previous lemma that we may restrict ourselves to target-languages $\fyield(L)$, where $L\in\Surface$ and $L\subseteq T_\Sigma$ for some $\Sigma$ such that $\Sigma_1=\emptyset$ and $e\notin\Sigma_0$.
	
	Obviously, $\fyield(L)=\emptyset$ iff $L=\emptyset$. 
Hence the emptiness-problem is solvable by Theorem~\ref{the.4.88}.
	
	Note that, by Example \ref{exa.2.17}, for a given $w\in\Sigma^*_0$, there are only a finite number of trees such that $\fyield(t)=w$. From this and Theorem \ref{the.4.88} the decidability of the membership-problem follows. Moreover it follows that $\fyield(L)$ is finite iff $L$ is finite. Hence, by Theorem \ref{the.4.90}, the finiteness-problem is solvable.
\end{proof}

We note that it can be shown that \Target{} is (properly, by Theorem \ref{the.4.92}) contained in the class of context-sensitive languages. Thus, \Target{} lies properly between the context-free and the context-sensitive languages.

We finally note that, in a certain sense, $\Surface\subseteq\Target$ (cf. the similar situation for \RECOG{} and \CFL). In fact, let $L\subseteq T_\Sigma$ be in \Surface. Let $\llbracket$ and $\rrbracket$ be two new symbols (``standing for $\lbrack$ and $\rbrack$").
Let $\Delta$ be the ranked alphabet such that $\Delta_0=\Sigma\cup\{\llbracket,\rrbracket\}$, $\Delta_1=\Delta_2=\Delta_3=\emptyset$ and $\Delta_{k+3}=\Sigma_k$ for $k\geq1$.
Let $M=(Q,\Sigma,\Delta,R,Q_d)$ be the (deterministic) top-down tree transducer such that $Q=Q_d=\{f\}$ and the rules are the following
\begin{enumerate}[label=(\roman*)]
	\item for $k\geq1$ and $a\in\Sigma_k$,
	      $f[\listt[a][x_i]]\to a[a\llbracket f[x_1]\cdots f[x_k]\rrbracket]$ is in $R$,
	\item for $a\in\Sigma_0$, $f[a]\to a$ is in $R$.
\end{enumerate}
(Note that $M$ is in fact a homomorphism). It is left to the reader to show that $\fyield(M(L))=L\langle\lbrack\,\gets\llbracket,\,\rbrack\gets\,\rrbracket\rangle$ (as string languages).

\section{Notes on the literature}\label{sec.5}
In the text there are some references to
\begingroup
\renewcommand{\section}[2]{}

\endgroup
An informal survey of the theory of tree automata and tree transducers (uptil 1970) is given by \cite{thatcher73}.
\subsection*{On Section \ref{sec.3}}
Bottom-up finite tree automata were invented around 1965 independently by [Doner, 1965, 1970] and \cite{thatcherwright68} (and somewhat later by \cite{pairquere68}). The original aim of the theory of tree automata was to apply it to the decision problems of second-order logical theories concerning strings. The connection with context-free languages was established in \cite{mezeiwright67} and \cite{thatcher67}, and the idea to give ``tree-oriented" proofs for results on strings is expressed in \cite{thatcher73} and \cite{rounds70a}. Independently, results concerning parenthesis languages and structural equivalence were obtained by \cite{mcnaugthon67}, \cite{knuth67}, \cite{ginsburgharrison67} and \cite{paullunger68}. Top-down finite tree automata were introduced by \cite{rabin69} and \cite{magidormoran69}, and regular tree grammars by \cite{brainerd69}. The notion of rule tree languages occurs in one form or another in several places in the literature.
\\

Most of the results of Section \ref{sec.3} can be found in the above mentioned papers.

Other work on finite tree automata and recognizable tree languages is written down in the following papers:
\begin{enumerate}[label=]
	\item \cite{arbibgiveon68}, automata on acyclic graphs, category theory;
	\item \cite{brainerd68}, state minimalization of finite tree automata;
	\item \cite{costich72}, closure properties of \RECOG{};
	\item \cite{eilenbergwright67}, category theoretic formulation of fta;
	\item \cite{itoando74}, axioms for regular tree expressions;
	\item \cite[1974]{maibaum72}, tree automata on ``many-sorted" trees;
	\item \cite{ricci73}, decomposition of fta;
	\item \cite[1973]{takahashi72}, several results;
	\item \cite{yeh71}, generalization of ``{}semigroup of fa" to fta.
\end{enumerate}

Remark: The notation of substitution (tree concatenation) can be formalized algebraically as in \cite{eilenbergwright67}, \cite{goguenthatcher74}, \cite{thatcher70} and \cite{yeh71}.
\\

Generalizations of finite automata are the following:
\begin{itemize}
	\item finite automata on derivation graphs of type 0 grammars, \\
\cite{benson70}, \cite{buttelmann71}, \cite{hart74};
	\item probabilistic tree automata, \\
\cite{ellis70}, \cite{magidormoran70};
	\item automata and grammars for infinite trees, \\
\cite{rabin69}, \cite{goguenthatcher74}, \cite{engelfriet72};
	\item recognition of subsets of an arbitrary algebra (rather than the algebra of trees), \cite{mezeiwright67}, \cite{shephard69}.
\end{itemize}

There are no real open problems in finite tree automata theory. It is only a question of (i) how far one wants to go in generalizing finite automata theory to trees (for instance, decomposition theory, theory of incomplete sequential machines, noncounting regular languages, etc), and (ii) which results on context-free languages one can prove via trees (for instance, Greibach normal form, Parikh's theorem, etc.).

\subsection*{On Section \ref{sec.4}}
For the literature on syntax-directed translation, see \cite{AU}. The notion of ``generalized syntax-directed translation" is defined in \cite{ahoullman71}.

The top-down tree transducer was introduced in \cite{rounds68}, \cite{thatcher70} and \cite{rounds70b} as a model for syntax-directed translation and transformational grammars. The bottom-up tree transducer was introduced in \cite{thatcher73}. The notion of a tree rewriting system occurs in one form or another in \cite{brainerd69}, \cite[1973]{rosen71}, \cite{engelfriet71} and \cite{maibaum74}.

Most results of Section \ref{sec.4} can be found in \cite{rounds70b}, \cite{thatcher70}, \cite{rosen71}, \cite{engelfriet71}, \cite{baker73} and \cite{ogdenrounds72}. Other papers on tree transformations are \cite{alagic73}, \cite{benson71}, \cite{bertsch73}, \cite{kosaraju73}, \cite{levyjoshi73}, \cite{martinvere70} and \cite{rounds73}.
\\

We mention the following problems concerning tree transformations.
\begin{itemize}
	\item ``Statements such as ``similar models have been studied by [x,y,z]'' are symptomatic of the disarray in the development of the theory of translation and semantics'' (free after \cite{thatcher73}).
	      
	      Establish the precise relationships between various models of syntax-directed translation, semantics of context-free languages and tree transducers (see \cite{goguenthatcher74}):
	      \begin{enumerate}[label=(\roman*)]
	      	\item Compare existing models of syntax-directed translation with top-down tree transducers, in particular with respect to the classes of translations they define (see Definition \ref{def.4.41}). See \cite{AU}, \cite{ahoullman71}, \cite{thatcher73} and \cite{martinvere70}.
	      	\item Define a tree transducer corresponding to the semantics definition method of \cite{knuth68}.
	      \end{enumerate}
	\item Develop a general theory of operations on tree languages (tree AFL theory). Relate these operations to the yield languages, as illustrated in Section \ref{sec.3}. This work was started by \cite{baker73}.
	      
	      It would perhaps be convenient to consider tree transducers $(Q,\Sigma,\Delta,R,Q_d)$ such that $R$ consists of rules $t_1\to t_2$ with $t_1\in Q[T_\Sigma(\X)]$, $t_2\in T_\Delta[Q(\X)]$ in the top-down case, or $t_1\in T_\Sigma(Q[\X])$, $t_2\in Q[T_\Delta(\X)]$ in the bottom-up case.
	\item Consider surface tree languages and target languages more carefully. Prove that the class \T-Target is incomparable with the class of indexed languages. Prove that the classes $\T^{\,k}$-Surface form a proper hierarchy (see \cite{ogdenrounds72} and \cite{baker73}).
	      Prove that the classes $\T^{\,k}$-Target form a proper hierarchy (then you also proved the previous one). Prove that $\DT\text{-Target}\subsetneq\T\text{-Target}$. Is it possible to obtain each target language by a sequence of nondeleting (and nonerasing) tree transducers? etc.
	\item Consider the complexity of target languages and translations (see \cite{ahoullman71}, \cite{baker73} and \cite{rounds73}).
	\item What is the practical use of tree transducers? (see \cite{AU}, \cite{deremer74}).
\end{itemize}

\subsection*{Other subjects}
We mention finally the following subjects in tree theory.
\begin{itemize}
	\item Context-free tree grammars, \cite{fischer68}, \cite[1970a, 1970b]{rounds69}, \cite{downey74}, \cite{maibaum74}.
	      
	      Let us explain briefly the notion of a context-free tree grammar. Consider a system $G=(N,\Sigma,R,S)$, where $N$ is a ranked alphabet of nonterminals, $\Sigma$ is a ranked alphabet of terminals, $V:=N\cup\Sigma$, $S\in N_0$ is the initial nonterminal, and $R$ is a finite set of rules of one of the forms
	      \begin{enumerate}[label=]
	      	\item $A\to t$ with $A\in N_0$ and $t\in T_V$, or
	      	\item $\listt[A][x_i]\to t$ with $A\in N_k$ and $t\in T_V(\X_k)$.
	      \end{enumerate}
	      $G$ is considered as a tree rewriting system on $T_V$.
	      If we let all variables in the rules range over $T_V$ then $G$ is called a context-free tree grammar. If we let all variables range over $T_\Sigma$, then $G$ is called a bottom-up (or inside-out) context-free tree grammar. The language generated by $G$ is $L(G)=\{t\in T_\Sigma\mid S\xRightarrow{*}t\}$. In general the languages generated by $G$ under the above two interpretations differ (consider for instance the grammar $S\to F[A] $, $ F[x_1] \to f[x_1x_1] $, $A\to a$, $A\to b$).
	      Thus restriction to bottom-up ( $\approx$ right-most in the string case) generation gives another language. On the other hand, restriction to top-down ( $\approx$ left-most) generation can be done without changing the language. The yield of the class of context-free tree languages is equal to the class of indexed languages. The yield of the class of bottom-up context-free tree languages is called \IO{}. These two classes are incomparable (\cite{fischer68}).
	      
	      How do these classes compare with the target languages? Is it possible to iterate the $\CFL\to\RECOG$ procedure and obtain results about (bottom-up) context-free tree languages from regular tree languages (of a ``higher order") (see \cite{maibaum74}). Is there any sense in considering pushdown tree automata?
	\item General computability on trees: \cite{rus67}, \cite{mahn69}, the Vienna method.
	\item Tree walking automata (at each moment of time the finite automaton is at one node of the tree; depending on its state and the label of the node it goes to the father node or to one of the son nodes): \cite{ahoullman71}, \cite{martinvere70}.
	\item Tree adjunct grammars (another, linguistically motivated, way of generating tree languages): \cite{joshilevytakahashi73}.
	\item Lindenmayer tree systems (parallel rewriting): \\
\cite{culikii74}, \cite{culikiimaibaum74}, \cite{engelfriet74}, \cite{szilard74}.
\end{itemize}

\makeatletter
\renewcommand\@biblabel[1]{}
\newcommand{\margin}{\hspace*{1cm}}
\makeatother

	\end{sloppypar}
\end{document}